\documentclass[11pt]{article}
\usepackage[margin=.8in,left=.8in]{geometry}
\usepackage{amsmath}
\usepackage{amsfonts}
\usepackage{amssymb}
\usepackage{amsthm}
\usepackage{mathtools}
\usepackage{subcaption}
\usepackage{float}
\usepackage{graphicx}
\usepackage{color}
\usepackage{xcolor}
\usepackage{xfrac}
\usepackage{stmaryrd}
\usepackage[utf8]{inputenc}
\usepackage{hyperref}
\usepackage[all]{xy}

\usepackage{tikz}
\usepackage{tikz-cd}
\usepackage{lipsum}
\usepackage{adjustbox}

\usepackage{multirow}

\usepackage{stmaryrd}    % for \sslash

\usepackage{enumerate} %Roman enumerate

\usepackage{helvet}   %vertical spacing in tabular

\usepackage{color, colortbl}
\definecolor{Gray}{gray}{0.9}
\definecolor{lightgray}{rgb}{0.9,0.9,0.9}

\definecolor{darkblue}{rgb}{0.05,0.25,0.65}
\definecolor{greenii}{RGB}{20,140,10}

\definecolor{orangeii}{RGB}{200,100,5}

\usepackage{mathptmx}
\usepackage{amsmath}

\usepackage{floatflt}  %Floating tables
\usepackage{wrapfig} %floating figure
\usepackage{array}     % specifying size of tables
\newcolumntype{L}[1]{>{\raggedright\let\newline\\\arraybackslash\hspace{0pt}}m{#1}}
\newcolumntype{C}[1]{>{\centering\let\newline\\\arraybackslash\hspace{0pt}}m{#1}}
\newcolumntype{R}[1]{>{\raggedleft\let\newline\\\arraybackslash\hspace{0pt}}m{#1}}

\usepackage{diagbox}  %diagonal in table

\makeatletter
\newcommand{\raisemath}[1]{\mathpalette{\raisem@th{#1}}}
\newcommand{\raisem@th}[3]{\raisebox{#1}{$#2#3$}}
\makeatother

\usepackage{tabularx}   %For sizing tabular entries

\usepackage[new]{old-arrows}   %For \longhookrightarrow

\newdir{> }{{}*!/10pt/@{>}}

\newcommand{\dslash}{/\!\!/}

\DeclareRobustCommand{\rchi}{{\mathpalette\irchi\relax}}
\newcommand{\irchi}[2]{\raisebox{\depth}{$#1\chi$}} % inner command, used by \rchi

\makeatletter
\newif\if@sup
\newtoks\@sups
\def\append@sup#1{\edef\act{\noexpand\@sups={\the\@sups #1}}\act}%
\def\reset@sup{\@supfalse\@sups={}}%
\def\mk@scripts#1#2{\if #2/ \if@sup ^{\the\@sups}\fi \else%
  \ifx #1_ \if@sup ^{\the\@sups}\reset@sup \fi {}_{#2}%
  \else \append@sup#2 \@suptrue \fi%
  \expandafter\mk@scripts\fi}
\def\tensor#1#2{\reset@sup#1\mk@scripts#2_/}
\def\multiscripts#1#2#3{\reset@sup{}\mk@scripts#1_/#2%
  \reset@sup\mk@scripts#3_/}
\makeatother

\makeatletter
\newbox\slashbox \setbox\slashbox=\hbox{$/$}
\def\itex@pslash#1{\setbox\@tempboxa=\hbox{$#1$}
  \@tempdima=0.5\wd\slashbox \advance\@tempdima 0.5\wd\@tempboxa
  \copy\slashbox \kern-\@tempdima \box\@tempboxa}
\def\slash{\protect\itex@pslash}
\makeatother

\def\clap#1{\hbox to 0pt{\hss#1\hss}}

\def\mathrlap{\mathpalette\mathrlapinternal}

\def\mathrlapinternal#1#2{\rlap{$\mathsurround=0pt#1{#2}$}}

\let\oldroot\root
\def\root#1#2{\oldroot #1 \of{#2}}
\renewcommand{\sqrt}[2][]{\oldroot #1 \of{#2}}

\DeclareSymbolFont{symbolsC}{U}{txsyc}{m}{n}
\SetSymbolFont{symbolsC}{bold}{U}{txsyc}{bx}{n}
\DeclareFontSubstitution{U}{txsyc}{m}{n}

\DeclareSymbolFont{stmry}{U}{stmry}{m}{n}
\SetSymbolFont{stmry}{bold}{U}{stmry}{b}{n}

\DeclareFontFamily{OMX}{MnSymbolE}{}
\DeclareSymbolFont{mnomx}{OMX}{MnSymbolE}{m}{n}
\SetSymbolFont{mnomx}{bold}{OMX}{MnSymbolE}{b}{n}
\DeclareFontShape{OMX}{MnSymbolE}{m}{n}{
    <-6>  MnSymbolE5
   <6-7>  MnSymbolE6
   <7-8>  MnSymbolE7
   <8-9>  MnSymbolE8
   <9-10> MnSymbolE9
  <10-12> MnSymbolE10
  <12->   MnSymbolE12}{}

\makeatletter
\def\Decl@Mn@Delim#1#2#3#4{%
  \if\relax\noexpand#1%
    \let#1\undefined
  \fi
  \DeclareMathDelimiter{#1}{#2}{#3}{#4}{#3}{#4}}
\def\Decl@Mn@Open#1#2#3{\Decl@Mn@Delim{#1}{\mathopen}{#2}{#3}}
\def\Decl@Mn@Close#1#2#3{\Decl@Mn@Delim{#1}{\mathclose}{#2}{#3}}
\Decl@Mn@Open{\llangle}{mnomx}{'164}
\Decl@Mn@Close{\rrangle}{mnomx}{'171}
\Decl@Mn@Open{\lmoustache}{mnomx}{'245}
\Decl@Mn@Close{\rmoustache}{mnomx}{'244}
\makeatother

\makeatletter
\DeclareRobustCommand\widecheck[1]{{\mathpalette\@widecheck{#1}}}
\def\@widecheck#1#2{%
    \setbox\z@\hbox{\m@th$#1#2$}%
    \setbox\tw@\hbox{\m@th$#1%
       \widehat{%
          \vrule\@width\z@\@height\ht\z@
          \vrule\@height\z@\@width\wd\z@}$}%
    \dp\tw@-\ht\z@
    \@tempdima\ht\z@ \advance\@tempdima2\ht\tw@ \divide\@tempdima\thr@@
    \setbox\tw@\hbox{%
       \raise\@tempdima\hbox{\scalebox{1}[-1]{\lower\@tempdima\box
\tw@}}}%
    {\ooalign{\box\tw@ \cr \box\z@}}}
\makeatother

\makeatletter
\def\udots{\mathinner{\mkern2mu\raise\p@\hbox{.}
\mkern2mu\raise4\p@\hbox{.}\mkern1mu
\raise7\p@\vbox{\kern7\p@\hbox{.}}\mkern1mu}}
\makeatother

\newcommand{\mathfr}{\mathfrak}

\newcommand{\h}{\mathfrak{h}}

\newcommand{\R}{\ensuremath{\mathbb R}}
\newcommand{\Z}{\ensuremath{\mathbb Z}}
\newcommand{\Q}{\ensuremath{\mathbb Q}}

\newcommand{\cE}{\ensuremath{\mathcal E}}
\newcommand{\cH}{\ensuremath{\mathcal H}}
\newcommand{\cK}{\ensuremath{\mathcal K}}
\newcommand{\cD}{\ensuremath{\mathcal D}}
\newcommand{\cC}{\ensuremath{\mathcal C}}

\renewcommand{\(}{\begin{equation}}
\renewcommand{\)}{\end{equation}}
\newcommand{\bea}{\begin{eqnarray*}}
\newcommand{\eea}{\end{eqnarray*}}

\newcommand{\mc}[1]{\mathcal{#1}}
\newcommand{\abs}[1]{\lvert #1 \rvert}
\newcommand{\on}[1]{\operatorname{#1}}

\def\BB{\mathbb{B}}
\def\CC{\mathbb{C}}
\def\GG{\mathbb{G}}
\def\HH{\mathbb{H}}
\def\NN{\mathbb{N}}
\def\PP{\mathbb{P}}
\def\QQ{\mathbb{Q}}
\def\RR{\mathbb{R}}
\def\ZZ{\mathbb{Z}}

\def\PPP{\mathbf{P}}

\DeclareMathOperator{\Aut}{Aut}

\DeclareMathOperator{\Der}{Der}
\DeclareMathOperator{\dR}{dR}
\DeclareMathOperator{\End}{End}
\DeclareMathOperator{\GL}{GL}
\DeclareMathOperator{\gr}{gr}
\DeclareMathOperator{\Hom}{Hom}
\DeclareMathOperator{\id}{id}

\DeclareMathOperator{\Lie}{Lie}
\DeclareMathOperator{\Map}{Map}
\DeclareMathOperator{\Mor}{Mor}
\DeclareMathOperator{\Pic}{Pic}

\DeclareMathOperator{\Spec}{Spec}

\usepackage{cleveref}

\crefformat{section}{\S#2#1#3} % see manual of cleveref, section 8.2.1
\crefformat{subsection}{\S#2#1#3}
\crefformat{subsubsection}{\S#2#1#3}

\newtheorem{theorem}{Theorem}[section]
\newtheorem{lemma}[theorem]{Lemma}
\newtheorem{prop}[theorem]{Proposition}
\newtheorem{cor}[theorem]{Corollary}

\newtheorem{conj}[theorem]{Conjecture}
\theoremstyle{definition}
\newtheorem{defn}[theorem]{Definition}

\newtheorem{example}[theorem]{Example}

\newtheorem{remark}[theorem]{Remark}
\newtheorem{note[theorem]}{Note}

\usepackage{amsfonts}

\begin{document}

\title{Mysterious Triality and Rational Homotopy Theory}

\author{Hisham Sati, \; Alexander A. Voronov}

\maketitle
  
%\epigraph{\textit{Le roi est mort, vive le roi!}}

\vspace{-3mm}
$\,$ \hspace{7.5cm} {\it To our teachers: Igor V. Dolgachev and Yuri I. Manin}
\bigskip 

\begin{abstract}
Mysterious Duality has been discovered by Iqbal, Neitzke, and Vafa
\cite{INV} as a convincing, yet mysterious correspondence between
certain symmetry patterns in toroidal compactifications of M-theory
and del Pezzo surfaces, both governed by the root system 
series $E_k$. 

It turns out that the sequence of del Pezzo surfaces is not the only sequence of objects in mathematics that gives rise to the same $E_k$ symmetry pattern. We present a sequence of topological spaces, starting with the four-sphere $S^4$, and then forming its iterated cyclic loop spaces $\mathcal{L}_c^k S^4$, within which we 
discover the $E_k$ symmetry pattern  via rational homotopy theory. For this sequence of spaces, the correspondence between its $E_k$ symmetry pattern and that of toroidal compactifications of M-theory is no longer a mystery, as each space $\mc{L}_c^k S^4$
is naturally related to the compactification of
M-theory on the $k$-torus via identification of the equations of motion of $(11-k)$-dimensional supergravity as the defining equations of the Sullivan minimal model of $\mc{L}_c^k S^4$. This gives an explicit duality between algebraic topology and physics.

Thereby, we
extend Iqbal-Neitzke-Vafa's Mysterious Duality between algebraic geometry and physics into 
a triality, also involving algebraic topology. 
Via this triality, duality between physics and mathematics is
demystified, and the mystery is transferred to the mathematical realm as duality
between algebraic geometry and algebraic topology.
Now the question is: Is there an explicit relation between the del Pezzo
surfaces $\mathbb{B}_k$ and iterated cyclic loop spaces of $S^4$ which would explain the common $E_k$ symmetry pattern?
 \end{abstract}

\medskip

\tableofcontents

%\vfill

%%%%%%%%%%%%%%%
\section{Introduction}
%%%%%%%%%%%%%%
\label{intro}

Mysterious Duality has been discovered by Iqbal, Neitzke, and Vafa
\cite{INV} as a remarkable, yet mysterious correspondence between
certain symmetry patterns in toroidal compactifications of M-theory
and del Pezzo surfaces, both governed by the root system 
corresponding to the exceptional series $E_k$, $k\leq 8$.

\paragraph{Del Pezzo surfaces.}
Consider the \emph{del Pezzo surface} $\mathbb{B}_k$ obtained as the blowup of 
the complex projective plane $\mathbb{C P}^2$ at $k$ generic points 
$x_1, \dots, x_k$, $0\leq k \leq 8$. 
The Picard group ${\rm Pic}(\mathbb{B}_k)$ of isomorphism classes of line bundles is in this case isomorphic to 
the divisor class group and the second cohomology group: 
$\Pic (\BB_k) \cong H^2(\mathbb{B}_k;\ZZ)$. This is a rank-$(k+1)$ lattice
with a natural Lorentzian inner product
% is a free abelian group
given by the intersection form.
Another important feature of the del Pezzo surface $\BB_k$ is
the \emph{anticanonical class}:
$-\cK_k := - \Omega^2_{\BB_k}$,
which is ample and defines a map $\BB_k \to \CC \PP^{9-k}$, also called \emph{anticanonical}. This map is an embedding for $k \le 6$. The \emph{degree of the del Pezzo surface} is the self-intersection number $(-\cK_k) \cdot (-\cK_k) = \cK_k \cdot \cK_k = 9 - k$.

\smallskip 
There is also an ``outlier'' del Pezzo surface $\BB_1' := \CC \PP^1 \times \CC \PP^1$ 
of degree 8 
with Picard group of rank 2.
This surface is related to the del Pezzo surfaces $\BB_k$ by a single blowup: if we blow up a point in $\CC \PP^1 \times \CC \PP^1$, we will obtain a surface isomorphic to $\BB_2$. 

\smallskip 
The connection between algebraic geometry and Lie theory comes from the fact that the Cartan matrices of the exceptional Lie algebras of type $E_k$ and their
root systems arise from the above data: a lattice with a distinguished element and inner product
(see \cite{Man}).

\smallskip 
In general, even for a fixed $k$, the surfaces $\BB_k$ obtained from varying the blowup points are not isomorphic
as complex manifolds. However, they are diffeomorphic, so these surfaces give rise to the same combinatorial data, 
and we will just speak of ``the'' del Pezzo surface $\BB_k$ for each $k$.

\paragraph{The Mysterious Duality correspondence.} 
There is a correspondence between del Pezzo surfaces
$\mathbb{B}_k$ and M-theory ``wrapped'' on tori $T^k$, \cite{INV}, as follows. Given an element $\omega$ of
$H^2(\mathbb{B}_k; \R)$, considered as a generalized K\"{a}hler form on $\BB_k$, the generalized volumes $
%\displaystyle{
\omega(\cC) := \omega \cdot \cC = \int_{\BB_k} \omega \cup \cC$
of the standard basic classes $ \cC = \cH, \cE_1, \dots, \cE_k$, see \Cref{dPtriple}, may be thought of as logarithms, up to certain constants, of the coordinates $(\ell_p, R_1, \linebreak[0] \dots, \linebreak[1] R_k)$ on the moduli space of M-theory compactified on the flat $k$-torus 
$T^k$:
$$
\omega(\cH)=-3 \ln \ell_p; \qquad \omega(\cE_i)=- \ln (2 \pi R_i)\;, \quad  i=1, \dots, k,
$$
where $\ell_p$ is the Planck scale and the $R_i$'s are the radii of the torus factors.
The moduli space of M-theory compactified on $T^k$ is usually taken to be the double quotient $K \backslash G / G(\ZZ)$, where $G$ is the U-duality group, which is the (real split form of the) Lie group $E_k$, $G(\ZZ)$ is its integral form, and $K$ is the maximal compact subgroup of $G$
\cite{HT}\cite{OP}. In  \cite{INV} a simpler moduli space $
\mathcal{M}_k := A/W
$
is used. Taking into account the Iwasawa decomposition $G=KAN$ with $A$ the $\RR$-split abelian factor and $N$ the unipotent factor and the identification of the Weyl group as $W = \mc{N}(A)/A$, where $\mc{N}(A)$ is the normalizer of $A$, one may think of passing from $K \backslash G /  G(\ZZ)$ to $
A/W
$
as some sort of abelianization:
\(
\label{abelianization} 
K \backslash G / G(\ZZ) = (AN) /G(\ZZ) \xymatrix{\ar@{~>}[rr]^{\rm \bf \color{darkblue} abelianize}&&}  
\mathcal{M}_k = A/W\;.
\)  

Further, the Mysterious Duality correspondence has the following features \cite{INV}:

\vspace{-2mm} 
\begin{enumerate}[{\bf (i)}]
\setlength\itemsep{-3pt}
\item Automorphisms of $\mathbb{B}_k$ and $H^2 (\BB_k; \ZZ)$
correspond to the U-duality
transformations of M-theory on $T^k$.
The U-duality group of M-theory on $T^k$, 
which for rectangular compactifications
with no $C$-field is given by the Weyl group of $E_k$,  
is related to a subgroup of the automorphism group of 
$H^2(\mathbb{B}_k; \ZZ)$ which preserves the intersection form and canonical class.

\item The moduli space $A/W$ of compactified M-theory corresponds to the moduli space $H^2 (\BB_k; \RR)/W$ of generalized K\"{a}hler forms $\omega$, regarded as metric/transcendental data on $\BB_k$, up to automorphisms of $\BB_k$ considered as an algebraic surface.
 
\item Two classes of rational curves $\cC_1$ and $\cC_2$ related as $\cC_1 + \cC_2= - \cK_k$ on $\BB_k$
correspond to two ${\rm D}p_1$-branes 
with ${\rm D}p_2$-branes  with $p_1 + p_2=7-k$, expressing electric-magnetic duality.
% relating ${\rm D}p_1$-branes 
% with ${\rm D}p_2$-branes  with $p_1 + p_2=6$

\item The $p$-branes of type IIA and type IIB string theory in 10 dimensions correspond to
classes of rational curves on $\mathbb{B}_1$ and $\CC \PP^1 \times \CC \PP^1$, respectively. 

\end{enumerate}

\vspace{-3mm} 
\paragraph{Our conceptual take.} We
bring in algebraic topology,
in the form of (rational) homotopy theory,
and propose a triality of the form:
\vspace{-7mm} 
\begin{figure}[htb]
% \centering
$$
\xymatrix@R=.1em{
 \fbox{\text{Algebraic Geometry}} \;\; \ar@{<..>}[rr]^-{\footnotesize
 \raisebox{.5pt}{\textcircled{\raisebox{-.9pt} {3}}}}
 \ar@{<..>}[dr]_-{\text{\footnotesize \raisebox{.5pt}{\textcircled{\raisebox{-.9pt} {1}}}}}&  & \;\; 
 \fbox{\text{Algebraic Topology}}
\\
& \fbox{\text{Physics}} \ar@{<->}[ur]_-{\text{\footnotesize
\raisebox{.5pt}{\textcircled{\raisebox{-.9pt} {2}}}}}& 
}
$$

\vspace{-5mm} 
\caption{\footnotesize {\it Mysterious Triality}. The solid arrow 
{\footnotesize \raisebox{.5pt}{\textcircled{\raisebox{-.9pt} {2}}}}
stands for a relationship between physics and algebraic topology which is made explicit, while the dotted arrows 
{\footnotesize \raisebox{.5pt}{\textcircled{\raisebox{-.9pt} {1}}}}
and
{\footnotesize \raisebox{.5pt}{\textcircled{\raisebox{-.9pt} {3}}}}
denote still mysterious correspondences 
based on combinatorial coincidences. It is enough to solidify one of the dotted arrows to remove the mystery veil off the face of Triality.}
\label{MT}
\end{figure}
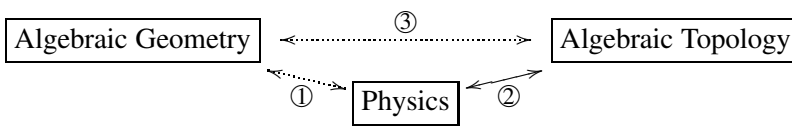

The remarkable fact is that the duality pictured by arrow {\footnotesize \raisebox{.5pt}{\textcircled{\raisebox{-.9pt} {2}}}} is explicit. Let us explain the main idea behind it.
Rational homotopy theory associates to a topological space a certain algebra in two ways:
the Quillen model, which is a differential graded (dg) Lie algebra, and the Sullivan model, which is 
a dg-commutative algebra. 
Roughly speaking, the Sullivan model is a rational homotopy model based on cohomology, 
while the Quillen model is 
% a rational homotopy model
based on 
homotopy, so they are dual in some sense; see \cite{Ta83}\cite{Maj}\cite{felix-halperin-thomas}.

\smallskip 
In our context, the Sullivan minimal model of the 4-sphere $S^4$ captures the dynamics of 
the fields in M-theory, as proposed in \cite{Sati13}, and developed further under
the name {\it Hypothesis H} in \cite{FSS17}\cite{FSS19b}\cite{FSS-WZW}\cite{GS21} (and applied in \cite{Roberts}).
% ;
% we review this below in \cref{smm}.
% This will be the space of form fields in M-theory. 
Particularly, the generators of the Sullivan minimal model
\[
M(S^4) = \RR\big[g_4, g_7 \; | \; dg_4 = 0, dg_7 = - \tfrac{1}{2} g_4^2\big]
\]

\vspace{-1mm} 
\noindent
correspond to the basic supergravity fields $G_4$ and $G_7 = *G_4$, expressed as differential forms on 11d spacetime, whereas the differential of $M(S^4)$ corresponds to the equations of motion (EOMs):
\vspace{-1mm} 
\begin{equation}
\label{EOM}
d G_4 = 0, \qquad d G_7 = - \tfrac{1}{2}G_4 \wedge G_4 = 0 \;.
\end{equation}
Moreover, the algebra of gauge transformations for these fields is captured 
by the Quillen model (see \Cref{Ex-Malgebra}). A gentle introduction to the relationship between 
M-theory and rational homotopy theory
is given in \cite{FSS19a}.

\smallskip 
Furthermore, we find a striking match between the equations
  for the Sullivan minimal model of the iterated cyclic loop space (\emph{cyclification}) ${\mc{L}_c^k S^4}$  of the
  four-sphere and the EOMs of $k$-fold circle reduction
% toroidal compactifications
of M-theory;
% to $11-k$ dimensions;
see \Cref{Sec-S4cyc}
% and \Cref{Sec-CyclicSugra}
(specifically, \Cref{smm}, Examples \ref{example1} and \ref{example2}, and \Cref{IIB}). This is deeply rooted in the fact that writing out the EOMs for the reduction of supergravity on a circle $S^1$ is akin to the process
  of working out the Sullivan minimal model of the cyclic loop space;
%  as an extension in the sense of Halperin \cite{Halperin83};
see \eqref{ext:cyc} and  \eqref{dc}, cf.\ adjunction \eqref{adjunction}.

\smallskip
We may formulate the above more precisely as the main mathematical physics result of the paper.

\begin{theorem}[Supergravity dynamics from rational homotopy theory]
The Sullivan minimal model $M({\cal L}_c^k S^4)$ determines the duality-symmetric equations of motion of $(11-k)$-di\-men\-sion\-al supergravity descending from 11-dimensional supergravity \cref{Sullmm}, with the 
 gauge algebra structure determined by the corresponding 
Quillen model (Examples \ref{Ex-Malgebra},  \ref{Ex-reducedAlg}). There is also a minimal Sullivan algebra corresponding to type IIB, with
T-duality between type II theories  (Prop.\ \ref{Prop-Tdual}) observed at the level of Sullivan minimal models.
\end{theorem}

This matching also implies the following general philosophy (extending the $k=1$ case in \cite{FSS-pbranes}\cite{FSS-L00}):
\begin{quote}
{\it Any feature of or statement about the Sullivan minimal model of an iterated cyclic
loop space $\mc{L}_c^k S^4$ (or the rational homotopy type thereof)
may be translated into a feature of or statement about the
compactification of M-theory on the $k$-torus.}
\end{quote}

Examples of such derive from the mathematical core of the paper: we find that toroidal symmetries of the rational homotopy type of $\mc{L}_c^k S^4$  lead naturally to the root system $E_k$, for each $k \ge 0$. Combined with the theorem above, this explains the appearance of the $E_k$ root system in toroidal compactifications of M-theory and the connection of the Weyl group
$W(E_k)$ to U-duality. The split real torus action on $M(\mc{L}_c^k S^4)$ translates into 
trombone and rescaling symmetries of $(11-k)$-dimensional supergravity (see \cref{symm-cycl-concrete}), the 27 exceptional vectors in the root system $E_6$ translate into a collection of 27 distinguished fields in 5d spacetime (see \cref{Sec-27}). 
The luxury of the principle above was not available
within the duality between del Pezzo surfaces and torodial
compactifications of M-theory, as there was only a collection of
surprising coincidences, which were not based on an explicit
relation. The very lack of an explicit relation is the essence of
the mystery behind the duality proposed by \cite{INV}. This is the dotted arrow {\footnotesize \raisebox{.5pt}{\textcircled{\raisebox{-.9pt} {1}}}} in \Cref{MT}.

\smallskip 
The associated series of mathematical results is collected in the following metatheorem.
 \begin{theorem}[Metatheorem]
% \vspace{-2mm} 
 \begin{enumerate}[{\bf (a)}]
 \setlength\itemsep{-2pt}
 \item   
The maximal $\RR$-split torus of the real algebraic group $\Aut M (\mc{L}_c^k S^4)$ for $k \ge 0$ is a $(k+1)$-dimensional torus
$T^{k+1}$ canonically isomorphic to $\GG_m^{k+1}$ over $\RR$ (Cor.\ \ref{split-rank}).

\item 
%  Toroidal symmetries from $S^4$
The action of the maximal split torus $\GG_m^{k+1}$ on $M(\mc{L}^k_c S^4)$ may be lifted to an action on the space $\mc{L}_c^k S^4$ in the rational homotopy category. In this way, the last factor $\GG_m$ of $\GG_m^{k+1}$ acts via self-maps of the target $S^4$ and the first $k$ factors act via self-maps of the source $S^1$s (Props.\ \ref{prop-torS4} and \ref{prop-torS1}).

\item For type IIB, a maximal split torus $T^B = \GG_m^2$ is identified explicitly
(Prop. \ref{prop-torIIB}).

\item
\label{bases}
    The $(k+1)$-dimensional real abelian Lie algebra $\h_k = \Lie
    (T^{k+1}) \subseteq \Der M(\mc{L}_c^k S^4)$, which plays the role of a Cartan subalgebra, of the maximal
    $\RR$-split torus $T = T^{k+1}$ of the algebraic group $\Aut
    M(\mc{L}_c^k S^4)$ has an explicit canonical basis. So does the linear dual $\h_k^*$, which plays the role of a weight space
(Thm. \ref{thm-bases}).

\item
\label{element}
%Degree in the $k$-fold cyclic loop space of $S^4$:
There is a unique element of the Lie algebra $\h_k$, 
which acts on the Quillen minimal model $ Q(\mc{L}_c^k S^4)$ as the
degree operator
(Thm. \ref{thm-deg}).

\item For each $k$, $0 \le k \le 8$, the above bases {\rm (\ref{bases})} and element {\rm (\ref{element})} give rise to the exceptional root data $E_k$. This data, extracted from cyclic loop spaces 
$M(\mc{L}^{k}_cS^4)$,
replicates the root data
determined by del Pezzo surfaces $\BB_k$
(Thm. \ref{rootdata}). 
  The construction of the root data 
 % $\big(\h_k^*, \{\epsilon_0, \dots, \epsilon_k\}, (-,-), K_k^*\big)$
 of Theorem \ref{rootdata} extends
 % from $k = 8$
 to $k\geq 9$
  (Remark \ref{k-le-8}).

\item For type IIB, the exceptional root data from the rational homotopy model for type IIB 
replicates the data  determined by the del Pezzo surface 
$\CC \PP^1 \times \CC \PP^1$
and produces the
% corresponding
root system $A_1$
(Prop. \ref{Root-IIB}).

\item \textbf{$27$ Lines via rational homotopy of $6$-fold cyclic loop space}:
In the weight decomposition
\[
\pi_\bullet^\RR (\mc{L}_c^6 S^4) = \bigoplus_{\alpha  \in \PPP(\h_6)} \pi_\bullet^\RR (\mc{L}_c^6 S^4)_\alpha
\]

\vspace{-3mm} 
\noindent
corresponding to the $7$-torus action on the Quillen minimal model $Q(\mc{L}_c^6 S^4) = \pi_\bullet^\RR (\mc{L}_c^6 S^4)[1]$,
the
$27$ exceptional vectors $\alpha_i \in \PPP(\h_6)$, $i = 1, \dots, 27$,
% , $(\alpha, \alpha) = (\alpha, K_6) = -1$,
single out precisely the second real homotopy group $\pi_2^\RR (\mc{L}_c^6 S^4)$:
\[
\pi_2^\RR (\mc{L}_c^6 S^4) = \bigoplus_{i=1}^{27} \pi_\bullet^\RR (\mc{L}_c^6 S^4)_{\alpha_i} \; .
\]
Moreover,
\[
\dim \pi_\bullet^\RR (\mc{L}_c^6 S^4)_{\alpha_i} = 1
\]
for each $i = 1, \dots, 27$, which means there are $27$ canonically defined, linearly independent lines in the $\RR$-vector space $\pi_2^\RR (\mc{L}_c^6 S^4)$ and $\dim \pi_2^\RR (\mc{L}_c^6 S^4) = 27$ 
(Thm. \ref{27lines}).

\end{enumerate} 

\end{theorem}

\bigskip 
\begin{sloppypar}
Thus, our work provides an explicit, conceptual correspondence {\footnotesize \raisebox{.5pt}{\textcircled{\raisebox{-.9pt} {2}}}} between physics and algebraic topology
in the Triality above and thereby uncovers the mystery of Mysterious Duality, if understood in a
broad sense as a duality between physics and mathematics.
The other two sides of the Triality, see \Cref{MT}, still remain a
mystery. 
Filling out either of the mysterious sides of the triangle would complete the
story and resolve the Mysterious Duality conjecture of \cite{INV}.

\medskip 
This leads to a new, conjectural
duality within mathematics, a duality between the algebraic geometry
of del Pezzo surfaces and the algebraic topology of cyclic loop spaces
of the 4-sphere, formulated as: 
\end{sloppypar}

\begin{conj}
\label{Conj1}
There must be an explicit relation between the series of del Pezzo
surfaces $\BB_k$, $0 \le k \le 8$, and $\CC \PP^1 \times \CC \PP^1$ on the one hand and the series of iterated loop
spaces $\mc{L}^{k}_c S^4$, $0 \le k \le 8$, and the topological model $IIB$, see \cref{IIB}, on
the other hand. In particular, blowing up a del Pezzo surface should correspond to taking a cyclification of an iterated cyclification of $S^4$. This relation should match
the $E_k$ symmetry patterns occurring in both series, as well as
relate other geometric data, such as relate the volumes of curves on del
Pezzo surfaces with certain metric data on the iterated loop
spaces.
\end{conj}

\paragraph{The $E_k$ symmetry patterns.} 
The dimensional reduction of M-theory on 
a $k$-torus gives rise to a theory in $D = 11 - k$ dimensions with the 
scalar fields with  symmetry pattern \cite{INV} in the five columns in Table 
\ref{table1} below, matching the familiar pattern for del Pezzo surfaces \cite{Man}, to which we add the 6th column for cyclic loop spaces (``cyclifications'' $\mc{L}_c^k S^4$ of $S^4$),
as well as the 7th column corresponding torus symmetry, both of which we discover in this paper. This highlights the interrelations among Lie theory (4th column), 
algebraic geometry (5th column), and topology/physics (6th column),
as appropriate by the trichotomy/triality in Figure \ref{MT}.

\begin{table}[H]
%\begin{center}
\centering
\begin{tabular}{ccccccc}
\hline
$D$ &  $k$ & {\bf Type of $E_k$} & {\bf Lie algebra $\mathfr{g}$} & {\bf del Pezzo} & 
{\bf Model} & {\bf Maximal Split Torus}  \\
\hline
 \hline
\rowcolor{lightgray} 11 & 0 & $A_{-1}$ & $\mathfr{sl}_0 = \varnothing$ & $\CC \PP^2$
&  $S^4$       & $\GG_m$  \\
10 & 1 & $A_0$ & $\mathfr{sl}_1 = 0$ & $\mathbb{B}_1$ 
&  $\mathcal{L}_c S^4$ & $\GG_m \times \GG_m$ \\
10 & 1 & $A_1$ & $\mathfr{sl}_2$ & $\CC \PP^1 \times \CC \PP^1$ 
& $IIB$ & $\GG_m \times \GG_m$ \\
\rowcolor{lightgray} 9 & 2 & $A_{1}$ &$\mathfr{sl}_2$ & $\mathbb{B}_2$ 
&  $\mathcal{L}_c^2 S^4$ & $\GG_m^2 \times \GG_m$  \\
8 & 3 & $A_2 \times A_1$ &$\mathfr{sl}_3 \oplus \mathfr{sl}_2$ & $\mathbb{B}_3$ 
&  $\mathcal{L}_c^3 S^4$ & $\GG_m^3 \times \GG_m$  \\
\rowcolor{lightgray} 7 & 4 & $A_4$&$\mathfr{sl}_5$ & $\mathbb{B}_4$ 
&  $\mathcal{L}_c^4 S^4$ & $\GG_m^4 \times \GG_m$ \\
6 & 5 & $D_5$&$\mathfr{so}_{10}$  & $\mathbb{B}_5$ 
&  $\mathcal{L}_c^5 S^4$ & $\GG_m^5 \times \GG_m$  \\
\rowcolor{lightgray} 5 & 6 & $E_6$&$\mathfr{e}_6$   & $\mathbb{B}_6$ 
&  $\mathcal{L}_c^6 S^4$ & $\GG_m^6 \times \GG_m$  \\
4 & 7 & $E_7$&$\mathfr{e}_7$  & $\mathbb{B}_7$ 
&  $\mathcal{L}_c^7 S^4$ & $\GG_m^7 \times \GG_m$ \\
\rowcolor{lightgray} 3 & 8 & $E_8$&$\mathfr{e}_8$  & $\mathbb{B}_8$ 
&  $\mathcal{L}_c^8 S^4$ & $\GG_m^8 \times \GG_m$  \\
\hline
\end{tabular}
%\end{center}
\vspace{-2mm} 
\caption{\label{table1} \footnotesize The $E_k$ pattern in Lie theory,
$(0 \leq k \leq 8)$ 
%and Kac-Moody theory $(k\geq 9)$, 
del Pezzo surfaces, and cyclifications of $S^4$.}
\end{table}

\vspace{-2mm} 
Our formulation allows us to extend to higher ranks, $k=9$, $10$, and $11$,
corresponding to the infinite-di\-men\-sion\-al cases. 
Indeed, since our discussion extends beyond the Lie setting to the  Kac-Moody 
setting, we 
% will
have extensions of the Triality in Figure \ref{MT} and of \Cref{Conj1} that go beyond 
the Fano case on the algebraic side in \cref{Sec-KM}.
We also observe that cyclic loop spaces, when we go beyond $k=8$, undergo a transition analogous to that on the del Pezzo/root systems/Lie algebra side: the degree of the cyclification $\mathcal{L}_c^k S^4$ in the sense of \eqref{deg-L} ceases to be positive, the corresponding root system becomes infinite, and the metric on the $k$-dimensional real vector space holding the root system is no longer Euclidean; see \Cref{k-le-8}.
The surface $\mathbb{B}_9$ gives rise to a rank-9 parabolic lattice, 
while $\mathbb{B}_{10}$ and $\mathbb{B}_{11}$ correspond to rank-10 and 11 hyperbolic lattices.

\begin{table}[H]
%\begin{center}
\centering
\renewcommand{\arraystretch}{1.2}
\begin{tabular}{ccccccc}
\hline
$D$ &  $k$ & {\bf Type of $E_k$} & {\bf Kac-Moody algebra $\mathfr{g}$} & {\bf Non-Fano Surface} & 
{\bf Model} & {\bf Maximal Split Torus}  \\
\hline
 \hline
\rowcolor{lightgray} 
2 & 9 &  $E_9= \widehat{E}_8$ & affine $\mathfr{e}_9=\widehat{\mathfr{e}}_8$  & 
%Enriques/Halphen
  $\mathbb{B}_9$ 
&  $\mathcal{L}_c^9 S^4$ & $\GG_m^9 \times \GG_m$  \\
1 & 10 & $E_{10}$ & hyperbolic $\mathfr{e}_{10}$  & 
%Enriques/Coble
  $\mathbb{B}_{10}$ 
&  $\mathcal{L}_c^{10} S^4$ & $\GG_m^{10} \times \GG_m$  \\
\rowcolor{lightgray} 
0 & 11 & $E_{11}$ & Lorentzian $\mathfr{e}_{11}$  & 
%Coble
  $\mathbb{B}_{11}$ 
&  $\mathcal{L}_c^{11} S^4$ & $\GG_m^{11} \times \GG_m$  \\
\hline 
\end{tabular}
%\end{center}
\vspace{-2mm} 
\caption{\label{table2} \footnotesize The $E_k$ pattern in 
%Lie theory $(0 \leq k \leq 8)$ and 
Kac-Moody theory 
$(k\geq 9)$, further blowups of $\CC \PP^2$ and cyclifications of $S^4$.}
\end{table}

%\bigskip 
Note that our approach can be made quite general by looking at other 
topological spaces than $S^4$. This would then lose the connection to M-theory, but 
the rational homotopy theory aspects would still be interesting to explore. 
For instance, other spheres would have the same toric symmetries of the rational homotopy type and produce the same root data.

\vspace{-3mm} 
\paragraph{\large Acknowledgments.}
We are grateful to Alexey Bondal, Igor Dolgachev, Amer Iqbal, Mikhail Kapranov, and Urs Schreiber for helpful discussions. We are also grateful for the suggestion of the referee and editor to split the paper into a more mathematical part, which is what this paper is, and a more physical follow-up part \cite{SV:M-theory}. We appreciate that the anonymous referee practically worked with us on weeding out errors and restructuring the exposition to improve the paper.
The first author thanks the University of Minnesota, the Aspen Center for Physics, and 
the Park City Mathematics Institute (IAS) for hospitality during the work on this project,
and acknowledges the support by Tamkeen under the NYU Abu Dhabi Research Institute grant CG008.
The second author thanks NYU Abu Dhabi and Kavli IPMU for creating remarkable opportunities to initiate and work on this project.
% The
His work
% of the second author
was also supported by World Premier International Research Center Initiative (WPI), MEXT, Japan, and a Collaboration Grant from the Simons Foundation (\#585720).

%%%%%%%%%%%%%%%%%%%%%%%%%%%
% \section{Perspective on M-theory via the 4-sphere} 
\section{The 4-sphere and its cyclifications as the universal targets for M-theory and its reductions} 
\label{Sec-S4cyc}
%%%%%%%%%%%%%%%%%%%%%%%%%%%%

We will provide our main topological setting for the rest of the paper using (rational) 
homotopy theory, along the lines outlined in the Introduction.

%%%%%%%%%%%%%%%%%%%%
\subsection{The Sullivan minimal model} 
\label{smm}
%%%%%%%%%%%%%%%%%%%%

Here we replace the notion of a \emph{rational Sullivan minimal model}
of a topological space with a less common notion of a \emph{real
  Sullivan minimal model}, given that real coefficients of physical
fields could be a bit more natural than rational ones (see the 
discussion in \cite{FSS-Chern}). We will therefore
assume that our algebraic models are defined over the reals $\RR$
(see \cite{BSzcz}\cite{GM13}). 

\medskip
To every path-connected, \emph{nilpotent} (the fundamental group is
nilpotent and acts nilpotently on higher homotopy groups) topological
space $Z$, rational homotopy theory associates a minimal Sullivan
algebra, called the \emph{Sullivan minimal model} $M(Z)$ of $Z$,
a differential graded commutative $\RR$-algebra (DGCA) $M = M(Z) = (S(V), d)$ of a certain type, called a minimal Sullivan algebra,  after \cite{Su77}; see the definition below and standard rational homotopy theory references, e.g., \cite{felix-halperin-thomas}\cite{FOT08}\cite{GM13}.

\medskip
Here and henceforth, we
will be restricting our attention to spaces which have
finite-dimensional real homology groups and, respectively, minimal
Sullivan algebras having \emph{strong finite type}, i.e., based
on a graded vector space $V$ of finite total dimension, $\dim V <
\infty$. The spaces of  interest below, namely 
the four-sphere $S^4$ and its cyclifications, satisfy this condition.

\begin{defn}[Sullivan minimal models]
\label{SMM}
\begin{enumerate}[{\bf (i)}]
\setlength\itemsep{-2pt}
    \item 
  A \emph{Sullivan algebra} is a differential graded commutative
  $\RR$-al\-ge\-bra (DGCA) $(M,d)$ based on the free graded commutative
  algebra $M = S(V)$ on a graded real vector space $V = \bigoplus_{n > 0}
  V^n$ with a differential $d: M \to M$ of degree 1,
$d^2 = 0$, satisfying the following nilpotence condition, known as the
  \emph{Sullivan condition}: $V$ is the union of an increasing series
  of graded subspaces
  \begin{equation}
  \label{filtration}
  V(0) \subseteq V(1) \subseteq \dots
  \end{equation}
  such that $d(V(0)) = 0$ and $d(V(k)) \subseteq S(V(k-1))$ for $k \ge 1$.

\item
A \emph{Sullivan model} of a DGCA $A$ is a Sullivan algebra $M$ with a
\emph{quasi-isomorphism} $M \to A$, i.e., a homomorphism which
induces an isomorphism on cohomology.

\item
 We say that a Sullivan algebra is \emph{minimal} if
  \[
  d(M) \subseteq M^+ \cdot M^+,
  \]
  where $M^+ := \bigoplus_{n>0} M^n = S^{\ge 1}(V)$.

\end{enumerate}
\end{defn}

A minimal Sullivan model of a \emph{connected} (i.e., $A^n = 0$ for $n < 0$ and
$A^0 = \RR$) DGCA $A$ exists and is unique up to isomorphism.

\medskip
To every topological space $Z$, Sullivan's construction in
rational/real homotopy theory associates a DGCA $A_{\on{PL}}(Z)$,
called the algebra of real polynomial differential forms on $Z$.\footnote{We clarify that, by a little abuse
of terminology, we will often say ``rational'' even when one should more
correctly say ``real.'' However, it will always be clear from the context what
field of coefficients we is working with.}
If
$Z$ is a smooth manifold, one can take the de Rham algebra of smooth
differential forms on $Z$ instead. (This example is the main reason why 
we prefer real homotopy theory to rational one). If $Z$ is path-connected, 
nilpotent, and has finite-dimensional
real homology groups, then the DGCA $A_{\on{PL}}(Z)$ gives rise to a
minimal Sullivan model $S(V)$, defined up to isomorphism, called the
\emph{$($real$)$ Sullivan minimal model} of $Z$.

%%%%%%%%%%%%%%%%%%%%%%%%
\subsection{$S^4$ as the universal target for M-theory via Hypothesis H} 
%%%%%%%%%%%%%%%%%%%%%%
We adopt the perspective proposed in \cite{Sati13} of viewing the 4-sphere  $S^4$ as the 
universal space of form fields in M-theory. The  significance of this is that  $S^4$ encodes, entirely in its topology, the field $G_4$ 
and its dual $G_7$ as well as their 
{\it dynamics}.
%from the topological point of view.  
This space is  viewed as a universal space in the sense that these field
configurations are given at the homotopy level by real homotopy classes of 
maps from spacetime $Y^{11}$ to $S^4$, and whenever geometry is included, one
would need all maps; see \cite{FSS15}\cite{FSS17}\cite{GS21} (but 
here we will concentrate on topology).

\medskip 
In 11-dimensional supergravity, which is the low-energy limit of M-theory, the equations of motion (EOMs) are \cite{CJS}
 $$
 d G_4 = 0, \qquad d*G_4 + \tfrac{1}{2}G_4 \wedge G_4=0\;.
 $$
When combined with the self-duality condition
\begin{equation}
  \label{C-fields}
G_7 := *G_4 \;,
\end{equation}
these may be rewritten as \eqref{EOM} by using $G_4$ and $G_7$, where
% \begin{equation}
% \label{EOM}
% d G_4 = 0, \qquad d G_7 + \tfrac{1}{2}G_4 \wedge G_4 = 0 \;.
% \end{equation}
the fields $G_4$ and $G_7$ are represented by differential forms of degree 4 and 7, respectively, on the 11-dimensional spacetime $Y^{11}$ 
of M-theory, and $*$ denotes the Hodge star operator, which captures the dependence on the metric on $Y^{11}$.
Note that {\it locally} we may write  
$$
G_4 = dC_3, \qquad G_7 = dC_6 - \tfrac{1}{2} C_3 \wedge G_4\;
$$ 
for some differential forms $C_3$ and $C_6$, viewed as the corresponding {\it potentials}.
However, we will use the duality-symmetric (doubled field) 
formulation, where $G_4$ and $G_7$ are treated 
independently \cite{BBS98} (see also \cite{MaS} \cite{Sati-Form}\cite{ST17}\cite{BSS}
for more global treatments). This will also 
suppress any explicit dependence on the metric, suitable for 
our topological perspective, which may be regarded as describing the topological background of the full story. To get the full picture at the level of fields, one simply adds metric data to spacetime and imposes the duality relation $* G_4 = G_7$. We have also found a way to add metric data to the universal real homotopy model of $S^4$ (and its cyclifications) via introducing moduli parameters; see the end of \Cref{Sec-Ek}.

\medskip
The topological aspects of the M-theory dynamics at the real homotopy level are captured by the real homotopy
theory description given by the Sullivan minimal model 
%(see more details in the Introduction and in \cref{smm}), 
of $S^4$, \cite{Sati13},
which we denote $M(S^4)$:\footnote{We will use lowercase letters for universal elements and uppercase letters to denote
% their incarnations as
spacetime fields.
% , see \eqref{homo}.
}  
\vspace{-1mm} 
\begin{gather}
 \nonumber
  M(S^4) = (\RR[g_4, g_7], d)\, ,\\
 \label{S^4}
dg_4 = 0, \qquad dg_7 = -\tfrac{1}{2} g_4^2 \,,
\end{gather}
where the degree of each of the generators $g_4$ and $g_7$ is given by the 
corresponding subscript:
$\abs{g_4} = 4$, $\abs{g_7} = 7$.
Here we are choosing to include the factor of $-\tfrac{1}{2}$ in the
expressions of the model, as opposed to being absorbed by the generators
(see \cite[Ex.\ 3.3]{FSS17}).

\medskip 
Comparing \eqref{EOM} with \eqref{S^4}, we see that there exists a differential graded (dg) algebra homomorphism
\begin{align}
\label{homo}
M(S^4) &\longrightarrow (\Omega^\bullet(Y), d),
\\[-1mm] 
\nonumber
g_4 & \longmapsto G_4,
\\[-1mm]
\nonumber
g_7 & \longmapsto G_7,
\end{align}
where $(\Omega^\bullet(Y), d)$ is the de Rham algebra of the 11-dimensional spacetime $Y^{11}$. The de Rham algebra is, in fact, a real homotopy model of the manifold $Y^{11}$, and this model could be different from the Sullivan minimal model.
Rational (or, actually, real \cite{FSS-Chern}) homotopy theory provides a canonical continuous map
\begin{equation}
\label{q-cohomotopy}
Y \longrightarrow  S^4_\R,
\end{equation}
where $S^4_\R$ is the \emph{rationalization over $\RR$}  
of $M(S^4)$, a certain universal topological space whose Sullivan minimal model $M(S^4_\R)$ is isomorphic $M(S^4)$, such that the pullback map from the Sullivan minimal model of $S^4_\R$ to the de Rham model of $Y^{11}$ is given by \eqref{homo}.
The space $S^4_\R$ has the same real homotopy type as $S^4$ via a map $S^4 \to S^4_\R$, but $H_n(S^4_\R; \ZZ)$ is a real vector space for each $n \ge 1$. The rationalization may be obtained via a simplicial construction from the dg-commutative algebra $M(S^4)$; see \cite{BSzcz}\cite{FSS-Chern}.

\medskip
In \cite{Sati13}, it is suggested that there is actually a continuous map to the honest-to-goodness 4-sphere $S^4$,
\begin{equation}
\label{cohomotopy}
Y \longrightarrow  S^4,
\end{equation}
that induces the homomorphism \eqref{homo}.
Indeed, a comparison of this target to its linearization, that is to say, the
Eilenberg-MacLane classifying space $K(\mathbb{Z}, 4)$,
which encodes the C-field as captured by a degree 4 class, with the 
4-sphere is presented in \cite{GS21} through a Postnikov tower analysis. 
 The nonabelian nature and the shift by Pontrjagin classes in the 
 quantization condition of the C-field are studied in \cite{FSS19b}.

\medskip
The generator $g_4$ of $M(S^4)$ may be realized as a volume form on the sphere $S^4$, but in the de Rham algebra of $S^4$, there is no room for $g_7$. One may desire to have a model of the universal target space for M-theory, so that not only $g_4$, but also $g_7$ may be realized as a differential form. One such model was suggested to us by M.~Kapranov: it is the complement  
$\HH \PP^\infty \mathbin{\vcenter{\hbox{$\scriptscriptstyle\mathrlap{\setminus}{\hspace{.2pt}\setminus}$}}} \HH \PP^{\infty-2} = \bigcup_N \HH \PP^N \mathbin{\vcenter{\hbox{$\scriptscriptstyle\mathrlap{\setminus}{\hspace{.2pt}\setminus}$}}} \HH \PP^{N-2}$ of a codimension 2 plane $\HH \PP^{\infty-2}$ in the quaternionic projective space $\HH \PP^\infty$. The infinite-dimensional manifold $\HH \PP^\infty \mathbin{\vcenter{\hbox{$\scriptscriptstyle\mathrlap{\setminus}{\hspace{.2pt}\setminus}$}}} \HH \PP^{\infty-2}$ is homotopy equivalent to $\HH \PP^1 \cong S^4$ and may prove to be useful in 11d supergravity.

%%%%%%%%%%%%%%%%%%%
\subsection{The Quillen model and M-theory gauge structure} 
\label{Sec-Quillen} 
%%%%%%%%%%%%%%%%%%%

The Sullivan minimal model $M(Z)$ (see \cref{smm}) 
of each of the spaces $Z$ we are
considering has quadratic differential. This model 
is actually the symmetric algebra on a space of
generators having a certain homotopy-theoretic meaning:
\[
M(Z) = (S(Q(Z)[-1]^*), d),
\]
where $Q(Z)$ is the \emph{Quillen minimal model \footnote{Here we abandon the traditional notion of minimality, based on a free graded Lie algebra, in favor of a more modern one: $Q(Z)$ is an $L_\infty$-algebra with the zero differential, see 
\cite{BFMT20}. The differential $d$ on $M(Z)$ may be
identified as the Chevalley-Eilenberg differential, but this is beside
the point here.} of} $Z$ \cite{Quillen}, which in our quadratic-differential
case is given by the \emph{graded Lie algebra of real homotopy groups}
\[
Q(Z) := \pi_\bullet(Z) \otimes \RR[1]
\]
of $Z$.  Let us explain what this means.
Consider the \emph{real homotopy groups of $Z$}, as a graded vector
space over $\RR$ 
\[
 \pi_\bullet (Z) \otimes \RR := \bigoplus_{i \in \ZZ} \pi_i (Z) \otimes
 \RR,
\]

\vspace{-2mm} 
\noindent where
\[
\pi_i (Z) \otimes \RR :=
\left\{ \!\!\!\!
\begin{array}{cl}
  0 & \quad \text{for $ i \le 0$},\\
  \big(\pi_1 (Z)/[\pi_1(Z), \pi_1(Z)]\big) \otimes_\ZZ \RR &
  \quad \text{for $i = 1$},\\
  \pi_i (z)  \otimes_\ZZ \RR & \quad \text{for $i \ge 2$}.
  \end{array}
  \right.
\]
 If we shift the grading down by one and consider the
natural isomorphism
\[
\pi_\bullet (Z) \otimes \RR [1] \xrightarrow{\;\;\sim\;\;} \pi_\bullet (\Omega
Z) \otimes \RR\;,
\]
where $\Omega Z := \Map_* (S^1, Z)$ is the based loop space of $Z$, we will
get a natural graded Lie-algebra structure with respect to the
\emph{Samelson product}. This is induced on $\pi_\bullet (\Omega Z)
\otimes \RR$ by the commutator map $[-,-]: \Omega Z \times \Omega Z
\to \Omega Z$ of the concatenation $\Omega Z \times \Omega Z \to
\Omega Z$ of based loops in $Z$:
\[
  S^{a+b} = S^a \times S^b /S^a \vee S^b \xrightarrow{\; \gamma_{\, 1} \times
    \gamma_{\, 2}\;} \Omega Z \times \Omega Z \xrightarrow{\;[-,-]\;} \Omega Z\;,
  \]
for $\gamma_{\, 1} \in \pi_{a} (\Omega Z)$ and $\gamma_{\, 2} \in \pi_{b}
(\Omega Z)$, where $S^a \vee S^b$ is the wedge sum of the pointed spaces $S^a$ and  $S^b$.

\medskip 
An equivalent, more standard description of the corresponding
Lie bracket on $\pi_\bullet (Z)[1] \otimes \RR \cong \pi_\bullet
(\Omega Z) \otimes \RR$, is known as \emph{Whitehead product}, which does
not appeal to based loop spaces; see, e.g.,
\cite{felix-halperin-thomas}. We will not need it here
(see \cite{FSS-WZW}\cite[\S 3.2]{FSS-Chern} for detailed discussion in this context).

\medskip 
Now let us return to the case
$Z = \mc{L}_c^k S^4$. We will start with
$k =0$.
The following is, in a  sense, dual
% (see \eqref{SulQuil})
to the description of the fields via the Sullivan minimal model of $S^4$ in \eqref{S^4}.
% and \eqref{homo}.

\begin{example}[M-theory gauge algebra via the Quillen model of $S^4$] 
\label{Ex-Malgebra} 
 The Quillen model of $S^4$ is just the graded Lie algebra on two
generators
\vspace{-2mm} 
\begin{gather*}
  Q(S^4) = \RR e_3 \oplus \RR e_6,\\
  \abs{e_3} = 3, \quad \abs{e_6} =
  6,\\
  [e_3,e_3] = e_6, \quad [e_3,e_6] = 0, \quad [e_6,e_6] = 0,
\end{gather*}
which is actually the free graded Lie algebra over $\RR$ generated by $e_3$. This captures the algebra of gauge transformation of the C-field and its
dual and, hence, also captures the Dirac quantization of the M-branes
(see \cite{CJLP2}\cite{LLPS99}\cite{KS03}\cite{tcu}).
\end{example} 

We will revisit the reduction of this algebra,
corresponding to cyclic loop spaces $\mc{L}_c^k S^4$ for $k > 0$
in Example \ref{Ex-reducedAlg}.

%%%%%%%%%%%%%%%%%%%%%%%%%%%%%%%%%%%%%%
\subsection{Dimensional reduction of M-theory on tori and iterated cyclifications of $S^4$}
\label{cyclification}
%%%%%%%%%%%%%%%%%%%%%%%%%%%%%%%%%%%%%

% {\color{red} Reduce}

%\medskip 
\paragraph{Type IIA and looping.} 
The reduction from M-theory in 11 dimensions to type IIA string theory in 10 dimensions is 
captured by looping. Such a process has been utilized topologically 
at the level of bundles  
leading to loop bundles in 10 dimensions starting from an $E_8$ gauge bundle
(capturing the purely topological aspects of $G_4$) in 11 dimensions \cite{MaS}. 
In our current case,  we do this at the level of universal target spaces, taking into account the rotation of the circle. This leads to the concept of 
 a cyclic loop space or cyclification, advocated in \cite{FSS-L00} (see \eqref{hquotient} below)
 % that is  the \emph{homotopy quotient} 
 \vspace{-2mm} 
$$
\mc{L}_c S^4: = \mc{L} S^4 \dslash S^1.
$$

\vspace{-2mm} 

\paragraph{Why cyclic loop space.} 
Let us provide mathematical justification of the appearance of the cyclic loop space. The
10-dimensional spacetime of type IIA string theory is actually the
``11-dimensional spacetime $Y^{11}$ wrapped on $S^1$'',
which translates to the mathematical language simply as the
10-dimensional quotient $Y^{11}/S^1$ by a free action of $S^1$. Now, there
is an adjunction
\vspace{-2mm} 
\begin{equation}
\label{adjunction}
\Mor_{/BS^1} (Y^{11} / S^1, \mc{L}_c Z) \xrightarrow{\;\; \sim \;\;} \Mor (Y^{11}, Z),
\end{equation}

\vspace{-1mm} 
\noindent where $Y^{11}$ is a space with a free action of $S^1$, $Z$ is another
topological space, the left-hand side is the set of morphisms in the
category $/BS^1$ of spaces over $BS^1$ (equivalent to the category of
principal $S^1$-bundles, such as $Y \to Y/S^1$ and $\mc{L} Z \to
\mc{L}_c Z$), and the right-hand side is the set of continuous maps $Y
\to Z$; see \cite{FSS-L00}\cite[Theorem 2.44]{BSS}. This adjunction produces a map
\begin{equation}
  \label{cyclified}
Y/S^1 \longrightarrow \mc{L}_c S^4
\end{equation}
from the map \eqref{cohomotopy} (or using the corresponding rationalization
$S^4_\RR$ of $S^4$ over $\RR$ and the map \eqref{q-cohomotopy}). If the action of $S^1$ on $Y^{11}$ 
is not free, one has to replace the naive quotient $Y^{11}/S^1$
by the homotopy quotient $Y^{11} \dslash S^1$ in the above.
Roughly speaking, thinking of $\mc{L}_c S^4$ as the space of unparameterized (also known as equivariant) free
loops in $S^4$, a map like \eqref{cohomotopy} from an $S^1$-space $Y^{11}$ will produce 
a map \eqref{cyclified} which assigns to a point $y \in Y^{11}/S^1$ the map that takes 
the $S^1$-orbit in $Y^{11}$ over $y$ to $S^4$ by the given map $Y^{11} \to S^4$. 
% This is a main point of $\mc{L}_c S^4$.
 
\medskip 
 We would like to make more precise what we mean by iterated cyclic loop spaces. 
 Let $\mc{L} Z= \Map (S^1, Z)$ be the \emph{free loop space} of a
topological space $Z$, which we will assume to be path-connected. The
free loop space admits a natural (right) action of the group $S^1$ by
rotating loops
\[
(f \cdot z) (z') = f(z z')\,,
\]
for $f \in \mc{L} Z$ and $z , z' \in S^1$, and we define the
\emph{cyclic loop space} or \emph{cyclification} $\mc{L}_c Z$ to be the \emph{homotopy quotient}
\begin{equation}
\label{hquotient}
\mc{L}_c Z: = \mc{L} Z \dslash S^1.
\end{equation}
One may construct it using the \emph{Borel construction} $(\mc{L} Z
\times ES^1)/S^1$, where $ES^1$ is the universal space for $S^1$ bundles,
as the quotient by the diagonal action of $S^1$,
i.e., the quotient by the relation
\[
(f \cdot z, e) \sim (f, z \cdot e)\,,
\]
for $f \in \mc{L} Z$, $e \in ES^1$, and $z \in S^1$. We will use
the convention that \emph{if $\, Z$ is not simply
  connected}, when $\mc{L} Z$ acquires path components, we
\emph{retain only the component of the constant loop} and in that case use the same
notation
\[
\mc{L} Z := (\mc{L} Z)_0\;,
\]
which also makes our cyclic loop spaces
\emph{path-connected}.

 \paragraph{M-theory on $T^k , k \leq  8$ and higher cyclic loop spaces.}
The reduction of the above system on tori $T^k = (S^1)^k$ leads to a low-dimensional
system corresponding to duality-symmetric supergravity actions in
these dimensions matching the EOMs of M-theory
compactified on a $k$-dimensional torus $T^k$ in an iterative way, as we explain below. 

\medskip 
Suppose the 10-dimensional spacetime $X^{10} = Y^{11}/S^1$ of type IIA string theory
(or the 10-dimensional spacetime $X^{10}$ of type IIB supergravity) also has 
a free action of $S^1$. Then applying the adjunction \eqref{adjunction} 
to the map \eqref{cyclified}, we get a natural map
\[
X^{10}/S^1 = (Y^{11}/S^1)/S^1 \longrightarrow \mc{L}_c^2 S^4.
\]
This would be the second toroidal compactification of M-theory to a 9-dimensional spacetime. Iterating this process, for each $k$, $ 0 \le k \le 11$, we will be getting a map 
\[
\left( \dots((Y^{11}/S^1)/S^1) \dots /S^1\right) \longrightarrow \mc{L}_c^k S^4,
\]
the quotient being by $k$ copies of $S^1$,
from the $k$th toroidal compactification of the 11-dimensional M-theory to the $k$-fold cyclification $\mc{L}_c^k S^4$ of the four-sphere.

\medskip 
Hence, for $k \ge 0$, the \emph{iterated
  cyclic loop space} $\mc{L}_c^k Z$ is the $k$-fold iteration of the
cyclic loop space construction:
\begin{align*}
  \mc{L}_c^0 Z & := Z,\\
  \mc{L}_c^k Z & := \mc{L}_c(\mc{L}_c^{k-1} Z)  \qquad \text{for $k \ge 1$}.
  \end{align*}
  We will often refer to iterated loop spaces as \emph{cyclifications}.
We will be interested mostly in the iterated cyclic loop spaces
$\mc{L}_c^k S^4$ of the 4-sphere $S^4$, both for $0 \le k \le 8$ and for $k \geq 9$.
For $k=1$ this is studied extensively in \cite{FSS-pbranes}\cite{FSS-L00}
in relation to T-duality and the twisted K-theory description of the fields
in type II string theory (cf. \cref{IIB}). The generalization
is possible thanks to the general construction in \cite{BSS}.

\medskip 
The above matching of the EOMs with the differential on the Sullivan model of the $k$-fold cyclic loop space $\mc{L}_c^k S^4$ is another striking phenomenon, which we observe in this paper; see \cref{Sullmm}.

\begin{example}[The brane/reduced M-theory gauge algebra via the Quillen model of the cyclification of $S^4$] 
\label{Ex-reducedAlg} 
We have seen the description of the M-theory gauge algebra via Quillen model of $S^4$ 
is just the graded Lie algebra on two generators in Example \ref{Ex-Malgebra}. 
The gauge algebra of the reduced fields in $11-k$ dimensions 
will correspond to the Quillen model of $\mc{L}_c^kS^4$.
We will not work out the details, as it is clear that matching the fields and EOMs 
of reduced M-theory with
the generators and
their differentials in the Sullivan minimal model of
the cyclification of $S^4$
% $\mc{L}_c^kS^4$
implies similar matching between the gauge algebra and the Quillen model. Thus, the gauge algebra 
formulas in \cite{CJLP2}\cite{LLPS99} should be reproduced just by looking at the Quillen model of $\mc{L}_c^kS^4$. 
\end{example}

%%%%%%%%%%%%%%%%%%%%%%%%%%%%%%%%%%%%%%%%%
\subsection{The Sullivan minimal model of the cyclic loop space}
%%%%%%%%%%%%%%%%%%%%%%%%%%%%%%%%%%%%%%%%%
\label{Sullmm}

Our goal here is to describe the Sullivan minimal model of the cyclic loop space 
$M(\mc{L}_c Z)$ in terms of the Sullivan minimal model $M(Z) = S(V)$ of the space $Z$.
Along the way, we also provide a duality-symmetric reduction of fields 
in M-theory on tori, extending and organizing partial/local results in the 
supergravity and M-theory literature. 

\medskip 
Suppose $Z$ is a path-connected, nilpotent space with finite-dimensional real 
homotopy groups. With our convention in \cref{cyclification}, we claim the following:
\begin{prop}[Basic properties of  the cyclic loop space] 
\label{Prop-LcZconn}
\begin{enumerate}[\bf (i)]
  \setlength\itemsep{-2pt}
\item
$\mc{L}_c Z$ is also path-connected, nilpotent, and has finite-dimensional
rational homotopy groups.
\item 
The Sullivan minimal model $M(\mc{L}_c Z)$ of $\mc{L}_c Z$ is given by the extension \eqref{ext:cyc} below.
\end{enumerate}
\end{prop} 

\begin{proof}
\textbf{(i)}
Indeed, the path-connected free loop space $\mc{L} Z$
sits in a fiber sequence 
\(
\label{Based-unbased}
\Omega Z \to \mc{L} Z \xrightarrow{\;{\rm ev}_*\;} Z,
\)
where ${\rm ev}_*$ is evaluation at the basepoint $* \in S^1$ and $\Omega Z =
\Map_*(S^1, Z)$ is the \emph{based loop space}.
The constant-loop section of ${\rm ev}_*$ splits the homotopy groups:
$\pi_i(\mc{L} Z) = \pi_i(Z) \oplus \pi_i(\Omega Z) = \pi_i(Z) \oplus \pi_{i+1} (Z)$ 
for $i \ge 2$ and $\pi_1 (\mc{L} Z) = \pi_1(\Omega Z) \rtimes \pi_1(Z)  = \pi_2(Z) \rtimes \pi_1(Z)$. The action of $\pi_1 (\Omega Z)$ on $\pi_i(\mc{L} Z)$ for $i \ge 2$ is trivial, while the action of $\pi_1 (Z)$ on $\pi_i(\mc{L} Z)$ for $i \ge 2$ is the sum of actions on $\pi_i(Z)$ and $\pi_{i+1}(Z)$. Therefore, the action of $\pi_1 (\mc{L} Z)$ on $\pi_i (\mc{L} Z)$, $i \ge 2$, is nilpotent. The group $\pi_1 (\mc{L} Z)$ itself is nilpotent as a semidirect product of a nilpotent group acting nilpotently on an abelian group. Indeed, the action of $\pi_1(Z)$ on $\pi_1(\Omega Z)$ defining the semidirect product is compatible with the standard action of $\pi_1(Z)$ on $\pi_2(Z)$, which is assumed to be nilpotent, see , e.g., \cite{AFO17}. As concerns the homotopy groups of the cyclic loop space $\mc{L}_c Z$, it just adds a copy of $\ZZ$ 
to $\pi_2(\mc{L} Z)$, on which $\pi_1(\mc{L}_c Z) = \pi_1(\mc{L} Z)$ acts trivially, as one can see  from the fiber sequence
\begin{equation*}
%\label{hfs}
S^1 \to \mc{L} Z \times ES^1 \to \mc{L}_c Z,
\end{equation*}
in which the inclusion of the fiber $S^1$ over the constant loop factors through the contractible space $ES^1$.

\medskip
\noindent
\textbf{(ii)} Vigu\'{e}-Poirrier and Burghelea \cite{vigue-burghelea} show\footnote{Vigu\'{e}-Poirrier and Burghelea assume that $Z$ is simply connected, but their argument applies to the more general nilpotent case verbatim, by taking $V[1]$ to be the truncated
desuspension \eqref{truncation} of $V$, given that Halperin's theorem on fibrations \cite[Theorem 20.3]{Halperin83} is done in the nilpotent case.} that the Sullivan minimal model
$M(\mc{L}_c Z)$ is isomorphic to $(S(V \oplus V[1] \oplus \RR w), d_c)$,
which is an \emph{extension} (in the sense of Halperin \cite{Halperin83}, also known as
a \emph{$\Lambda$-extension})
\begin{equation}
  \label{ext:cyc}
(\RR[w],0) \longrightarrow (S(V \oplus V[1] \oplus \RR w), d_c) \longrightarrow (S(V \oplus
V[1]), d_f)
\end{equation}
corresponding to the homotopy fiber sequence
\begin{equation*}
%\label{hfs}
\mc{L} Z \to \mc{L}_c Z \to  B S^1 .
\end{equation*}
Here
$
V[1] = \bigoplus_{n > 0} V[1]^n
$
is the \emph{truncated desuspension}, the graded vector space with
components
\vspace{-2mm}
\(
\label{truncation}
V[1]^n : =
\begin{cases}
  V^{n+1} & \text{for $n > 0$},\\
0 & \text{for $n \le 0$}.
\end{cases}
\)
The truncation affects only the non-simply connected case, when $V^1
\ne 0$, as it truncates only this graded component. The differentials
$d_f$ and $d_c$ may be described as follows. Let
\begin{align*}
  s: V & \longrightarrow V[1],\\
  v & \longmapsto
  \begin{cases}
    v & \text{for $v \in V^n$, $n > 1$},\\
    0 & \text{for $v \in V^1$},
    \end{cases}  \end{align*}
be the natural map, a surjection of degree $-1$. Then, for all $v \in V$,
\begin{align}
 \nonumber
 d_f v  &:= dv, \\
 \label{df}
 d_f sv  &:= -s dv,
\end{align}
 where $s$ is extended to a unique degree-$(-1)$ derivation
of $S(V \oplus V[1])$ such that $s (u) = 0$ for all $u \in V[1]$. Since this derivation $s$ is a version of a graded polynomial de Rham differential, we have $s^2 =0$. Then, for all $ v \in V$, we define
\vspace{-2mm}
\begin{align}
\nonumber d_c v &:= d v + sv \cdot w,\\[-2pt]
  \label{dc}
  d_c sv  &:= -s d v,\\[-2pt]
  \nonumber
  d_c w  &:= 0. \qedhere
\end{align}
\end{proof}

\medskip 
Similarly, the Sullivan minimal model of the free loop space $\mc{L}
Z$ is a Halperin extension
\begin{equation}
  \label{ext:free}
(S(V), d) \longrightarrow (S(V \oplus V[1]), d_f) \longrightarrow (S(V[1]), 0)
\end{equation}
corresponding to the fiber sequence \eqref{Based-unbased}. 

 \begin{example}[Reduction on a circle and  the cyclification
    $\mc{L}_c S^4$] 
    \label{example1}
From $M(S^4)$ in \eqref{S^4}, the Sullivan minimal model $ M(\mc{L}_c S^4)$ of the cyclic loop
space of $S^4$ can be presented as (cf. \cite[Ex. 3.3]{FSS17}\cite[Ex. 2.7]{FSS-L00})
\begin{gather*}
% \nonumber
  M(\mc{L}_c S^4) = (\RR[g_4, g_7, sg_4, sg_7, w], d)\;,\\
%  \label{EOM1}
  dg_4 = sg_4 \cdot w, \qquad dg_7 = -\tfrac{1}{2} g_4^2 + sg_7 \cdot w\;,\\
%\label{EOM2}
  d sg_4 = 0, \quad dsg_7 = sg_4 \cdot g_4, \quad d w = 0\;.
\end{gather*}
As in \cite[Ex. 2.7]{FSS-L00},
making the change of variables $f_2=w$, $h_3=sg_4$, $f_4=g_4$, 
$f_6=sg_7$, and $h_7= g_7$, this can be rewritten as
\[
M(\mathcal{L}_cS^4)=\big(\mathbb{R}[f_2,f_4,f_6,h_3,h_7],\ df_2=0,\, dh_3=0,\, df_4=h_3f_2,\, df_6=h_3f_4,\,
dh_7= - \tfrac{1}{2} {f_4}^2 + f_2f_6\big).
\]
Being defined for $\mathcal{L}_cS^4$, these equations are universal, and we obtain the corresponding ones in
spacetime by pulling back, 
giving the datum of a closed 3-form $H_3$ and of 2-, 4- and 6-forms $F_2$, $F_4$ and $F_6$ on $X$ such that
\[
dF_2=0;\qquad dF_4=H_3\wedge F_2;\qquad dF_6= H_3\wedge F_4,
\]
together with a 7-form $H_7$, which is a potential for a certain
8-form:
\[
d H_7 = - \tfrac{1}{2} F_4\wedge F_4 + F_2 \wedge F_6;
\]
cf.\ \cite{Campbell}\cite{Huq}\linebreak[0]\cite{Giani} for the classical theory and \cite{BNS04}
for the duality-symmetric formulations of type IIA $D=10$ supergravity. 
Hence, if $Y\to X$ is
rationally a principal $S^1$-bundle, then a
% rational
map $Y \to S^4$ in the rational category
% cocycle on $Y$
will induce, by the hofiber/cyclification adjunction \eqref{adjunction}, such a
set of differential forms on the base $X$.  As explained in
\cite{FSS17}\linebreak[0]\cite{FSS-L00}, the above equations for the differentials
of the $F_{2n}$'s are precisely (a subset of) the equations for a
$H_3$-twisted cocycle $\sum_n F_{2n}u^n$ in
$(\Omega^\bullet(X)(\!(u)\!),d_{H_3})$ with $F_0=0$, corresponding
to the EOMs and Bianchi identities captured by twisted
(rational) even K-theory, as appropriate for type IIA string theory. 
We consider the type IIA and IIB perspectives in \cref{IIB}. 
\end{example}

Proposition \ref{Prop-LcZconn} can be iterated to give the same result for 
$\mc{L}_c^2Z$, so that we can now continue with further dimensional reduction.
\footnote{
We will have multiple circle fiber directions and corresponding labels on 
the contractions $s_i$ and the classes of the circles $w_i$. We realize that
the notation is not fully in parallel with the convention of using such labels to indicate the degree, but choosing another notation such as $s_{(i)}$ might overload the expressions when multiple such occur below. We hope the distinction 
will be clear from the context.}

\begin{example}[Reduction on a 2-torus and the double cyclification
    $\mc{L}_c^2 S^4$]
  \label{example2}
%Let us move on one more step and describe 
\begin{sloppypar}
For the Sullivan minimal model
$M(\mc{L}^2_c S^4)$ of the double cyclification of the sphere $S^4$, we have
  \begin{align*}
% \nonumber
  M(\mc{L}^2_c S^4) & = \left(\RR[g_4, g_7, s_1 g_4, s_1 g_7, w_1, s_2 g_4,
    s_2 g_7, s_2 s_1 g_4, s_2 s_1 g_7, s_2 w_1, w_2], d\right),\\
%  \label{EOM1}
  dg_4  &= s_1 g_4 \cdot w_1 + s_2 g_4 \cdot w_2, \qquad dg_7 = -\tfrac{1}{2} g_4^2 + s_1 g_7 \cdot w_1 + s_2 g_7 \cdot w_2,\\
%\label{EOM2}
  d s_1 g_4 & = s_2 s_1 g_4 \cdot w_2, \quad ds_1 g_7 = s_1 g_4 \cdot g_4 + s_2 s_1 g_7 \cdot w_2, \\
  d s_2 g_4 & = - s_2 s_1 g_4 \cdot w_1 + s_1 g_4 \cdot s_2 w_1, \quad ds_2 g_7 = s_2 g_4 \cdot g_4 - s_2 s_1 g_7 \cdot w_1 - s_1 g_7 \cdot s_2 w_1, \\
  d s_2 s_1 g_4 & = 0, \quad d s_2 s_1 g_7 = -s_2 s_1 g_4 \cdot g_4 + s_1 g_4 \cdot s_2 g_4, 
  \\
   d w_1 & =  s_2 w_1 \cdot w_2,
  \quad d s_2 w_1 = 0, \quad d w_2 =  0\;.
  \end{align*}
  \end{sloppypar}
\noindent These equations are again universal, and we again obtain the corresponding ones in
spacetime by pullback. These are the EOMs and Bianchi identities
of type II string theory
at low energy, i.e., type II supergravity in 9 dimensions in the 
duality-symmetric formulation. A more common physics notation for the above fields (once pulled back to the 9d spacetime $(X/S^1)/S^1$, but omitting the pullback notation) is
\vspace{-2mm} 
$$
\begin{array}{cccccc}
g_4=F_4, & g_7=H_7, & s_1 g_4=H_3^{(1)}, & s_1 g_7=F_6^{(1)}, &
s_2g_4=H_3^{(2)}, &   s_2 g_7=F_6^{(2)}, 
\\
s_2 s_1 g_4=\mathcal{F}_2, & s_2 s_1 g_7=\mathcal{F}_5, &  w_1=F_2^{(1)},&
w_2=F_2^{(2)}, 
& s_2 w_1 = F_1^{(2)}.
\end{array} 
$$
The classical EOMs are given in \cite{BHO95}\cite{DR}, so the above can be
viewed as a duality-symmetric extension. 
\end{example}

\begin{remark}[Iterated $S^1$ vs.\ direct $T^2 = S^1 \times S^1$ reduction]  
\label{Rem-2S1vsT2}
We compare the two settings:
\item {\bf (i)} In the iterated case, $M(\mc{L}_c^2 S^4)$, notice the appearance of the \emph{axion} $s_2 w_1$, which would be absent in the direct reduction on $T^2$
 corresponding to the ``toroidification'' $M(\Map(T^2, S^4) \dslash T^2)$.
% See \cref{Sec-toroid} for further details. 

\item {\bf (ii)} Note that $dw_1 = s_2 w_1 \cdot w_2$ above, whereas in the direct $T^2$-reduction, we would have $d w_1 = d w_2 = 0$. Likewise, note the two terms in the differential $ds_2g_4$ above, while in the direct reduction, it
  would be on an equal footing with $ds_1 g_4$, that is, $ds_2 g_4=s_1 s_2
  g_4 \cdot w_1$. 
\end{remark}

\begin{example}[Reduction on a 3-torus and the triple cyclification
    $\mc{L}_c^3 S^4$]
  \label{example3}
Since the double cyclification $\mc{L}_c^2 S^4$ is not simply
connected and its Sullivan minimal model has one generator, $s_2 w_1$,
of degree one, for the triple cyclification $\mc{L}_c^3 S^4$, a new
phenomenon happens: the desuspension $V[1]$ becomes truncated, see \eqref{truncation}. We will
describe  the Sullivan minimal model
$M(\mc{L}_c^3 S^4)$. The list of generators will include a new
generator, $w_3$, and those of the previous case,
\Cref{example2}, as well as their desuspensions, $s_3 g$, where $g$ is
a former generator of $M(\mc{L}_c^2 S^4)$, except for $g = s_2 w_1$,
which gets truncated: we may just as well assume that
\begin{equation}
  \label{s_3}
s_3 s_2 w_1 = 0\,.
\end{equation}
This truncation will also affect the equations for the differentials
in the following way. In accordance with \eqref{dc}, the differentials
of the former generators $g$ of $M(\mc{L}_c^2 S^4)$ from 
Example  \ref{example2} above will all acquire an extra term $s_3 g \cdot w_3$, such as
\[
d g_4 = s_1 g_4 \cdot w_1 + s_2 g_4 \cdot w_2 + s_3 g_4 \cdot w_3\,,
\]
except for the equation
\(
\label{eq-ds2w1}
d s_2 w_1 = 0\,,
\)
which will remain intact. 
The differentials of the desuspensions $s_3
g$ of the former generators $g$ of $M(\mc{L}_c^2 S^4)$ will work by
the expression in \eqref{dc}, such as
\[
d s_3 g_4 = - s_3 s_1 g_4 \cdot w_1 + s_1 g_4 \cdot s_3 w_1 - s_3 s_2
g_4 \cdot w_2 + s_2 g_4 \cdot s_3 w_2\,,
\]
but those generators $g$ whose differentials contained $s_2 w_1$ will
be affected in the way that in \eqref{dc}, we shall impose the
relation \eqref{s_3}:
\begin{gather*}
d s_3 w_1 = s_2 w_1 \cdot s_3 w_2, \qquad
d s_3 s_2 g_4 = - s_3 s_2 s_1 g_4 \cdot w_1 - s_2 s_1 g_4 \cdot s_3
w_1 - s_3 s_1 g_4 \cdot s_2 w_1\;,\\
%\quad
d s_3 s_2 g_7 = - s_3 s_2 g_4 \cdot g_4 + s_2 g_4 \cdot s_3 g_4 + s_3
s_2 s_1 g_7 \cdot w_1 + s_2 s_1 g_7 \cdot s_3 w_1 + s_3 s_1 g_7 \cdot
s_2 w_1\;.
  \end{gather*}
Let us list all the $d$-closure equations for the generators of
$M(\mc{L}_c^3 S^4)$, as this will be important for us later:
\begin{equation*}
 d s_3 s_2 s_1 g_4 = 0, \qquad   d s_2 w_1 = 0, \qquad d s_3 w_2 = 0,
  \qquad d w_3 = 0\;.
  \end{equation*}
The above equations are again universal, and we obtain the corresponding ones in
spacetime by pullback, and capture the equations of motion and Bianchi identities 
of type II string theory
at low energy, i.e., type II supergravity in 8 dimensions in the 
duality-symmetric formulation, extending, for instance, \cite{AT85}. 
\end{example}

\begin{example}[Reduction on $T^k$ and 
$k$-fold cyclifications $\mc{L}_c^k S^4$ for $k \ge 3$]
  \label{examplek}
  
The pattern of Example \ref{example3}, as predicted by Equations
\eqref{ext:cyc}--\eqref{dc}, pertains. The $d$-closed generators,
which will play an important role later, consists of 

\vspace{-3mm} 
\begin{enumerate}[{\bf (i)}]
 \setlength\itemsep{-2pt}
     \item 
$k$ elements of
degree one:
\[
s_3 s_2 s_1 g_4, \quad s_2 w_1, \quad s_3 w_2, \quad \dots, \quad s_k w_{k-1}\,,
\]
\item and an element of degree two:
$
w_k$.
\end{enumerate} 

\vspace{-2mm} 
\noindent The number $n_k$ of generators does not follow the recursion $n_{k} =
2 n_{k-1} +1$ seemingly suggested by \eqref{ext:cyc}, because of the
truncations. However, the set of generators is easy to account for:
\begin{align*}
s_{i_l} \dots s_{i_1} g_4,& \qquad \text{where $0 \le l \le 3$ and $1 \le
  i_1 < \dots < i_l \le k$},\\
s_{i_l} \dots s_{i_1} g_7,& \qquad \text{where $0 \le l \le
  6$ and $1 \le i_1 < \dots < i_l \le k$},\\
w_i,& \qquad 1 \le i \le k,\\
s_j w_i,& \qquad 1 \le i < j \le k.
\end{align*}
% and
% \[
% s_j w_i, \qquad 1 \le i < j \le k.
% \]
The pullbacks of their differentials to spacetime 
correspond, likewise, to the EOMs and Bianchi identities of 
duality-symmetric low energy string theory/supergravity in dimensions $11-k$.
Supplying these leads to a duality-symmetric extension of 
the non-duality symmetric versions studied, e.g., in \cite{LP96}\linebreak[0]\cite{LLP98},
and surveyed in \cite{CDF}\cite{Ta98}. 
\end{example}

%%%%%%%%%%%%%
\subsection{Type IIB and T-duality} 
%%%%%%%%%%%%
\label{IIB}

The equivalence of real Sullivan minimal models $M(X) \cong M(Y)$ is in
correspondence with the underlying topological spaces $X$ and $Y$
being rationally (in fact, under our formalism, ``really") homotopy
equivalent. This will allow us to obtain a candidate for the real model
for type IIB without having to immediately  work out a topological space analogue 
for IIB of the cyclification of $S^4$ on the type IIA side. 

\medskip 
Starting with a topological model of the real homotopy type for type
IIA, let us call that $IIA$, and one for type IIB, let us call that
$IIB$, it would be ideal to have an equivalence upon dimensional reduction of both to
nine dimensions, i.e., in the spirit of the approach of \cite{FSS-pbranes}, upon
cyclification, $M(\mathcal{L}_c IIA) \cong M(\mathcal{L}_c IIB)$. This
would correspond to a real homotopy equivalence $\mathcal{L}_c IIA \sim
\mathcal{L}_c IIB$, a universal version of \emph{T-duality}. The situation here is subtler, and this section is devoted to discussing the intricacies of type IIB, versions of the equivalence $\mathcal{L}_c IIA \sim
\mathcal{L}_c IIB$ of real homotopy types, and the discrepancy between them.

\medskip 
Now starting with $S^4$ as the universal space for M-theory (as in 
\cite{Sati13} and \eqref{S^4}), we get the cyclification $\mathcal{L}_c S^4$ 
as the model for type IIA in ten dimensions (as in \cite{FSS17} and Example \ref{example1}). 
Dimensionally reducing 
further amounts to double cyclification 
$\mathcal{L}^2_c S^4$ in nine dimensions (as in Example \ref{example2}). 
On the other hand, dimensionally reducing 
type IIB to nine dimensions leads to the once-cyclified space 
$\mathcal{L}_c IIB$. In terms of $S^4$, the expected equivalence in nine dimensions  
would hence amount to the equivalence of Sullivan models 
$$
M(\mathcal{L}_c IIB) \cong M(\mathcal{L}^2_c S^4)
$$
and a real homotopy equivalence at the level of spaces
$$
\mathcal{L}_c IIB \sim \mathcal{L}^2_c S^4.
$$
While we will not work out the decyclification, we do have a path along which to proceed:
$$
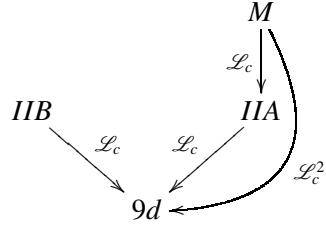
\xymatrix{
&& M \ar[d]_{\mathcal{L}_c}  \ar@/_-3pc/[ddl]^{\mathcal{L}_c^2}
\\
IIB \ar[dr]^{\mathcal{L}_c}  && IIA  \ar[dl]_{\mathcal{L}_c}
\\
& 9d
}
$$
This diagram is not only motivated by physics but also by
Mysterious Duality and the del Pezzo story, in which we have
$$
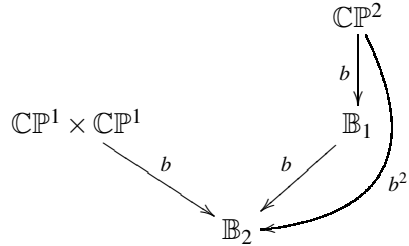
\xymatrix{
&& \CC \PP^2 \ar[d]_{b}  \ar@/_-3pc/[ddl]^{b^2}
\\
\CC \PP^1 \times \CC \PP^1 \ar[dr]^{b}  && \BB_1  \ar[dl]_{b}
\\
& \BB_2
}
$$
where $b$ denotes the process of blowing up (rather than the blowup map, which would go in the opposite direction).

\paragraph{The Sullivan minimal model $M(IIB)$ of the real homotopy type $IIB$.}
Consider the free graded commutative algebra $\RR[\omega_1, h_3, \omega_3, \omega_5, h_7, \omega_7]$, the subscripts denoting the degrees of the respective elements. Define a differential by the following equations:
\begin{align}
\nonumber
d\omega_1&=0,\;\;\;\; \qquad \qquad \quad dh_3  =0,   
\\
\label{htypeIIB}
 d \omega_3&=h_3 \omega_1,  \qquad \qquad \;\; d\omega_5  =h_3 \omega_3,
\\
\nonumber
dh_7& =\omega_3 \omega_5 + \omega_7 \omega_1, \quad d\omega_7 =h_3 \omega_5.
\end{align}
Pulling back these universal differential forms to spacetime, 
we get the EOMs of $D=10$ type IIB supergravity
in the duality-symmetric formulation,
 and without imposing self-duality;
 cf.\  \cite{SW83}\cite{HW84}\cite{Schwarz83} for the classical formulation
 and \cite{DLS97}\cite{DLT98}
for the duality-symmetric formulation. 
Note that we do not include fields of degree greater than seven,
as in the dual picture of type IIA  (see \Cref{example2}),
this would require parametrized homotopy theory \cite{BSS}.
 Additionally, there are several subtleties in type IIB which makes a purely topological 
 perspective delicate due to the mixing between geometry and topology,
 in the sense that some of the fields arise from the metric, as we explain 
 further below (see \cite{BSS}).
% Remark \ref{Rem-notbad}).

\paragraph{T-duality: comparing $M(\mathcal{L}_c IIA)$ and $M(\mathcal{L}_c IIB)$.}
Computing $M(\mathcal{L}_c IIB)$ using the recipe of \cref{Sullmm} and changing the variables $ sh_3, w, s\omega_3, s\omega_5, s\omega_7$ in the notation of \eqref{ext:cyc} and \eqref{dc} to
% This will be the DGCA
% \[
% M(\mathcal{L}_c IIB) = \big( \RR[\omega_1, sh_3, s\omega_3, w, h_3, \omega_3, s\omega_5, \omega_5, sh_7, s\omega_7, h_7, \omega_7], d \big),
% \]
% with the differential $d$ given by
% \begin{gather*}
%   d\omega_1 = 0, \qquad dsh_3 = 0, \qquad d s\omega_3 = - sh_3  \cdot \omega_1 ,\\
%     dw = 0, \qquad dh_3 = sh_3 \cdot w, \qquad d \omega_3 = h_3 \cdot \omega_1 + s \omega_3 \cdot w,\\
% %\label{EOM2}
%   d s\omega_5 = - sh_3 \cdot \omega_3 + h_3 \cdot s\omega_3, \quad d \omega_5 = h_3 \cdot \omega_3 + s \omega_5 \cdot w,\\
%   d sh_7 = - s \omega_3 \cdot \omega_5 + \omega_3 \cdot s\omega_5 - s \omega_7 \cdot \omega_1, \quad d s \omega_7 = - s h_3 \cdot \omega_5 + h_3 \cdot s \omega_5, \\
%   d h_7 = \omega_3 \cdot \omega_5 + \omega_7 \cdot \omega_1 + sh_7 \cdot w, \quad d\omega_7 =h_3 \cdot \omega_5 + s \omega_7 \cdot \,.
%   \end{gather*}
%   A more common notation suitable for comparison to physics would use
\begin{gather*}
  c_2 := s h_3, \qquad \tilde{c}_2 := w, \qquad \omega_2 := s \omega_3,\\
    \omega_4 := s \omega_5, \qquad \omega_6 := s\omega_7,
  \end{gather*}
more common in physics, we get the following presentation.
 \begin{prop}[Cyclification of type IIB to $d=9$ ]
    \begin{gather}
 \nonumber
 M(\mathcal{L}_c IIB) = \left(\RR[\omega_1, c_2,  \tilde{c}_2, \omega_2, h_3, \omega_3, \omega_4, \omega_5, sh_7, \omega_6, h_7, \omega_7], d \right)\\
  \nonumber
d\omega_1 = 0, \qquad d c_2 = 0, \qquad d \tilde{c}_2 = 0,\\
 \nonumber
   d \omega_2 = - c_2 \cdot \omega_1, \quad d h_3 = c_2 \cdot \tilde{c}_2, \quad d \omega_3 = h_3 \cdot \omega_1 + \omega_2 \cdot \tilde{c}_2,\\
\label{LIIB}
  d \omega_4 = - c_2 \cdot \omega_3 + h_3 \cdot \omega_2, \qquad d \omega_5 = h_3 \cdot \omega_3 + \omega_4 \cdot \tilde{c}_2,\\
  \nonumber
 d sh_7 = - \omega_2 \cdot \omega_5 + \omega_3 \cdot \omega_4 - \omega_6 \cdot \omega_1, \quad d \omega_6 = - c_2 \cdot \omega_5 + h_3 \cdot  \omega_4, \\
 \nonumber
  d h_7 = \omega_3 \cdot \omega_5 + \omega_7 \cdot \omega_1 + sh_7 \cdot \tilde{c}_2, \quad d\omega_7 =h_3 \cdot \omega_5 + \omega_6 \cdot \tilde{c}_2\,.
  \end{gather}
\end{prop}

We would like to match it with $M(\mc{L}_c IIA) = M(\mc{L}_c^2 S^4)$, see \Cref{example2}. Writing 
\begin{gather*}
  \omega_2 := w_1, \qquad h_3:= s_1 g_4, \qquad \omega_4 := g_4 ,\\
    \omega_6 := s_1 g_7, \qquad h_7 := g_7,\\
  \omega_1 := s_2 w_1, \qquad c_2 := s_2 s_1 g_4, \qquad \tilde{c}_2 := w_2,\\
    \omega_3 := s_2 g_4, \qquad  \omega_5 := s_2 s_1 g_7, \qquad sh_7 := s_2 h_7\, ,
  \end{gather*}
we get a compatible presentation for the DGCA $M(\mc{L}_c IIA)$.
\begin{prop}[Cyclification of type IIA to $d=9$ ]
\begin{gather}
\nonumber
  M(\mc{L}_c IIA) = \left( \RR[\omega_1, c_2, \tilde{c}_2, \omega_2, h_3,  \omega_3, \omega_4, \omega_5, s h_7, \omega_6,  h_7], d \right),\\
%  \label{EOM1}
 \nonumber
 d \omega_1 = 0, \qquad d c_2 = 0, \qquad d \tilde{c}_2 = 0, \\
 \nonumber
 d \omega_2 = \omega_1 \cdot \tilde{c}_2 , \quad dh_3 = c_2 \cdot \tilde{c}_2, \quad d \omega_3 = - c_2 \cdot \omega_2 + h_3 \cdot  \omega_1,\\
\label{LIIA}
d \omega_4 = h_3 \cdot \omega_2 + \omega_3 \cdot \tilde{c}_2, \qquad d \omega_5 = - c_2 \cdot \omega_4 + h_3 \cdot \omega_3,\\
  \nonumber
 d sh_7 = \omega_3 \cdot \omega_4 - \omega_5 \cdot \omega_2 -  \omega_6 \cdot \omega_1, \quad d \omega_6 = h_3 \cdot \omega_4 + \omega_5 \cdot \tilde{c}_2, \\
  \nonumber
 d h_7 = -\tfrac{1}{2} \omega_4^2 + \omega_6 \cdot \omega_2 + sh_7 \cdot \tilde{c}_2\,.
  \end{gather}
\end{prop}

Now it is time to compare $M(\mc{L}_c IIA)$ with $M(\mc{L}_c IIB)$. 
Here is a striking conclusion:

\begin{prop}
[Matching in 9 dimensions]
\label{Prop-Tdual}
Up to replacement
\vspace{-2mm} 
\[
c_2 \leftrightarrow - \tilde{c}_2,
\]
all the equations for the differential of $M(\mc{L}_c IIA)$ in \eqref{LIIA} exactly match those of $M(\mc{L}_c IIB)$ in \eqref{LIIB}, except for the following mismatch for the generators of degree $7$:
\begin{center}
\begin{tabular}{ccc} 
\hline 
\bf Model & \bf $M(\mc{L}_c IIA)$ & \bf $M(\mc{L}_c IIB)$ 
\\ 
\hline 
\hline 
\rowcolor{lightgray}
 Generators & $h_7$ & $h_7$, $\omega_7$\\
 \multirow{2}{2.2cm}{Differentials} & \multirow{2}{5.2cm}{$dh_7 = -\tfrac{1}{2} \omega_4^2 + \omega_6 \cdot \omega_2 + s h_7 \cdot \tilde{c}_2$} & $d h_7 = \omega_3 \cdot \omega_5 + \omega_7 \cdot \omega_1 - s h_7 \cdot c_2$,\\
 & & $d\omega_7 = h_3 \cdot \omega_5 - \omega_6 \cdot c_2$\\ 
 \hline 
\end{tabular}
\end{center}
\end{prop}
With relations imposed algebraically, this means
that
\begin{equation}
\label{T-duality}
M(\mc{L}_c IIA)/(h_7, dh_7) \cong M(\mc{L}_c IIB)/(h_7, \omega_7, d h_7, d \omega_7)\,,
\end{equation}
that is to say, the two DGCAs, $M(\mc{L}_c IIA)$ and $M(\mc{L}_c IIB)$, are isomorphic modulo the differential ideals generated by the generators of degree 7. Perhaps, topologically more interesting is an isomorphism between the dg-subalgebras $M_A \subset M(\mc{L}_c IIA)$ and $M_B \subset M(\mc{L}_c IIB)$ generated by all the generators except those of degree 7:
\begin{equation}
\label{T-duality-sub}
\xymatrix{
{}_{\phantom{A}_{\phantom{A}}}{M_A} \;\; \ar@{^{(}->}[d] \ar[r]^\sim &
\;\;\; M_B  {}_{\phantom{A}_{\phantom{A}}} \ar@{^{(}->}[d]\\
M(\mc{L}_c IIA) & M(\mc{L}_c IIB)\, .
}
\end{equation}
These dg-subalgebras are minimal Sullivan, and their isomorphism morally corresponds to the rational equivalence over $\RR$ of quotient spaces of $\mc{L}_c IIA$ and $\mc{L}_c IIB$.

\smallskip 
Another way to evade the above mismatch is to introduce ``stable'' models of types IIA and IIB, which turn out to be perfectly compatible with T-duality, at the expense of using spectra in lieu of spaces.
See details in \cite{FSS-pbranes}.
% pushing the mismatch off to infinity in a certain sense (which is not an issue for finite-dimensional spacetime).
We plan to give a physical interpretation of the discrepancy in an upcoming paper \cite{SV:M-theory}.

%%%%%%%%%%%%%%%%%%%%%%%%%%%%%%%%%%%%%%%%% 
\section{Toroidal symmetries}
\label{Sec-TorSymm} 
%%%%%%%%%%%%%%%%%%%%%%%%%%%%%%%%%%%%%%%%%
Here we describe the toroidal symmetry of the iterated cyclic loop spaces
$\mc{L}^{k}_c S^4$ for $k \ge 0$ 
from \cref{Sullmm}, for which 
we provide topological interpretation.

%%%%%%%%%%%%%%%%%%%%%%%%%%%%%%%%%%%%%%%
\subsection{Toroidal symmetries of minimal algebras}
\label{tor-symm}
%%%%%%%%%%%%%%%%%%%%%%%%%%%%%%%%%%%%%%%

We will be interested in real toroidal symmetries of minimal Sullivan
algebras $M$, i.e., diagonalizable actions $T \to \Aut M$ of a
\emph{real split torus}, an affine algebraic group $T$ isomorphic over
$\RR$ to the group $\GG_m^k$ for some $k \ge 1$, where
$\GG_m = \GL(1)$ is the multiplicative group.
Here $\Aut M$ is the group of automorphisms of $M$ as a DGCA. When $M$
has strong finite type, $\Aut M$ is an affine algebraic group over
$\RR$, because it is defined by the invertibility of the Jacobian
condition in the affine $\RR$-variety $\End M$; see
\cite{Su77}\cite{renner-thesis}. Thus, by an \emph{action} $T \to \Aut M$ above, we mean
a morphism of algebraic groups defined over $\RR$.

\medskip 
We will consolidate essentially all toroidal symmetries of $M$, which
could be done by considering a \emph{maximal $\RR$-split torus} $T
\subseteq \Aut M$. In our study, a maximal split torus will play a
role similar to that of a maximal torus in the theory of compact Lie
groups or that of a Cartan subalgebra in the theory of complex
semisimple Lie algebras. Indeed, a real split torus $T \subseteq \Aut M$ gives rise to a
\emph{weight decomposition}:
\begin{equation}
    \label{weight-d}
M = \bigoplus_{\alpha \in \mathfrak{X}(T)} M_\alpha
\end{equation}
indexed by the \emph{character group} $\mathfrak{X}(T) = \Hom_\RR (T, \GG_m)$ of
real algebraic group homomorphisms from $T$ to the multiplicative group
$\GG_m$, so that $T$ acts on each \emph{weight space} $M_\alpha$ by
the character $\alpha$:
\[
M_\alpha = \big\{m \in M \; | \; t \cdot m = \alpha(t) m \quad \text{for
  all $t \in T$}\big\}.
\]
If $T$ is a real split torus of dimension $n \ge 0$, then $\mathfrak{X}(T) \cong
\ZZ^n$; see \cite[Corollary 8.2 and Proposition 8.5]{borel}. Since
automorphisms of $M$ have to respect the DGCA structure, i.e.,
the $\ZZ$-grading, differential, and multiplication on $M$, the weight
decomposition is automatically compatible with it. That is, we have
\vspace{-1mm} 
  \begin{enumerate}[{\bf (1)}]
  \setlength\itemsep{-1pt}
  \item $M^n = \bigoplus_{\alpha \in \mathfrak{X}(T)} (M_\alpha \cap M^n)$ for
    each $n \ge 0$;
  \item $d(M_\alpha ) \subseteq M_\alpha$ for each $\alpha \in \mathfrak{X}(T)$;
    \item $M_\alpha \cdot M_\beta \subseteq M_{\alpha+\beta}$ for all
      $\alpha, \beta \in \mathfrak{X}(T)$.
    \end{enumerate}
    
\vspace{-1mm}     
\noindent Here and henceforth we write the group law in the character group
$\mathfrak{X}(T)$ additively and employ the exponential notation:
\[
t^\alpha := \alpha(t) \qquad \text{for $t \in T$, $\alpha \in \mathfrak{X}(T)$.}
\]
Specifically for minimal Sullivan algebras, weight decompositions have
been considered by L.~Renner in his Master's thesis
\cite{renner-thesis} on automorphism groups of minimal algebras.

\medskip 
Maximal split tori are unique up to conjugation in a real agebraic group; see
\cite[Theorem 15.14]{borel}. This implies that the weight
decompositions corresponding to different maximal split tori are related by automorphisms of $M$.

\medskip
Given an abstract $\RR$-split torus $T$, a weight decomposition \eqref{weight-d} defines an obvious action of $T$ on $M$.

\medskip 
A weight decomposition of a minimal Sullivan algebra $M = S(V)$
induces a weight decomposition 
$$
V = \bigoplus_{\alpha \in \mathfrak{X}(T)}
V_\alpha
$$ 

\vspace{-3mm} 
\noindent 
of the generating space $V$, because the latter may be canonically
identified with the \emph{space of indecomposables}: $I(M) = M^+/(M^+)^2 = S^{\ge 1}(V)/ S^{\ge 2}(V)$, and $M^+$ and $(M^+)^2$ are split
into weight spaces by definition, whereas a diagonalizable action
diagonalizes on invariant subquotients;
% the correspondence $M \mapsto V$ is functorial, as 
cf.\ \cite[Proposition 3.4.2]{renner-thesis}. This does not mean that $V$ is necessarily a $T$-invariant subspace of $M$. For that matter, we will distinguish the generating space $V \subset S(V) = M$ and the subquotient $I(M)$ of the indecomposables of $M$.

%%%%%%%%%%%%%%%%%%%%%%%%%%%%%%%%%
\subsection{Toroidal symmetries of a cyclification}
\label{symm-cycl}
%%%%%%%%%%%%%%%%%%%%%%%%%%%%%%%%%

Here we apply the results of \Cref{tor-symm} to a
``cyclification'', i.e., to the cyclic loop space $\mc{L}_c Z$ of a space $Z$
considered in \Cref{cyclification}.
Let us start with a few observations.
\begin{prop}
[Split tori of a Sullivan model]
\label{kerd}
Suppose an $\RR$-split torus $T$ acts on a minimal Sullivan algebra $M = (S(V), d)$ of strong finite type. 
\vspace{-2mm} 
\begin{enumerate}[{\bf (i)}]
   \setlength\itemsep{-2pt}
    \item
Then the weights of the action, that is to say, the characters $\alpha \in \mathfrak{X}(T)$ for which the weight space $M_\alpha$, see \eqref{weight-d}, is nontrivial, are generated multiplicatively on
\[
P(M) := Z(M^+)/Z(M^+) \cap (M^+)^2 \subseteq I(M),
\]
where $Z(M^+) := \ker d|_{M^+} = \{ x \in M^+ \; | \; dx = 0 \}$.

\item The action of an $\RR$-split torus on $M$ is determined by its action on $P(M)$. In other words, if $T \subseteq \Aut M$ is an $\RR$-split torus, then the composition $T \subseteq \Aut M \to \GL(P(M))$ is injective. In particular, $\dim T \le \dim P(M)$.
% \label{previous}

\item
\label{maxdim}
If $T$ is an $\RR$-split torus of $\Aut M$ of dimension  $\dim T = \dim P(M)$, then $T$ is a maximal split torus of $\Aut M$.

\item
\label{quadr}
If the differential on $M = S(V)$ is quadratic, then the natural map $\ker d \cap V = Z(M^+) \cap V \to P(M)$ is an isomorphism.
\end{enumerate}
\end{prop}

\begin{proof}
{\bf (i)} This part generalizes Proposition 3.4.2 of Renner's thesis \cite{renner-thesis}, and so does the proof. However, our generalization to the non-simply connected and non-algebraically closed case requires new ideas and more work.

We will need to use a more invariant version of the Sullivan nilpotence condition \Cref{SMM}(i), due to Bousfield and Gugenheim \cite{Bousfield-Gugenheim}. Every minimal Sullivan algebra $M$ admits a canonically defined double filtration
\begin{gather*}
M(0) \subseteq M(1) \subseteq M(2) \subseteq \dots,\\
M(n,0) \subseteq M(n,1) \subseteq M(n,2) \subseteq \dots,
\end{gather*}
$n \ge 1$, by dg-subalgebras, such that $M(0) = \RR$, $\bigcup_n M(n) = M$, and
\begin{equation}
\label{Mnp}
\bigcup_p M(n,p) = M(n) \qquad \text{for each $n \ge 1$}.
\end{equation}
The subalgebra $M(n)$ is defined as the dg-subalgebra generated by $M^i$ for $1 \le i \le n$. The subalgebra $M(n,p)$ is defined inductively, starting from $M(n,0) = M(n-1)$, as the subalgebra generated by $M(n,p-1)$ and the elements $x \in M^n$ such that $dx \in M(n,p-1)$. Conversely, every connected DGCA $M$ which is free as a graded commutative algebra and for which the subalgebras defined above satisfy \eqref{Mnp} is a minimal Sullivan algebra, see \cite[Proof of Proposition 7.5]{Bousfield-Gugenheim}. Bousfield and Gugenheim simply call connected, free as graded commutative algebras DGCAs satisfying \eqref{Mnp} minimal. The equivalence of the two types of models, the minimal Sullivan model and Bousfield-Gugenheim's minimal model follows from the existence and uniqueness theorems for each type, see \cite{Su77} and \cite{Bousfield-Gugenheim}, respectively.

The augmentation ideal $M^+ \subset M$ and its powers $(M^+)^n = S^{\ge n} (V)$ are $T$-invariant. Moreover, the weight decomposition of $M$ maps isomorphically to the weight decomposition of the associated graded algebra $\gr M := \bigoplus_{n \ge 0} (M^+)^n/(M^+)^{n+1} = S(I(M))$. The subquotient $I(M)$ of $M$ is $T$-invariant and inherits a weight decomposition. Thus, the weights of $M$ are generated multiplicatively by the weights of $I(M)$.
% \footnote{However, it is not generally true that every weight decomposition of $I(M)$ induces a weight decomposition of $M$, as it has to be compatible with the differential.}
Similarly, $P(M)$ is a $T$-invariant subquotient of $M$ and inherits a weight decomposition. Note that the subalgebras $M(n)$ and $M(n,p)$, defined intrinsically by using the multiplicative and dg-structures on $M$, are $T$-invariant and get a weight decomposition.

We will show that the weights of the subalgebras $M(n,p)$ are multiplicatively generated by the weights of $P(M)$ and $M(n,p-1)$ and run a double induction on $n$ and $p$. By definition, the subalgebra $M(n,p)$ is generated by $M(n,p-1)$ and $x \in M^n$ such that $dx \in M(n,p-1)$. Its subspace $Z(M^n) \subseteq M^n$ of $n$-cycles fits into a short exact sequence of $T$-invariant spaces:
\[
0 \longrightarrow Z(M^n) \cap (M^+)^2 \longrightarrow Z(M^n) \longrightarrow P(M^n) \longrightarrow 0\,.
\]
% where $P(M^n) \subseteq P(M)$ is defined by this sequence.
Observe that $Z(M^n) \cap (M^+)^2 \subseteq M(n-1) = M(n,0)$, because the component of degree $n$ of $(M^+)^2$ is spanned by products of elements in $M^i$ for $1 \le i \le n-1$.
On the other hand, the weights of a $T$-invariant complement $C(M^n)$ to $Z(M^n)$ in the space $\{x \in M^n \; | \;  dx \in M(n,p-1)\}$ are among the weights of $M(n,p-1)$, because for $x \ne 0 \in C(M^n)$, $dx$ will be nonzero and have the same weight.

\vspace{1mm} 
\noindent {\bf (ii)}  From Part (i), we can deduce that if an $\RR$-split torus $T$ acts faithfully on $M$, that is to say, $T \subseteq \Aut M$, then it would also act faithfully on $P(M)$. Indeed, if there is a $t \in T$ acting on $P(M)$ trivially, then by Part 1, for all weights $\alpha$ of $M$, we have $\alpha(t) =0$, which means $t$ acts trivially on $M$. Thus, $T$ embeds into the general linear group $\GL(P(M))$ of $P(M)$ regarded as a vector space. Since the maximal torus of $\GL(P(M))$ has dimension equal to $\dim P(M)$, we conclude that $\dim T \le \dim P(M)$.

\vspace{1mm} 
\noindent 
{\bf (iii)} Follows
% easily
from (ii).

\vspace{1mm} 
\noindent 
{\bf (iv)}  Note that the composition $V \subseteq M^+ \to I(M)$ is injective. Therefore, $\ker d|_V = Z(M^+) \cap V$ maps injectively 
to $P(M)$. To show that this map is surjective, let $\bar x \in P(M)$ and $x \in M^+ = S^{\ge 1} (V)$, $dx = 0$, represent $\bar x$. 
Decompose $x = x_1 + x_2 + \dots$, with $x_i \in S^i(V)$. Then $dx = 0$ implies $dx_1 = dx_2 = \dots = 0$, because when the differential 
$d$ is quadratic, $dx_i \in S^{i+1} (V)$. Therefore, $x_1 \in V$ also represents $\bar x$.
\end{proof}

Within the context of 
\Cref{Sullmm}, we now establish the following results. For the rest of this section, let $Z$ be a path-connected and nilpotent space and 
$M(Z) = (S(V), d)$ its real Sullivan minimal model, which we assume to have strong finite type. Consider the real Sullivan minimal model 
$M(\mc{L}_c Z) = (S(V \oplus V[1] \oplus \RR w), d_c)$ of the cyclic loop space $\mc{L}_c Z$ of $Z$. 

\begin{lemma}[The growth of dimension of the space of closed generators]
\label{fund.lemma}
If the differential $d$ on $M(Z)$ is quadratic, then
\[
\dim \left( \ker d_c \cap (V \oplus V[1] \oplus \RR w) \right) = \dim (\ker d \cap V) + 1.
\]
\end{lemma}

\begin{proof}
%[Proof of Lemma]
We will start with a simpler case of simply connected $Z$ and make necessary adjustments in the more general case.

\noindent \textbf{\underline{Simpler Case}: $Z$ is simply connected}. At the level of minimal Sullivan model $M = M(Z) = S(V)$, 
this means that $V^1 = 0$.

We claim that
\begin{equation}
  \label{claim}
\ker d_c \cap (V \oplus V[1] \oplus \RR w) = (\ker d \cap V)[1] \oplus
\RR w.
\end{equation}
Indeed, from formulas \eqref{dc}, we see that no nonzero
element of $V = V^{> 1}$ could be in $\ker d_c$, whereas $w$ is always in $\ker
d_c$. We claim that $\ker d_c \cap V[1] = (\ker d \cap V)[1]$. Indeed, to justify the inclusion
$(\ker d \cap V)[1] \subseteq \ker d_c \cap V[1]$, we start with a
$d$-closed element $v \in V \subseteq M(Z)$ and observe that $d_c sv =
- s d v = 0$, again from \eqref{dc}.

Let us prove the opposite inclusion: $\ker d_c \cap V[1] \subseteq
(\ker d \cap V)[1]$. If $d_c sv =0$ for some $v \in V$, then $s dv
=0$. Note that $dv$ is in $S(V)$, which may be thought of as the
algebra of polynomial functions on the graded manifold $\mathbb{V}^* =
\Spec S(V)$ or, passing to the grading modulo 2, on the supermanifold
$\mathbb{V}^*$, the affine superspace associated with the super vector
space $V^*$. The minimality condition on $d$ moreover implies that $dv \in
S^{\ge 2}(V)$. Note also that the differential $s$ may be identified
with the de Rham differential of the affine superspace $\mathbb{V}^*$:
$$
s = d_{\dR}: \Omega^0(\mathbb{V}^*) = S(V) \; \longrightarrow \;
\Omega^1(\mathbb{V}^*) = S(V) \otimes V[1]\;.
$$ 
(In $\ZZ$-graded
geometry, one usually has $V[-1]$ for the cotangent space, but it is
the same as $V[1]$ under mod 2 grading.)  So, by the super
Poincar\'{e} lemma \cite[Proposition 3.4.5]{manin:gauge}, $\ker s =
S^0(V) = \RR$. Therefore, if $sdv =0$ then $dv \in S^{\ge 2}(V) \cap
S^0(V) = 0$, i.e., $v \in \ker d \cap V$. We conclude that
$\ker d_c \cap V[1] = (\ker d \cap V)[1]$. Summing up all the parts of
$V \oplus V[1] \oplus \RR w$, we see that \eqref{claim} holds.

\medskip 
\noindent \textbf{\underline{General Case}: $Z$ is not necessarily simply connected}. We will adjust formula \eqref{claim}. Now we claim that
\[
\ker d_c \cap (V \oplus V[1] \oplus \RR w) = (\ker d \cap V^1) \oplus
(\ker d \cap V^{\ge 2})[1] \oplus \RR w.
\]

So, it is again clear from \eqref{dc} that $w \in \ker d_c$. Given that $sv
= 0 $ for any $v \in V^1$, it is also clear that we have 
$\ker d_c \cap V = \ker d \cap V^1$. What remains to be proven is that $\ker d_c \cap
V[1] = (\ker d \cap V^{\ge 2})[1]$.

Let us start with showing $(\ker d \cap V^{\ge 2})[1] \subseteq \ker
d_c \cap V[1]$. If $v \in V^{\ge 2}$ is such that $dv = 0$, then $d_c
s v = - sdv = 0$.

For the opposite inclusion, $ \ker d_c \cap V[1] \subseteq (\ker d
\cap V^{\ge 2})[1]$, suppose that $d_c sv = 0$ for some $v \in V^{\ge
  2}$. (We ignore $v \in V^1$, as in this case $sv=0$.) The second
formula \eqref{dc} then implies $sdv=0$. Note that $dv \in S^{\ge
  2}(V)$ because of the minimality of $(S(V), d)$. Another observation
is that since $sv =0$ for any $v \in V^1$, the differential $s: S(V)
\to S(V) \otimes V^{\ge 2}[1]$ may now be identified with the relative
de Rham differential of the relative affine superspace $\mathbb{V}^* =
(\mathbb{V}^1)^* \times (\mathbb{V}^{\ge 2})^*$ over
$(\mathbb{V}^1)^*$:
\vspace{-2mm} 
\[
s = d_{\dR}: \Omega^0_{\mathbb{V}^*/(\mathbb{V}^1)^*}(\mathbb{V}^*) =
S(V) \; \longrightarrow \;
\Omega^1_{\mathbb{V}^*/(\mathbb{V}^1)^*}(\mathbb{V}^*) = S(V)
\otimes V^{\ge 2}[1].
\]

\vspace{-2mm} 
\noindent As in Case 1, the super Poincar\'{e} lemma implies that $\ker s = S(V^1)$.

Getting back to our $v \in V^{\ge 2}$ such that $d_c sv =0$, we see
that $dv \in S^{\ge 2}(V) \cap S(V^1) = S^{\ge 2} (V^1)$. Since the
degree of $v$ is at least two, the degree of $dv$ must be at least
three. This means that $dv \in S^{\ge 3} (V^1)$. But the differential
is assumed to be quadratic: $dv \in S^2(V)$, whence $dv =0$. Lemma is proven.
\end{proof}

\begin{prop}[Automorphisms of the Sullivan model of a cyclification]
\label{AutL}
% Let $Z$ be a path-connected and nilpotent space and $M(Z) = (S(V), d)$ its real Sullivan minimal model, which we assume to have strong finite type. Consider the real Sullivan minimal model $M(\mc{L}_c Z) = (S(V \oplus V[1] \oplus \RR w), d_c)$ of the cyclic loop space $\mc{L}_c Z$ of $Z$.
The automorphisms of $M(Z)$ extend naturally to automorphisms of $M(\mc{L}_c Z)$. Moreover, one has a natural inclusion
  \begin{equation}
      \label{inclusion}
  \Aut M(Z) \times \GG_m \subseteq \Aut M(\mc{L}_c Z).
 \end{equation}
\end{prop}

\begin{proof}
In general, knowing that the group of automorphisms of the Sullivan minimal model of a space is isomorphic to the group of rational homotopy self-equivalences thereof and the functoriality of the construction $Z \mapsto \mc{L}_c Z$, we get a morphism $\Aut M(Z) \to \Aut M(\mc{L}_c Z)$. This morphism is injective, because if $g \in \Aut M(Z)$ acts trivially on $M(\mc{L}_c Z)$, then it will act trivially on every subquotient, including $M(Z)$, see \eqref{ext:cyc} and \eqref{ext:free}.

It will be useful to have explicit formulas for this extension, $\rho: \Aut M(Z) \hookrightarrow \Aut M(\mc{L}_c Z)$. Suppose that $g: M(Z) \to M(Z)$ is an automorphism of $M(Z) = S(V)$. Then $g$ defines an automorphism $\rho(g)$ of $M(\mc{L}_c Z) = S(V \oplus V[1] \oplus \RR w)$ by acting on the free generators $V \oplus V[1] \oplus \RR w$ as follows:
\begin{align*}
    \rho(g) v & := gv \qquad & \text{for $v \in V$},\\
    \rho(g) (sv) & := s(gv) \qquad & \text{for $sv \in V[1]$},\\
    \rho(g) w & := w.
\end{align*}
Now, for $t \in \GG_m$, define an action
\begin{align}
\nonumber
    t \cdot v & := v \qquad & \text{for $v \in V$},\\
\label{inverse}
    t \cdot (sv) & := t (sv) \qquad & \text{for $sv \in V[1]$},\\
\nonumber
    t \cdot w & := t^{-1} w.
\end{align}
These formulas define an explicit inclusion \eqref{inclusion}.
\end{proof}

\begin{theorem}[Maximal split tori of the Sullivan model of a  cyclification] 
 \label{step}
% Let $Z$ be a path-connected, nilpotent space and $M(Z)$ its real Sullivan minimal model, assumed to be of strong finite type.
Let us also assume that the differential in the Sullivan minimal model $M(Z) = (S(V),d)$ is \emph{quadratic, i.e.,} the restriction of $d$ to $V$ maps $V$ to $S^2(V)$:
$
d|_V: V \to S^2(V).
$
Suppose there is an $\RR$-split torus $T \subseteq \Aut M(Z)$ such that $\dim T = \dim P(M(Z))$, see $\Cref{kerd}$.

Then $T$ is a maximal split torus of $\Aut M(Z)$, $\dim (T \times \GG_m) = \dim P(M(\mc{L}_c Z))$ and $T \times \GG_m$ is a maximal split torus of $\Aut M(\mc{L}_c Z)$.
\end{theorem}

\begin{proof}
By \Cref{kerd}(\ref{maxdim}), $T$ is a maximal split torus of $\Aut M(Z)$. \Cref{AutL} implies that $T \times \GG_m$ is a split torus in $\Aut M(\mc{L}_c Z)$. To prove that it is maximal, using the \Cref{kerd}(\ref{maxdim}) once again, it is enough to show that $\dim P(M(\mc{L}_c Z)) = \dim P(M(Z)) + 1$, which is equal to $\dim (T \times \GG_m)$ by assumption. Applying \Cref{kerd}(\ref{quadr}), we see that $\ker d \cap V \cong P(M(Z))$. Formulas \eqref{dc} for the differential $d_c$ of $M(\mc{L}_c Z)$ show that the differential is also quadratic. Therefore, we have $\ker d_c \cap (V \oplus V[1] \oplus \RR w) \cong P(M(\mc{L}_c Z))$, and conclude with \Cref{fund.lemma}.
\end{proof}

Given the  decomposition of the Sullivan minimal model into weight spaces of the previous section, \cref{tor-symm}, the toroidal symmetries of the theorem yield the following statement.

\begin{cor}[Weight decomposition in Sullivan models]
  \label{corr}
If $M(Z) =
\bigoplus_{\alpha \in \mathfrak{X}(T)} M_\alpha$ is the weight decomposition
corresponding to the action of a maximal split torus $T$, which
induces a weight decomposition $I(M(Z)) = \bigoplus_{\alpha \in \mathfrak{X}(T)}
I(M(Z))_\alpha$ on the space $I(M(Z))$ of indecomposables, then the weight
decomposition of $M(\mc{L}_c Z)$ corresponding to its maximal split
torus $T_c = T \times \GG_m$ induces the following weight
decomposition on its space $I (M(\mc{L}_c Z)) = I(M(Z)) \oplus I(M(Z))[1] \oplus \RR w$ of indecomposables:
\vspace{-2mm} 
\begin{enumerate}[{\bf (i)}]
   \setlength\itemsep{-1pt}
    \item 
  The weight of $w \in M(\mc{L}_c Z)$ is $\epsilon_1 = (0,-1) \in
    \mathfrak{X}(T_c) = \mathfrak{X}(T) \times \mathfrak{X}(\GG_m)$;
  \item If $v \in (I(M(Z))_\alpha$ for some weight $\alpha \in
    \mathfrak{X}(T)$, then the image $v$ in $I(M(Z)) \subseteq I(M(\mc{L}_c Z))$ has
    weight $\alpha := (\alpha, 0) \in \mathfrak{X}(T) \times \mathfrak{X}(\GG_m)$;
  \item If $v \in (I(M(Z))_\alpha$ for some weight $\alpha \in
    \mathfrak{X}(T)$, then the image $sv$ in $(I(M(Z))[1] \subseteq I(M(\mc{L}_c Z))$ has
    weight $\alpha - \epsilon_1 \in \mathfrak{X}(T) \times \mathfrak{X}(\GG_m)$.
\end{enumerate}
\end{cor}
%%%%%%%%%%%%%%%%%%
\subsection{Toroidal symmetries of the minimal algebras of cyclifications of $S^4$}
\label{symm-cycl-concrete}
%%%%%%%%%%%%%%%%%%%%%%%%

Here we apply the results of the previous section to the
``cyclifications,'' the iterated cyclic loop spaces $\mc{L}^k_c S^4$
of the 4-sphere $S^4$, $k \ge 0$, from \cref{cyclification}.
% and \cref{Sullmm}.
The resulting symmetries of the Sullivan minimal model hold universally for fields of $k$-dimensional reductions of M-theory and may be interpreted as trombone and torus rescaling symmetries discussed in \cite{CDF}\cite{CLPS} and \cite{DS}.
% ; see more on this in \cref{natural0}.
Let us start with an immediate consequence of \Cref{step}.

\begin{cor}[Toroidal symmetries]
\label{split-rank}
The maximal $\RR$-split torus of the real algebraic group $\Aut M (\mc{L}_c^k S^4)$ for $k \ge 0$ is $T^{k+1}$, isomorphic to $\GG_m^{k+1}$ over $\RR$. The structure of $\mc{L}_c^k S^4$ as an iterated cyclification $($or the structure of $M(\mc{L}_c^k S^4)$ as a sequence of Halperin extensions \eqref{ext:cyc} and \eqref{ext:free}$)$ determines a canonical splitting $T^{k+1} \cong \GG_m^{k+1}$.
\end{cor}

\begin{proof}
The statement follows from \Cref{step} by induction. The base case $k=0$ is done in \Cref{k=0} below.
\end{proof}

\begin{example}[$k=0$: The 4-sphere $S^4$]
\label{k=0}

We start with the automorphism group $\Aut M(S^4)$ of the Sullivan minimal model $M(S^4)$
of $S^4$; see \eqref{S^4}. By the degree argument,
an automorphism of $M(S^4)$ must take $g_4$ to a scalar multiple of itself:

\vspace{-3mm} 
\begin{equation}
  \label{rg_4}
g_4 \mapsto tg_4
\end{equation}

\vspace{-1mm} 
\noindent for some $t \in \GG_m(\RR)$, and this determines the action of the
automorphism on $g_7$:
\vspace{-1mm} 
\begin{equation*}
%  \label{rg_7}
  g_7 \mapsto t^2 g_7,
\end{equation*}

\vspace{-1mm} 
\noindent
and thereby on the whole DGCA $M(S^4)$. This gives an identification
$\Aut M(S^4) \cong \GG_m$ over $\RR$, and, therefore, $\Aut M(S^4)$
automatically coincides with its maximal split torus $T$ and $\dim T = \dim P(M(S^4)) = 1$, as $P(M(S^4)) \cong \ker d \cap (\RR g_4 \oplus \RR g_7) = \RR g_4$, see \Cref{kerd}(\ref{quadr}). Note that this
identification is unique up to automorphism of $\GG_m$, which could
only be $t \mapsto t^{-1}$ if not trivial. Thus, we get a weight
decomposition determined by
\begin{equation*}
g_4 \in M(S^4)_{\epsilon_0}, \qquad g_7 \in M(S^4)_{2 \epsilon_0},
\end{equation*}
with the weights defined up to common sign, that is to say,
\vspace{-2mm} 
\begin{equation}
  \label{S4weights}
t^{\epsilon_0} = t
\end{equation}

\vspace{-2mm} 
\noindent
or $t^{\epsilon_0} = t^{-1}$.
(Again, the weights are just determined by the weight of $g_4$, as
that is the only $d$-closed generator).  We can always normalize this
ambiguity so as to have positive weights and assume \eqref{S4weights} is valid. This choice also has 
topological motivation, as we will see in \cref{natural0}.

\end{example}

%\newpage 

\begin{example}[$k=1$: The cyclification $\mc{L}_c
S^4$ of the sphere $S^4$]
\label{k=1}

We now look at the maximal real split torus of $\Aut M(\mc{L}_c
S^4)$ compatible with the structure of $M(\mc{L}_c S^4)$ as the
Sullivan minimal model of the cyclic loop space of $S^4$; see
Example \ref{example1}. 
\Cref{split-rank} canonically identifies the
maximal split torus of $\Aut M(\mc{L}_c S^4)$ as $\GG_m \times
\GG_m$, acting on $M(\mc{L}_c S^4) = \RR[g_4, g_7, sg_4, sg_7,
  w]$ as follows:
\begin{gather}
  \label{tg_4}
  t \cdot g_4 = t^{\epsilon_0} g_4, \qquad   t \cdot g_7 = t^{2\epsilon_0} g_7,\\
  \nonumber
  t \cdot sg_4 = t^{\epsilon_0 - \epsilon_1} sg_4,  \qquad
  t \cdot sg_7 = t^{2\epsilon_0 - \epsilon_1} sg_7,\\
\nonumber
t \cdot w =  t^{\epsilon_1} w,
  \end{gather}
where $t \in (\GG_m \times \GG_m) (\RR)$. This corresponds to a weight
decomposition determined by
\begin{gather*}
g_4 \in M(\mc{L}_c S^4)_{\epsilon_0}, \qquad g_7 \in M(\mc{L}_c S^4)_{2
  \epsilon_0},\\ sg_4 \in M(\mc{L}_c S^4)_{\epsilon_0 - \epsilon_1}, \qquad
sg_7 \in M(\mc{L}_c S^4)_{2 \epsilon_0 - \epsilon_1},\\
w \in M(\mc{L}_c S^4)_{\epsilon_1},
  \end{gather*}
as per Corollary \ref{corr}.
\end{example}

\begin{example}[$k=2$: The double cyclification $\mc{L}^2_c
  S^4$]
  \label{k=2}
We now consider Example \ref{example2}. 
Again, iterating the application of Theorem \ref{step}, the maximal
torus of $\Aut M(\mc{L}^2_c S^4)$ is identified canonically with the
product $\GG_m^3 = (\GG_m \times \GG_m ) \times \GG_m$, where the
first factor $\GG_m \times \GG_m$ refers to the maximal torus of $\Aut
M(\mc{L}_c S^4)$ identified in the previous example. Continuing the
use of notation of Corollary \ref{corr}, we obtain a weight
decomposition of $M(\mc{L}^2_c S^4)$, which is determined on its
generators (see Example \ref{example2}) as follows:
\begin{gather}
  \nonumber
  g_4 \in M_{\epsilon_0}, \qquad g_7 \in M_{2
  \epsilon_0},\\
  \nonumber
s_1 g_4 \in M_{\epsilon_0 - \epsilon_1}, \qquad s_1 g_7 \in M_{2 \epsilon_0 - \epsilon_1},\\
\label{w1}
w_1 \in M_{\epsilon_1},\\
  \nonumber
s_2 g_4 \in M_{\epsilon_0 - \epsilon_2}, \qquad s_2 g_7 \in M_{2 \epsilon_0 - \epsilon_2},\\
  \nonumber
s_2 s_1 g_4 \in M_{\epsilon_0 - \epsilon_1 - \epsilon_2}, \qquad s_2 s_1 g_7 \in M_{2 \epsilon_0 - \epsilon_1 - \epsilon_2},\\
  \nonumber
s_2 w_1 \in M_{\epsilon_1 -\epsilon_2}, \qquad w_2 \in M_{\epsilon_2},
  \end{gather}
where, for brevity, we have been writing $M$ for $M(\mc{L}^2_c S^4)$.
\end{example}

\begin{example}[$k=3$: The triple cyclification $\mc{L}^3_c
  S^4$]
  \label{k=3}
Now we consider Example \ref{example3}.
Again the maximal torus of $\Aut M(\mc{L}^3_c S^4)$ splits canonically to
become $\GG_m^4 = (\GG_m^3) \times \GG_m$, where the first factor
comes from the double cyclification. The resulting weight decompositon
of $M = M(\mc{L}^3_c S^4)$ repeats the formulas \eqref{w1} verbatim
for the weights of those generators which are the generators of
$M(\mc{L}^2_c S^4)$. The weight of $w_3$ is $\epsilon_3$. For the
weights of generators of the type $s_3 g$, where $g$ is a generator on
the list \eqref{w1}, the formulas are the same as \eqref{w1}, except
that weight $\epsilon_3$ gets subtracted, e.g.,
\begin{gather*}
  s_3 g_4 \in M_{\epsilon_0 - \epsilon_3}, \qquad s_3 w_2 \in M_{\epsilon_2-\epsilon_3}.
  \end{gather*}
The weight $\epsilon_1 - \epsilon_2 -\epsilon_3$ will not be present, as
$s_3 s_2 w_1$ gets truncated to zero.
% See also Remark \ref{Rem-scalars}. 
\end{example}

\begin{example}[$k\ge 3$: The $k$-fold cyclification $\mc{L}^k_c
  S^4$]
  \label{kgeq3}

Let us say a few words on the general pattern we see for $k \ge
3$. All the weights of the generating space $V$ for $S(V) =
M(\mc{L}^k_c S^4)$ will be of the form
\begin{align*} 
\epsilon_0 - & \sum_{j=1}^l \epsilon_{i_j}, \qquad \text{where $0 \le l \le
  3$ and $1 \le i_1 < \dots < i_l \le k$},
\\
 2\epsilon_0 - & \sum_{j=1}^l \epsilon_{i_j}, \qquad \text{where $0 \le l \le
  6$ and $1 \le i_1 < \dots < i_l \le k $},
\\
 & \epsilon_i, \; \qquad \qquad 1 \le i \le k,
\\
 \epsilon_i & -  \epsilon_j, \qquad \quad  1 \le i < j \le k.
\end{align*}
Each of the corresponding weight spaces in $V$ will be
one-dimensional; see Example \ref{examplek}.
\end{example}

%%%%%%%%%%%%%%%%%%%%%%%%%%%%%%%%%%%%%%%%
\subsection{Topological interpretation of toroidal symmetries}
\label{Sec-toptor}
%%%%%%%%%%%%%%%%%%%%%%%%%%%%%%%%%%%%%%%

In this section, we interpret the toroidal symmetries \eqref{inverse}
and \eqref{rg_4} as rational (real) homotopy equivalences. The idea of
doing that was suggested by A.~Bondal. These symmetries have physical
interpretation of trombone and torus rescaling symmetries, mentioned in 
\cref{symm-cycl-concrete}.
 
%%%%%%%%%%%%%%%%%%%%%%%%%%%%%
\subsubsection{Toroidal symmetries coming from $S^4$}
 \label{natural0}
%%%%%%%%%%%%%%%%%%%%%%%%%%%%%

Let us start with the symmetry \eqref{rg_4}. Given an integer $n \in \ZZ$,
define $\varphi_0(n): S^4 \to S^4$ to be any continuous map of
degree $n$. Such a map induces a homomorphism, given by multiplication by $n$:
\begin{align*}
  \varphi_0(n)_*: \pi_4(S^4)  & \longrightarrow \pi_4(S^4)\\[-2pt]
%\varphi_0(n)_*: H_4(S^4; \ZZ) & \to H_4(S^4; \ZZ),\\
x & \longmapsto n x
\end{align*}

\vspace{-2mm}
\noindent on the degree-four homotopy group $\pi_4(S^4)$ of $S^4$.  It also
induces the identity morphism on $\pi_0(S^4)$.
% and $H_0(S^4; \ZZ)$,
Recall that $\pi_7(S^4) \cong \ZZ \oplus \ZZ_{12}$, where the free
part may be canonically identified with the subgroup $\ZZ y \subset
\pi_7(S^4)$, $y$ being the class of the Hopf fibration $S^7 \to
S^4$. We know how $\varphi_0(n)$ acts on $y$:
\begin{align*}
  \varphi_0(n)_*: y & \mapsto n^2 y,
\end{align*}
because the Whitehead square $[x,x]$ of the generator $x \in
\pi_4(S^4) \cong \ZZ x $, the homotopy class of $\id: S^4 \to S^4$, is
twice the generator $y$ of $\ZZ y \subset \pi_7(S^4)$ (see \cite{FSS19b} for explanation in this context):
\[
  [x,x] = 2y.
\]
When we pass to rational homotopy groups, all torsion disappears, and we
have $\pi_4 (S^4) \otimes \QQ = \QQ x$ and $\pi_7 (S^4) \otimes \QQ = \QQ y$. Since there are no other rational
homotopy groups $\pi_i (S^4) \otimes \QQ$, $i \ge 1$, we see that, when $n$ is nonzero, $\varphi_0(n):
S^4 \to S^4$ is a rational homotopy equivalence. Hence, it has an inverse
$\varphi_0(n)^{-1}: S^4 \to S^4$ in the rational homotopy category. We may
denote this morphism $\varphi_0(1/n) := \varphi_0(n)^{-1}$, given that
it acts on the generator $x \in \pi_4(S^4) \otimes \RR$ as
\vspace{-2mm} 
\[
x \mapsto \frac{1}{n} x.
\]
Composing a map $\varphi_0(p): S^4 \to S^4$ of degree $p \in \ZZ
\mathbin{\vcenter{\hbox{$\scriptscriptstyle\mathrlap{\setminus}{\hspace{.2pt}\setminus}$}}} \{0\}$ 
and a map $\varphi_0(1/q)$ for $q \in \NN$, we are
getting an automorphism $\varphi_0(p/q): S^4 \to S^4$ in the rational
homotopy category. This way we get a group homomorphism
\begin{align*}
  \QQ^\times & \longrightarrow \Aut_\QQ S^4,\\
  p/q & \longmapsto \varphi_0(p/q),
\end{align*}
where $\Aut_\QQ S^4$ stands for the automorphism group in the rational
homotopy category. Given our description of $\Aut M(S^4)$ in  Example \ref{k=0} and the
fact that the rational homotopy category (of rational, nilpotent, finite-type spaces) is equivalent (via a
contravariant functor) to the category of minimal Sullivan algebras over $\QQ$,
we see that \eqref{aut} actually defines an isomorphism
\[
  \QQ^\times \xrightarrow{\; \sim \;} \Aut_\QQ S^4 \xrightarrow{\; \sim \;} \Aut
  M(S^4)(\QQ),
  \]
where $\Aut M(S^4)(\QQ)$ is the group of rational points of the algebraic
group $\Aut M(S^4)$.
Since the action formulas are polynomial, the group isomorphism defines an isomorphism of algebraic groups over $\QQ$:
\vspace{-2mm} 
\begin{align}
  \label{aut}
  \GG_m & \longrightarrow \Aut M(S^4),\\
  \nonumber
  r & \longmapsto \varphi_0(r).
\end{align}
Via the action on the target $S^4$ of $\mc{L}^k_c
S^4$, the isomomorphism \eqref{aut} may be canonically lifted, as in \Cref{AutL}, to a $\QQ$-isomomorphism
\[
  \GG_m \xrightarrow{\; \sim \;} \Aut
  M(\mc{L}^k_c S^4),
  \]
which is exactly the action of the multiplicative group $\GG_m$ on $M(\mc{L}^k_c S^4)$ coming from the action of $\GG_m$ on
$g_4 \in M(\mc{L}^k_c S^4)$, as in \eqref{rg_4} and \eqref{tg_4}.
Summarizing, we have:

\begin{prop}[Toroidal symmetries from $S^4$]
\label{prop-torS4}
Consider the degree $n$ maps $\varphi_0(n): S^4 \to S^4$, for $n \in \ZZ\setminus \{0\}$.

\vspace{-3mm} 
  \begin{enumerate}[{\bf (i)}]
   \setlength\itemsep{-1pt}
  \item
These maps are invertible in the rational homotopy category.

\item The compositions of these maps with their inverses
 gives a group isomorphism 
$\Q^\times  \xrightarrow{\sim} \Aut
  M(S^4)(\Q)$, which defines naturally an isomorphism of algebraic groups over $\QQ$
$$
\GG_m \xrightarrow{\; \sim \;} \Aut
  M(S^4)\;.
  $$

\vspace{-2mm}
\item  This lifts canonically to an algebraic-group homomorphism
$
  \GG_m  \xrightarrow{\; \sim \;} \Aut
  M(\mc{L}^k_c S^4)
$
for $k \ge 1$, which provides an action of $\GG_m$ on the Sullivan minimal models of the cyclifications of $S^4$.
\end{enumerate}
\end{prop}

%%%%%%%%%%%%%%%%%%%%%%%
\subsubsection{Toroidal symmetries coming from $S^1$}
 \label{natural1}
 %%%%%%%%%%%%%%%%%%%%%%%

The situation with the action \eqref{inverse} is subtler. Let us consider the
general case of $\mc{L}_c Z$. For $n \in \ZZ$, the $n$-fold
winding map
\vspace{-2mm} 
\begin{align}
  \nonumber
  S^1 & \longrightarrow S^1\\[-3pt]
  \label{power}
  z & \longmapsto z^n
\end{align}

\vspace{-2mm} 
\noindent
induces a continuous map $\psi(n): \mc{L} Z \to \mc{L} Z$:
\vspace{-2mm} 
\[
\psi(n) (f) (z) := f(z^n)
\]
for $f \in \mc{L} Z = \Map(S^1, Z)$ and $z \in S^1 \subset \CC$. For $n \ne 0$, the
map $\psi(n)$ is a rational (and real) homotopy equivalence, because
so is the power map \eqref{power}. However, $\psi(n)$ is not
$S^1$-equivariant, unless $n=1$, as, for instance, for the right action $(f
\cdot z') (z) := f(z' z)$ of $S^1$ on $\mc{L} Z$, we have
\vspace{-1mm} 
\[
\psi(n) (f \cdot z') (z) = f (z' z^n) \ne f((z'z)^n) = (\psi(n) (f)
\cdot z') (z).
\]

\vspace{-1mm} 
\noindent
Moreover, we can say that
\[
\psi(n) (f \cdot z') = \psi(n) (f) \cdot \sqrt[n]{z'}
\]
\newpage 
\noindent for any choice of the $n$th root, or, better, just using the rational
homotopy inverse of the rational equivalence \eqref{power}.
Accordingly, the map $\psi(n)$ would not induce a map $\mc{L}_c Z \to
\mc{L}_c Z$ of the homotopy quotient $\mc{L}_c Z = \mc{L} Z
\times_{S^1} ES^1$. Indeed, by definition of the quotient $\mc{L} Z
\times_{S^1} ES^1$, a point $(f \cdot z', e) \in \mc{L}_c Z$ is
equivalent to the point $(f, z' \cdot e)$, but $ (\psi(n) (f \cdot z'), e) =
(\psi(n) (f) \cdot \sqrt[n]{z'}, e) \sim (\psi(n) (f), \sqrt[n]{z'}
e)$, which is not equivalent to $ (\psi(n) (f), z' \cdot e)$.

\medskip 
What saves the situation is that the map $\psi(n)$ extends in the
rational homotopy category to a
morphism $\mc{L} Z \times ES^1 \to \mc{L} Z \times ES^1$ that respects the equivalence relation
\begin{equation}
  \label{equiv-rel}
 (f \cdot z', e) \sim (f, z' \cdot e)
\end{equation}
and thereby descends to the quotient $\mc{L} Z \times_{S^1}
ES^1$. Indeed, note that the topological group morphism \eqref{power}
induces a continuous map $\rchi(n): ES^1 \to ES^1$ of the total space
$ES^1$ of the universal bundle $ES^1 \to BS^1$ by functoriality. In
the standard simplicial model of $ES^1$, this map $\rchi(n)$ can be
expressed as $[z_0, \dots, z_p] \mapsto [z_0^n, \dots, z_p^n]$, where
$[z_0, \dots, z_p]$, $z_i \in S^1$, is a $p$-simplex of $ES^1$. The
map $\rchi(n): ES^1 \to ES^1$ is not $S^1$-equivariant, either. Say,
for the (left) action $z' \cdot [z_0, \dots, z_p] := [z' z_0, \dots,
  z' z_p]$ of $S^1$ on $ES^1$, we have
  \vspace{-1mm} 
\[
\rchi(n) (z' \cdot [z_0, \dots, z_p]) = [(z' z_0)^n, \dots, (z' z_p)^n]
\ne [z' z_0^n, \dots, z' z_p^n] = z' \cdot \rchi(n) ([z_0, \dots,
  z_p]).
\]

\vspace{-1mm} 
\noindent
What we have is
\[
\rchi(n) ( z' \cdot e) = (z')^{n} \cdot \rchi(n) (e), \qquad e \in ES^1.
\]
For $n \ne 0$, the map $\rchi(n)$ is a rational (and real) homotopy equivalence and
therefore has a rational homotopy inverse $\rchi(n)^{-1}$, so that
\[
\rchi(n)^{-1} ( z' \cdot e) = \sqrt[n]{z'} \cdot \rchi(n)^{-1} (e).
\]
Now, the morphism
\[
\psi(n) \times \rchi(n)^{-1}: \mc{L} Z \times ES^1  \to \mc{L} Z \times ES^1,
\]
which is invertible in the rational (real) homotopy category, respects
the equivalence relation \eqref{equiv-rel}:
\vspace{-2mm} 
\begin{multline*}
\psi(n) \times \rchi(n)^{-1} ( f \cdot z', e)) = (\psi(n) (f) \cdot
\sqrt[n]{z'}, \rchi(n)^{-1}(e)) 
% \;\;\; 
\\
\sim \;\;\;
(\psi(n) (f), 
\sqrt[n]{z'}  \cdot  \rchi(n)^{-1}(e))
 = \psi(n) \times \rchi(n)^{-1} ( f ,  z' \cdot e),
\end{multline*}
and therefore induces a rational automorphism of $\mc{L}_c Z$,
which we denote by
\[
\varphi_1(n): = \psi(n) \times \rchi(n)^{-1}: \mc{L}_c Z 
\longrightarrow \mc{L}_c Z\,.
\]

As in the case of $\varphi_0(n)$ in \cref{natural0}, the rational homotopy equivalence
$\varphi_1 (n): \mc{L}_c Z \xrightarrow{\sim} \mc{L}_c Z$ extends
to a group homomorphism
\vspace{-2mm} 
\begin{align*}
%  \label{aut}
  \QQ^\times & \longrightarrow \Aut_\QQ  \mc{L}_c Z,\\
%  \nonumber
  r & \longmapsto \varphi_1(r).
\end{align*}
From this, we get a $\QQ$-algebraic-group homomorphism
\[
\GG_m \longrightarrow \Aut M( \mc{L}_c Z),
\]
so as $r \in \GG_m$ acts on $w \in H^2(BS^1; \R) \subset M(BS^1)$
as $r^{-1} w$, because $\rchi(n)$ induces the action $w \mapsto nw$ on
degree-two cohomology and we used $\rchi(n)^{-1}$, and $r \in
\GG_m$ acts on $sv \in V[1] \subset M(\mc{L} Z)$ as $r(sv)$. This
motivates formulas \eqref{inverse}.
Summarizing, we have: 

\begin{prop}[Toroidal symmetries from $S^1$]
 \label{prop-torS1}
The $n$-fold winding maps $S^1  \to S^1$, $n \in \ZZ \setminus \{0\}$, induce morphisms $\varphi_1(n): \mc{L}_c Z \to \mc{L}_c Z$ in the rational homotopy category by the construction above. 

\vspace{-2mm} 
  \begin{enumerate}[{\bf (i)}]
   \setlength\itemsep{-2pt}
  \item
These morphisms are invertible, and the compositions of them with their inverses
give a group homomorphism 
$\Q^\times  \xrightarrow{\sim} \Aut
  M(\mc{L}_c Z)(\Q)$, which defines naturally a morphism of $\QQ-algebraic groups$
$$
\GG_m  \xrightarrow{\;\; \sim \;\;} \Aut M( \mc{L}_c Z)\;.
$$

\vspace{-1mm} 
\item  For an iterated cyclic loop space $\mc{L}_c^k Z$, $k \ge 1$, the algebraic-group morphisms corresponding to different iterations commute and thereby define an algebraic-group morphism
$
  (\GG_m)^k  \xrightarrow{\sim} \Aut
  M(\mc{L}^k_c Z),
 $ which provides an action of $(\GG_m)^k$ on the Sullivan minimal model of the $k$-fold cyclification $\mc{L}_c^k Z$ of $Z$.
 \end{enumerate} 
\end{prop}

%%%%%%%%%%%%%%%%
\subsection{Toroidal symmetries of type IIB}
%%%%%%%%%%%%%%%%

Since the type IIB model (see \cref{IIB}) falls out of the previous sequence of cyclifications of $S^4$, 
we need to treat it separately. We will work with our ``unstable'' model $M(IIB) = (S(V),d)$ of type IIB; 
see \eqref{htypeIIB}. 

\medskip 
By \Cref{kerd}, to identify a maximal $\RR$-split torus $T^B$ of $\Aut M(IIB)$, we need to start with computing $\ker d \cap V = \RR h_3 \oplus \RR \omega_1$ in the notation of the system \eqref{htypeIIB}. 
Thus, a maximal  $\RR$-split torus is at most 2-dimensional. The explicit formulas below identify a 2-dimensional split torus acting faithfully on $M(IIB)$, which has to be maximal by the dimension argument.

\medskip 
Take the $\RR$-split torus $T^B := \GG_m^2$ acting as the group of diagonal matrices on the real plane spanned by $h_3$ and $\omega_1$.
Denote by $\beta_0 = (1,0)$ the weight of $h_3$ and by $\beta_1 = (0,1)$ the weight of $\omega_1$, so that
we have the action 
\begin{equation}
\label{linear}
    t \cdot h_3 = t^{\beta_0} h_3, \qquad t \cdot \omega_1 = t^{\beta_1} \omega_1, \qquad \qquad t \in T^B = \GG_m^2.
\end{equation}

Then, with \eqref{htypeIIB}, we have the following:

\begin{prop}[Toroidal symmetry in type IIB] 
\label{prop-torIIB}
The formulas 
\begin{gather}
\label{tIIB}
  t \cdot \omega_{\, 3} = t^{\beta_0 + \beta_1} \omega_{\, 3}, 
  \qquad   t \cdot \omega_{\, 5} = t^{2\beta_0 + \beta_1} \omega_{\, 5},\\
  \nonumber
  t \cdot h_7 = t^{3\beta_0 + 2 \beta_1} h_7,  \qquad
  t \cdot \omega_{\, 7} = t^{3\beta_0 + \beta_1} \omega_{\, 7}
  \end{gather}
  extend the action \eqref{linear} to an action of the torus $T^B = \GG_m^2$ on $M(IIB)$. This identifies $T^B$ as a maximal $\RR$-split torus of $\Aut M(IIB)$.
\end{prop}

%%%%%%%%%%%%%%%%%%%%%%%%%%%%%%%%%%%%%%%%% 
\section{The $E_k$ symmetry of iterated cyclic loop spaces}
%%%%%%%%%%%%%%%%%%%%%%%%%%%%%%%%%%%%%%%%%
Here we unravel the $E_k$ symmetry of the iterated cyclic loop spaces
$\mc{L}^{k}_c S^4$, described in \cref{Sullmm}, 
where $E_k$ for $k \ge 0$ is understood in the sense of \Cref{table1}
and \Cref{table2}
in the Introduction. Our goal is to use the toroidal symmetries of the cyclic loop
spaces $\mc{L}_c^k S^4$ from \cref{Sec-TorSymm} 
and build certain canonical combinatorial data: ``a lattice
$N_k$ with an inner product $(-,-)$ and a distinguished element
$K_k^* \in N_k$'', similar to the triple $(N_k, (-,-), \cK_k)$
in the theory of del Pezzo surfaces, see below. This will
automatically produce the $E_k$ root system, see \Cref{rootdata}.

%%%%%%%%%%%%%%%%%
\subsection{Reminder: the combinatorial data from del Pezzo surfaces}
%%%%%%%%%%%%%%%%%
\label{dPtriple}

Let us recall how the triple $(N_k, (-,-), \cK_k)$ shows up in the del Pezzo theory, for the sake of motivation and setting up notation. The Picard group $\Pic \BB_k$ happens to be isomorphic to the 2nd cohomology group $H^2(\mathbb{B}_k; \Z)$. This is a lattice with basis $\cH, \cE_1, \dots, \cE_k$:
$$
H^2(\mathbb{B}_k; \Z) \cong \Z \cH \oplus \Z \cE_1 \oplus \cdots \oplus \Z \cE_k,
$$
where $\cH$ is the class of the proper transform of the line
(here also a hyperplane) $\CC \PP^1$ in $\mathbb{C P}^2$ 
and $\cE_i$ is the class of the exceptional divisor over the blowup point $x_i \in \mathbb{C P}^2$.
See \cite{Man}\cite{Demazure}\cite{Beauville}\cite{KSC}. 
The 2nd integral cohomology has a natural inner product given by the intersection form:
\(
\label{HE-inter}
\cH \cdot \cH=1, \quad \cH\cdot \cE_i=0, \quad \cE_i \cdot \cE_j = - \delta_{ij}, \qquad 1 \leq i, j \leq k\;. 
\)
Thus the intersection matrix of $\mathbb{B}_k$ is given by the Lorentzian form 
$
Q={\rm diag}(1, -1, -1, \dots, -1)
$.
Hence $H^2(\mathbb{B}_k;\ZZ) \linebreak[0] \cong \ZZ^{1,k}$
and $\mathbb{B}_k$ has Betti numbers $b_2^+=1$, $b_2^{-}=k$, with signature $\sigma=1-k$.

\medskip 
% Another important feature of the del Pezzo surface $\BB_k$ is
The \emph{canonical class} of $\BB_k$ may be expressed as
\begin{equation}
\label{k-canclass}
\cK_k := \Omega^2_{\BB_k} = - 3 \cH + \cE_1 + \dots + \cE_k,
\end{equation}
while the ample \emph{anticanonical class} becomes
 \(
 \label{AntiK}
-\cK_k = 3 \cH - \cE_1 - \dots - \cE_k.
\)
The \emph{degree of a divisor $\cD \in \Pic (\BB_k)$} is measured with respect to the anticanonical map as:
\(
\label{eqDegD}
\deg \cD := -\cK_k \cdot \cD\;.
\)

The ``outlier'' del Pezzo surface $\BB_1' := \CC \PP^1 \times \CC \PP^1$
of degree 8 has Picard group of rank 2:
$$
\Pic (\CC \PP^1 \times \CC \PP^1) \cong H^2 (\CC \PP^1 \times \CC \PP^1; \ZZ)  
\cong \ZZ \mathfrak{l}_1 \oplus \ZZ \mathfrak{l}_2\;,
$$
where $\mathfrak{l}_1$ and $\mathfrak{l}_2$ are the classes of the two $\CC \PP^1$ factors. The intersection pairing works as follows (cf. \eqref{HE-inter}):
\vspace{-1mm}  
\[
% \label{HE-inter}
\mathfrak{l}_1 \cdot \mathfrak{l}_1 = \mathfrak{l}_2 \cdot \mathfrak{l}_2 = 0, \quad \mathfrak{l}_1 \cdot \mathfrak{l}_2 = 1 \;. 
\]
The anticanonical class is given by $- \cK_{\BB_1'} = 2 \mathfrak{l}_1 + 2 \mathfrak{l}_2$. 
% $\CC \PP^1 \times \CC \PP^1$ is another del Pezzo surface of degree 8. 

%%%%%%%%%%%%%%%%%%%%%%%%%%
\subsection{The ``Cartan subalgebra'' and weight lattice}
%%%%%%%%%%%%%%%%%%%%%%%%%%%
% \label{Cartan}

% {\color{red} Move intersection pairing from intro here}

%\medskip 
In this section, we will present the first element of the triple arising from $\mc{L}_c^k S^4$, the lattice. The idea is to use the weight lattice coming from the weight
decompositions of Corollary \ref{corr} and of
\cref{symm-cycl-concrete} and use the Lie algebra $\h_k = \Lie
(T^{k+1})$ of the maximal real split torus $T^{k+1}$ of $\Aut M(
\mc{L}_c^{k+1} S^4)$ with its canonical factorization $T
\xrightarrow{\sim} \GG_m^{k+1}$, see \Cref{split-rank}. This Lie algebra constitutes the
infinitesimal symmetries corresponding to the toroidal symmetries of
 \cref{symm-cycl}--\cref{symm-cycl-concrete} and acts on the
Sullivan minimal model $ M( \mc{L}_c^{k+1} S^4)$ by derivations.
That is, we have a Lie algebra homomorphism:
\vspace{-2mm} 
\[
\h_k \longrightarrow \Der M(\mc{L}_c^k S^4),
\]
which comes from taking the differential of the action
\[
T^{k+1} \longrightarrow \Aut M(\mc{L}_c^k S^4).
\]

Under the action of the Lie algebra $\h_k$ on $M = M(\mc{L}_c^k S^4)$,
the weight decomposition of \cref{tor-symm} and Corollary
\ref{corr} becomes
\[
M = \bigoplus_{\alpha \in \PPP(\h_k)} M_{\alpha},
\]
where $\PPP(\h_k) \subseteq \h_k^* = \Hom_\RR (\h_k, \RR)$ is the
\emph{weight lattice}, the image of the character group 
$$
\mathfrak{X}(T) =
\Hom_\RR (T, \GG_m) \linebreak[0] \cong \ZZ^{k+1}
$$
under the differential map
\vspace{-3mm} 
\begin{align*}
  \mathfrak{X}(T) &\longrightarrow \h_k^*,\\
  \beta & \longmapsto d \beta.
\end{align*}
The Lie algebra $\h_k$ acts on each weight space $M_\alpha$ with the
weight $\alpha$:
\[
M_{\alpha} = \{m \in M \; | \; h \cdot m = \alpha(h) m \quad
\text{for all $h \in \h_k$}\}.
\]

As we will see soon, the Lie algebra $\h_k$ is an avatar of the Cartan
subalgebra of the Lie algebra of type $E_k$ of ``hidden'' symmetries
of the cyclic loop spaces of the four-sphere.

\begin{theorem}[Bases for the Lie algebra of symmetries and its dual]
\label{thm-bases}
$\,$

\vspace{-3mm} 
  \begin{enumerate}[{\bf (i)}]
   \setlength\itemsep{-1pt}
  \item
    The $(k+1)$-dimensional real abelian Lie algebra $\h_k = \Lie
    (T^{k+1}) \subseteq \Der M(\mc{L}_c^k S^4)$ of the maximal
    $\RR$-split torus $T = T^{k+1}$ of the algebraic group $\Aut
    M(\mc{L}_c^k S^4)$ has a canonical basis $\{h_0, h_1,
    \linebreak[0] \dots, \linebreak[1] h_k\}$.
\item
  The weights $\epsilon_0$ of $g_4$ and $\epsilon_i$ of $w_i$ for $1 \le i
  \le k$ give a canonical basis $\{\epsilon_0, \epsilon_1, \dots,
  \epsilon_k\}$ of the vector space $\h_k^*$. This is also a basis of
  the weight lattice $\PPP(\h_k) \subseteq \h_k^*$.
  \end{enumerate}
\end{theorem}

\begin{proof}
(i) The factorization $T = \GG_m^{k+1}$ of the maximal split torus of
$\Aut M(\mc{L}_c^k S^4)$ was canonically defined from compatibility
with the iterated cyclic loop space structure (see 
\Cref{split-rank}):
\[
T = \{ (t_0, t_1, \dots, t_k) \; | \; t_i \in \GG_m \}\;,
\]
with $t_0$ acting on $g_4$ by $t_0 g_4$ and trivially on $w_1, \dots,
w_k$ and $t_i$ acting on $w_i$ by $t_i^{-1} w_i$ and trivially on the
other $w_j$'s and $g_4$. This factorization implies a canonical
factorization of the tangent space $\h_k$ at $\id \in T$: $\h_k =
\RR^{k+1}$. It determines a basis $\{h_0, h_1, \dots, h_k\}$, the
standard basis of $\RR^{k+1}$.

(ii) The identification of $T$ as $\GG_m^{k+1}$ in \Cref{split-rank} was derived iteratively from
Theorem \ref{step} as coming from the action of the torus on the generators
$g_4$, $w_1$, \dots, and $w_k$ of $M(\mc{L}_c^k S^4)$. Therefore, the
weights of these generators provide a natural set of weights
$\epsilon_0, \epsilon_1, \dots, \epsilon_k$.  Differentiating the action of
$T$ on these generators (see the previous paragraph), we obtain
\vspace{-3mm} 
\begin{align*}
h_0 \cdot g_4 &= g_4, \qquad h_0 \cdot w_i = 0 \qquad \;\, \text{ for $i =
  1, \dots, k$},\\
h_i \cdot g_4  &= 0, \qquad \; \; h_i \cdot w_j = -\delta_{ij} \quad \text{
  for $i, j = 1, \dots, k$}.
\end{align*}
This implies that
\vspace{-3mm}
\begin{align*}
\epsilon_0 (h_0)& = 1, \qquad \epsilon_i (h_0) = 0 \qquad \;\, \text{ for $i = 1,
  \dots, k$},\\
\epsilon_0 (h_i) & = 0, \qquad \epsilon_j (h_i) = -\delta_{ij} \quad \text{
  for $i,j = 1, \dots, k$}.
\end{align*}
From these equations, we conclude that $\{ \epsilon_0, \epsilon_1, \dots,
\epsilon_k\}$ form a basis of the dual vector space $\h_k^*$ as well as
the weight lattice $\PPP(\h_k) \subseteq \h_k^*$.
\end{proof}

From  \cref{symm-cycl-concrete}, we can compute the action of
the Lie algebra $\h_k$ on the other generators of the minimal DGCA
$M(\mc{L}_c^k S^4)$. For instance:

\begin{example}[Action of the Lie algebra on the cyclic loop space in type IIA]
For $k =1$, we have
\begin{gather*}
    h \cdot g_4 = \epsilon_0(h) g_4, \quad h \cdot g_7 = 2 \epsilon_0(h) g_7, \quad h \cdot
    sg_4 = (\epsilon_0(h) - \epsilon_1(h)) sg_4,\\
    h \cdot sg_7 = (2 \epsilon_0(h) - \epsilon_1(h)) sg_7, \quad h \cdot w =
  \epsilon_1(h) w,
\end{gather*}
where
\vspace{-3mm}
\begin{align*}
  \epsilon_0(h_0)& = 1, \qquad \epsilon_0 (h_1)= 0,\\
  \epsilon_1(h_0)& = 0, \qquad \epsilon_1 (h_1) = -1.
\end{align*}
\end{example}

\begin{cor}[The dual lattices and inner product]
$\,$

\vspace{-3mm} 
\begin{enumerate}[{\bf (i)}]
 \setlength\itemsep{-1pt}
\item
  The lattice $\h_k^\ZZ := \ZZ h_0 \oplus \ZZ h_1 \oplus \dots \oplus
  \ZZ h_k \subset \h_k$ is the dual of the weight lattice
  $\PPP(\h_k) = \ZZ \epsilon_0 \oplus \dots \oplus \ZZ \epsilon_k \subset \h_k^*$.
\item
  The nondegenerate bilinear form $\h_k \otimes \h_k \to \RR$
  determined by the isomorphism
  \vspace{-2mm}
  \begin{align*}
    \h_k &\to \h_k^*\\[-1pt]
    h_i & \mapsto \epsilon_i
  \end{align*}
  
  \vspace{-3mm}
\noindent  provides the vector space $\h_k$ with a canonical Lorentzian inner
  product $(-,-)$. This inner product satisfies the formulas:
  \vspace{-1mm}
  \[
  (h_0, h_0) = 1, \qquad (h_i, h_j) = - \delta_{ij} \quad \text{ for \; 
    $i \ge 0$, $j \ge 1$.}
  \]
 \item
The inner product induced on the dual space $\h_k^*$ is given by the formulas:
 \vspace{-1mm}
  \[
  (\epsilon_0, \epsilon_0) = 1, \qquad (\epsilon_i, \epsilon_j) = - \delta_{ij} \quad \text{ for \; 
    $i \ge 0$, $j \ge 1$.}
  \]
  \end{enumerate}
\end{cor}

%%%%%%%%%%%%%%%%%%%%%%%%%%%%%%%%%%%%
\subsection{The ``anticanonical class''}
\label{can-section}
%%%%%%%%%%%%%%%%%%%%%%%%%%%%%%%%%%%%

In this section, we identify a distinguished element $-K_k \in
\h_k^\ZZ$, analogous to the anticanonical class $-\cK_k$ (see the Introduction,
\cref{intro}, and \cref{dPtriple}) of the del Pezzo surface $\BB_k$. Recall, that the
anticanonical class $-\cK_k$ of the del Pezzo surface ``acts'' on the
Picard group $\Pic (\BB_k) \cong H^2 (\BB_k;
\Z)$ by degree:
\begin{equation}
  \label{degree}
\deg \cD := - \cK_k \cdot \cD, \qquad \cD \in \Pic (\BB_k).
\end{equation}
Here the ``action'' is understood as the intersection product $\Pic
(\BB_k) \otimes \Pic (\BB_k) \to \Z$. The \emph{degree of a divisor}
$\cD \in \Pic (\BB_k)$ is defined using the anticanonical morphism
$f: \BB_k \to \CC \PP^d$, and 
$$
d = h^0 (\BB_k, -\cK_k) - 1 = (-\cK_k) \linebreak[0] \cdot (-\cK_k) =
9 - k
$$ 
is known as the \emph{degree of the del Pezzo surface $\BB_k$},
whence $-\cK_k$ is the pullback $f^* \mc{H}$ of the hyperplane class $\mc{H} \in  \Pic ( \CC \PP^d)$ and formula \eqref{degree} makes sense; see, e.g., \cite{Dolgachev}
and also \cref{dPtriple}.

\medskip
In the case of cyclic loop spaces, we have been dealing with one
notion of degree, in the sense of $\Z$-grading of the Sullivan minimal
model $M(\mc{L}_c^k S^4)$. There is another one, natural for the
Quillen minimal model. The significance of that other notion of degree
for us is coming from the fact that it corresponds to the degree of
the C-fields, i.e., the potentials $C_3$ and $C_6$ of the basic
fields $G_4$ and $G_7$, see \eqref{C-fields},
% and \eqref{scale-C3C6},
and thereby to the dimension of the corresponding
branes, the M2- and M5-branes, respectively. This other notion of degree
just differs from the degree we have been using on the generators of
the Sullivan minimal model by one, but has a homotopy-theoretic
origin, as we  now explain.

\medskip 
We have seen the Quillen model in \cref{Sec-Quillen} and at the end of \cref{cyclification}. The maximal split torus $T^{k+1}$ and its Lie algebra $\h_k$ act on the
Quillen minimal model $Q(\mc{L}_c^k S^4)$ with the same weights as on
the generators of the Sullivan minimal model $M = M(\mc{L}_c^k
S^4)$. Indeed, the weights on the space $V = M^+/(M^+)^2$ of
generators, its dual $V^*$ and its degree shift $V^*[1] = Q(\mc{L}_c^k
S^4)$ will just be the same as those on $V$. Let us denote the
elements of the basis of $V^*[1] = Q(\mc{L}_c^k S^4)$ dual to the
basis 
$$
\{g_4, \, g_7, \, w_1, \, s_1 g_4, \, s_1 g_7, \, s_1 w_1, \dots, \, w_k\} 
\qquad \text{by}
\qquad 
e_3, \, e_6, \, x_1,  \, s_1 e_3, \, s_1 e_6, \, s_1 x_1, \dots, \, x_k,
$$
respectively. From the fact that $S(Q(Z)[-1]^*), d)$ is the
Chevalley-Eilenberg cochain complex of the graded Lie algebra $Q(Z)$,
one can deduce that the subspace dual to $\ker d \cap V[1] = \ker d
\cap Q(Z)^*$ generates the Quillen minimal model as a graded Lie
algebra.

\medskip 
A remarkable fact is that the \emph{degree operator}
\begin{equation}
\label{degree-operator}
x \mapsto \abs{x} \cdot x, \qquad x \in Q(\mc{L}_c^k S^4),
\end{equation}
in the Quillen minimal model singles out a distinguished element of the Lie
algebra $\h_k$. Here is a more precise statement.

\begin{theorem}[Degree in the $k$-fold cyclic loop space of $S^4$]
\label{thm-deg}
There is a unique element of the Lie algebra $ \h_k$, namely,
\begin{equation}
  \label{elt}
-K_k := 3h_0 - h_1 - \dots - h_k,
\end{equation}
which acts on the Quillen minimal model $ Q(\mc{L}_c^k S^4)$ as the
degree operator:
\[
-K_k \cdot x = \abs{x} \cdot x.
\]
\end{theorem}

\begin{proof}
Indeed, the operator $3h_0$ acts on $g_4$ and thereby $e_3$ with
weight 3 and on $x_i$, $i = 1, \dots, k$, with weight zero:
\begin{equation*}
%  \label{3-1}
3h_0 \cdot e_3 = 3 e_3, \qquad 3 h_0 \cdot x_i = 0,
\end{equation*}
whereas the operator $-h_i$ acts on $g_4$, $e_3$, and $x_j$ by zero,
except for $x_i$, on which it acts by 1:
\[
- h_i \cdot x_i = x_i.
\]
Likewise, the element \eqref{elt} acts on $e_6$ via multiplication by
its degree, which is 6. Also the degree of $s_ix$ for $x \in Q(Z)$ will be
one less than the degree of $x$ and the weight of $s_ix$ will be
$\alpha - \epsilon_i$ for $x$ of weight $\alpha$. Since $\epsilon_i(3h_0) =
0$ and $\epsilon_i (-h_j) = \delta_{ij}$, the element \eqref{elt} will
act on $s_i x$ by its degree, provided we know that it acts on $x$ by
the degree of $x$.
%Since the degrees add under the Lie bracket and so
%do the weights,
This way we get a complete matching between the action of $3 h_0 - h_1
- \dots - h_k$ and the degree operator on the Quillen minimal model
$Q(\mc{L}_c^k S^4)$.

The uniqueness of an element of $\h_k$ which acts on $Q(\mc{L}_c^k
S^4)$ by degree comes from the fact that an arbitrary element $a_0 h_0
+ a_1 h_1 + \dots + a_k h_k \in \h_k$ acts on $e_3$ by $a_0$ and on
each $x_i$ by $-a_i$, which forces $a_0$ to be 3 and each $a_i$ to be
$-1$, because the degree of $e_3$ is 3 and the degree of $x_i$ is 1.
\end{proof}

Recall from \eqref{degree} that the anticanonical class $- \cK_k$ of the del Pezzo
surface $\BB_k$ acts on the Picard group via intersection pairing by degree.  Thus, it makes
all sense to use the element
\[
K_k = -3h_0 + h_1 + \dots + h_k \; \in \; \h_k
\]
as a distinguished element, the analogue of the canonical class. 
Extending the analogy with del Pezzo surfaces, 
\begin{defn}
[Degree of cyclification]
We define the \emph{degree of the cyclic loop space $\mc{L}_c^k S^4$ of the four-sphere} as
\begin{equation}
\label{deg-L}
\deg \mc{L}_c^k S^4 = (-K_k, -K_k) = 9 - k.
\end{equation}
\end{defn}

%%%%%%%%%%%%%%%%%%%%%%%%%%%%%%%%%%
\subsection{The $E_k$ root system and its Weyl group}
\label{Sec-Ek} 
%%%%%%%%%%%%%%%%%%%%%%%%%%%%%%%%%%

We now explain the role of the Weyl group and how we obtain the exceptional root data from 
$\mc{L}_c^k S^4$.

\begin{theorem}[Exceptional root data from cyclic loop spaces of the 4-sphere]
\label{rootdata}
$\,$

\noindent {\bf (i)}  For each $k \ge 0$,
  % with $0 \le k \le 8$,
  the data 
   $$
  \big(\h_k^*, \linebreak[1] \{\epsilon_o, \epsilon_1, \linebreak[0] \dots,
  \epsilon_k\}, (-,-), K_k^*\big)
  $$ 
  associated to the cyclic loop space
  $\mc{L}^{k}_c S^4$ and its Sullivan minimal model $M(\mc{L}^{k}_c
  S^4)$ consists of
  \vspace{-2mm} 
  \begin{itemize}
   \setlength\itemsep{-1pt}
  \item[{\bf (a)}] a real vector space $\h_k^*$ with a basis $\{\epsilon_0, \epsilon_1, \dots,
    \epsilon_k\}$, which generates a lattice $\PPP(\h_k) \subset \h_k^*$;
  \item[{\bf (b)}] a symmetric bilinear form $\h_k^* \otimes \h_k^* \to \R$
    given by
    \vspace{-2mm}
    \[(\epsilon_0, \epsilon_0) = 1, \qquad (\epsilon_i, \epsilon_j) = -\delta_{ij}, \qquad i > 0, j \ge 0;
    \]
    \vspace{-8mm} 
    \item[{\bf (c)}] a distinguished element $K_k^* = -3 \epsilon_0 + \epsilon_1 + \dots +
      \epsilon_k$.
        \end{itemize}
        
\vspace{-2mm}         
\noindent {\bf (ii)} 
This data replicates the data
\vspace{-3mm} 
$$\big(H^2(\BB_k; \R), \{\cH, \cE_1, \dots,
\cE_k\}, (-, -), \cK_k \big)$$ 

\vspace{-1mm} 
\noindent 
determined by the rational surface $\BB_k$,
considered as the blowup of $\CC \PP^2$ at $k$ points; see
\cref{dPtriple} in the del Pezzo case, when $k \le 8$. For $k \le 8$, the data produces the root system
\vspace{-2mm}         
\begin{equation}
\label{roots}
R_k := \big\{ \alpha \in \PPP(\h_k) \; | \; (\alpha, K_k^*) = 0 , (\alpha, \alpha) = -2 \big\}
\subset (K_k^*)^\perp \subset \h_k^*
\end{equation}

\vspace{-2mm}         
\noindent
of type \footnote{True/genuine $E_k$ for $k =6, 7$, and 8, and using the conventions of \Cref{table1}
for $0 \le k \le 5$.} $E_k$ and the Weyl group $W(E_k)$, generated by the
reflections in the hyperplanes orthogonal to the
roots $r \in R_k$, now in the context of cyclic loop spaces of $S^4$.
\end{theorem}

This result has an independent interest, apart from
Mysterious Duality, as it uncovers a new symmetry pattern for the
series $\mc{L}^{k}_c S^4$, $0 \le k \le 8$, of cyclic loop spaces of
the 4-sphere. This should have a number of topological consequences
shedding new light on these spaces. For instance, one may wonder: what
is the analogue of an exceptional curve on a del Pezzo surface? What corresponds to the famous statement about the 27 lines on the cubic surface $\BB_6$
on the topological side, the six-fold cyclification $\mc{L}^{6}_c S^4$ of $S^4$?
We investigate these intriguing questions in
% \cref{Sec-excvec} and
\cref{Sec-27}.

\begin{remark}[Why $k \le 8$ vs. $k \geq 9$]
\label{k-le-8}
  The construction of the data $\big(\h_k^*, \{\epsilon_0, \dots, \epsilon_k\}, (-,-), K_k^*\big)$ of Theorem \ref{rootdata} above extends beyond $k = 8$ verbatim. However, the identification of the root system for $k > 8$ needs to be treated with care. For $0 \le k \le 8$, the degree $\deg (\mc{L}_c^k S^4) = (K_k, K_k) = (K_k^*, K_k ^*) = 9-k$ of the cyclic loop space is positive, just like the degree of the del Pezzo surface $\BB_k$. From simple linear algebra of Lorentzian inner products, we can see that $(K_k, K_k) > 0$ implies that the inner product induced on the subspace
  $K_k^\perp = \{ x \in \h_k \; | \; (x, K_k) = 0 \}$ of
  $\h_k$ by the Lorentzian inner product $(-,-)$ is negative-definite. (If we switch the sign and use $-(-,-)$, it would be a more familiar positive-definite inner product). This also implies that the root system $R_k$ is finite; see \cite{Man}. For $k \ge 9$, the inner product loses its negative-definiteness, and the root
  system $R_k$ becomes infinite and can be identified as the set of \emph{real roots} of a
  more general Kac-Moody algebra. In fact, for $k = 9$, the subspace $K_k^\perp$ is negative semi-definite of nullity 1. For $k > 9$, the orthogonal complement $K_k^\perp$ gets a Lorentzian inner product. See the discussion of the $k \ge 9$ cases in \Cref{Sec-KM}.
  % See \cref{sec-blow10} and \cref{sec-blow11}.
\end{remark}

\paragraph{Weyl group as symmetry of symmetries.}
For each root $\alpha_{}$, an element of the set $R_k$ (see expression \eqref{roots})  
define a reflection
\vspace{-1mm} 
\[
\sigma_\alpha: \beta \;\; \longmapsto \;\; \beta - 2 \frac{(\beta,\alpha)}{(\alpha,\alpha)} \alpha = \beta + (\beta,\alpha) \alpha, \qquad \beta \in \h_k^*,
\]

\vspace{-1mm} 
\noindent 
of the Lorentzian space $\h_k^*$. The reflections $\{\sigma_\alpha \; | \; \alpha \in R_k \}$ generate a group, 
known to be the Weyl group $W(E_k)$ of the root system $E_k$. It is also known that $W(E_k)$ is the group of all linear isometries of $\PPP(\h_k)$ which preserve the element $K_k^*$ \cite{Man}.
Thus, the Weyl group, being an isometry group of (the dual of) the abelian Lie algebra $\h_k$ of infinitesimal 
symmetries of $M(\mc{L}_c^k S^4)$, is a ``second derived'' object with respect to
$\mc{L}_c^k S^4$: $W(E_k)$ is the group of ``symmetries of symmetries'' of $\mc{L}_c^k S^4$.
This is typical for the role of Weyl groups in Lie theory: a Lie group is usually a group of 
symmetries of a certain mathematical object, the Cartan subalgebra is the maximal abelian Lie algebra of infinitesimal symmetries of that object. The Weyl group is a group of symmetries of the Cartan subalgebra.

\paragraph{The moduli space of $k$-fold cyclifications of $S^4$.}
We can interpret an element $\omega$ of the linear dual space $\h_k^* \cong \R^{k+1}$ of the abelian Lie algebra $\h_k$ of infinitesimal symmetries of $M(\mc{L}_c^k S^4)$ as some sort of an extra, geometric ingredient complementing the purely topological data carried by the real homotopy type of the $k$-fold cyclification $\mc{L}_c^k S^4$. Indeed, an element $\omega \in \h_k^*$ is a {\it weight}, and as such, it tells us which ``spectral parameters'' we might want to assign to the basic infinitesimal symmetries $h_0, h_1, \dots, h_k$.

\medskip 
For instance, recall from \Cref{natural0} that $h_0$ comes from the action of the real 1-torus $\R^\times$ by automorphisms of the real homotopy type of $S^4$ and the resulting action of $\R^\times$ on $\mc{L}_c^k S^4$. That action originates ultimately from the folding self-maps of $S^4$. The value $\omega(h_0) \in \R$ tells us how much we shall value the effect of the folding self-maps of $S^4$. In this sense, $\omega(h_0)$ is akin to the size of $S^4$, such as its radius $R_0$, or rather the logarithm $\log R_0$ thereof, since $\omega(h_0)$ is not necessarily positive. This value $\omega(h_0)$ is analogous to the logarithmic Planck scale $\log \ell_p$ in the 11-dimensional supergravity and the generalized K\"{a}hler volume $\omega(H) = \int_H \omega$ of the line $H = \CC \PP^1$ in $\CC \PP^2$ and its image in $\BB_k$ in the del Pezzo story; cf.\ \cite[\S 3.1]{INV}.

\medskip 
Similarly, as per \Cref{natural1}, the element $h_i$ for each $i$, $1 \le i \le k$, comes from the action of $\R^\times$ by the folding self-maps of the $i$th source circle of the cyclic loop space $\mc{L}_c^k S^4$. Assigning $h_i$ a real value $\omega(h_i)$ tells us how much we shall value the effect of the folding self-maps of the $i$th source circle of $\mc{L}_c^k S^4$. In this way, $\omega(h_i)$ is analogous to the logarithm $\log R_i$ of the radius $R_i$ of the $i$th source circle and the $i$th compactification circle in M-theory wrapped on $T^k = (S^1)^k$, as well as the generalized K\"{a}hler volume $\omega(\cE_i)$ of the exceptional divisor $\cE_i$ in $\BB_k$ in the del Pezzo story; cf.\ \cite[\S 3.1]{INV} again.

\medskip 
In this sense the choice of a weight $\omega \in \h_k^*$ adds a certain ingredient of \emph{metric} flavor, missing in the real homotopy model $M(\mc{L}_c^k S^4)$ of $\mc{L}_c^k S^4$. For example, an arbitrary weight $\omega \in \h_k^*$ will not be a real homotopy invariant of $\mc{L}_c^k S^4$. The values $\omega(h_i)$, $ i = 0, 1, \dots, k$, that is to say, the logarithmic radii of the target sphere $S^4$ and the source circles $S^1$, may be thought of as the coordinates of the weight $\omega$ in the space of all weights.

\medskip 
The Weyl group $W(E_k)$, being the group of symmetries of the $E_k$ data $(\PPP(\h_k), (-,-), K_k^*)$, acts on the vector space $\h_k^*$. 
Thus, 
it makes sense to identify the weights $\omega$ brought together by this action. 
We call the corresponding quotient orbifold the \emph{moduli space $\mc{M}_k$ of $k$-fold cyclifications of $S^4$}:
\vspace{-1mm}
\[
\mc{M}_k = \mc{M}_k (\mc{L}^k S^4)
:= [\h_k^* / W(E_k)].
\]

\vspace{-1mm}
\noindent This is a stacky quotient $[\h_k^* / W(E_k)]$, which is different from the naive, topological quotient
$\h_k^* / W(E_k)$. Another reincarnation of the quotient orbifold $[\h_k^* / W(E_k)]$ is the familiar homotopy 
quotient $\h_k^* \dslash W(E_k)$, which may be realized via the Borel construction, cf.\ \eqref{hquotient}, but 
in this context, the orbifold viewpoint would be more common. The $(k+1)$-dimensional topological quotient 
$\h_k^* / W(E_k)$ contains the $k$-dimensional quotient $K_k^\perp/ W(E_k)$, which for $k \le 8$ may be 
identified with the closure % $\overline{C}$
of a \emph{Weyl chamber}
% $C$
for the Weyl group action in the Euclidean space $K_k^\perp = \{\omega \in \h_k^* \; | \; \omega (K_k) = 0\} 
\subset \h_k^*$, cf.\ \cite[Prop. 8.29]{Hall15}.

%\newpage 

\begin{remark}[Further interpretation]
\label{Interpretation}
We highlight the following:
\vspace{-1mm} 
\item {\bf (i)} Identifying the weights $\omega$ under a permutation of the radii of the circles entering
$\mc{L}_c^k S^4$ is similar to identifying punctured Riemann surfaces under a permutation of the punctures,
if we wish to consider the moduli space of Riemann surfaces with unlabeled punctures.
% the action of automorphisms of $(\h_k^\ZZ, (-,-), K_k)$, \emph{i.e}., the Weyl group.
The identification of weights under the Weyl group action is also analogous to identifying the Planck scale 
and the sequence of radii of the circle components in compactified M-theory under U-duality. 
In the del Pezzo story, one also identifies generalized K\"{a}hler classes on a del Pezzo surface under
the action of the Weyl group, see \cite[\S 3.1]{INV}, and the fundamental Weyl chamber in the Picard 
group of a del Pezzo surface plays a prominent role in studying Cremona isometries \cite[\S 8.2.8]{Dolgachev}.

\vspace{-1mm}
\item {\bf (iii)} 
 Ultimately, physical configurations have to be U-duality invariant, so 
it makes sense to mod out by that symmetry. 
The Weyl group $W= W(E_k)$ is traditionally taken as a subgroup 
of the discrete $\mathbb{Z}$-form of the U-duality group $E_k(\RR)$;
see the discussion before \eqref{abelianization}.
% at the end of \cref{dP}. 
And since $W$ already contains a substantial part of that symmetry, 
modding out by $W$ makes sense and is also `close' to the 
ultimate moduli space, in the sense of \eqref{abelianization}.
We will consider this in more detail in \cite{SV:M-theory}.	
\end{remark}

% \newpage 

%%%%%%%%%%%%%%%%%%%%%
\subsection{The $E_k$ symmetry of type IIB}
%%%%%%%%%%%%%%%%%%%%%

Given that the real homotopy type $IIB$ was worked out so as to match with type $IIA$ via T-duality \eqref{T-duality} and \eqref{T-duality-sub},
we would like to choose a compatible trivialization $T^B \cong \GG_m \times \GG_m$ of the maximal $\RR$-split 
torus of $\Aut M(IIB)$ and a compatible basis of weights of the action of $T^B$ on $M(IIB)$.

\medskip 
We have been identifying the weights corresponding to the maximal $\RR$-split tori actions on $M(\mc{L}_c^k S^4)$ in \Cref{split-rank} starting from $k = 0$ and moving up to higher $k$ with the use of \Cref{step}, which 
related the toroidal symmetries of $M(\mc{L}_c^{k+1} S^4)$ to those of $M(\mc{L}_c^k S^4)$. The real homotopy type 
$IIB$ is connected to this sequence by a single cyclification: $M(\mc{L}_c IIB)$ is closely related to
$M(\mc{L}_c IIA) = M(\mc{L}_c^2 S^4)$, expressing T-duality; see \Cref{T-duality}. Thus, if $T^B_c$ is the maximal $\RR$-split torus of $\Aut M(\mc{L}_c IIB)$, 
then it maps naturally
to the maximal split torus $T^B$ of $\Aut M(IIB)$, because $\RR[w]$, where $w$ is as in the extension
\eqref{ext:cyc}, is $T^B_c$-invariant. Since the differential ideal $(h_7, \omega_7, dh_7, d\omega_7)$ 
of $M(\mc{L}_c IIB)$ is also $T^B_c$-invariant, $T^B_c$ acts on the quotient
$M(\mc{L}_c IIB)/(h_7, \omega_7, dh_7, d\omega_7)$ in \eqref{T-duality}. 
The maximal $\RR$-split torus $T^A_c$ of $\Aut M(\mc{L}_c IIA)$ also keeps the differential ideal $(h_7, d h_7)$ 
invariant and maps naturally to $T^B_c$. This way, we have a map 
$$
T^A_c \longrightarrow T^B_c \longrightarrow T^B.
$$ 
One can reverse these maps and get a natural map $T^B \to T^A_c$. This also implies that the weights of $T^A_c$
pull back to the weights of $T^B$.

\medskip 
Now, let us use these maps to create distinguished bases of the Lie algebra $\h_B$ of the maximal $\RR$-split 
torus $T^B$ and the dual space vector space $\h_B^*$ of weights.
We will use the weights of the elements $h_3$ and $\omega_3$ as fundamental 
(corresponding to spacetime fields $H_3$ and $F_3$ associated with the fundamental 
string and its S-dual, the D1-brane). With that, we write
\begin{equation*}
%\label{th_3}
t \cdot h_3 = t^{\gamma_{\, 1}} h_3, \qquad t \cdot \omega_3 = t^{\gamma_{\, 2}} \omega_3,
\end{equation*}
where
\[
\gamma_{\, 1} := \epsilon_0-\epsilon_1, \quad \text{and} \quad \gamma_{\, 2} := \epsilon_0 -\epsilon_2,
\]
given that
$$
h_3  \mapsto s_1 g_4, \quad \text{and} \quad  
\omega_3  \mapsto s_2 g_4
$$
in the correspondence between $M(\mc{L}_c IIB)$ and $M(\mc{L}_c IIA) = M(\mc{L}_c^2 S^4)$. Here $\epsilon_0$, $\epsilon_1$, and $\epsilon_2$ are the generating weights of $\h_2^*$, the Lie algebra of the maximal $\RR$-split torus $T^A_c$ of $\Aut M(\mc{L}_c IIA) = \Aut M(\mc{L}_c^2 S^4)$. We view these weights as being pulled back to $T^B$ (and $\h_B$) via the homomorphism $T^B \to T^A_c$. Since $\gamma_{\, 1} = \beta_0$ and $\gamma_{\, 2} = \beta_0 + \beta_1$ by Equations \eqref{tIIB}, $\gamma_{\, 1}$ and $\gamma_{\, 2}$ form a basis of $\h_B^*$ and, moreover, of the weight lattice $\PPP(\h_B) := \ZZ \gamma_1 \oplus \ZZ \gamma_2 \subseteq \h_B^*$.

\medskip 
We will use the inner product induced on $\h_{B}^*$ from $\h_2^*$:
$$
(\gamma_{\, 1}, \gamma_{\, 1})  := (\epsilon_0 - \epsilon_1, \epsilon_0 - \epsilon_1) = 0,\quad 
(\gamma_{\, 2}, \gamma_{\, 2})  := (\epsilon_0 - \epsilon_2, \epsilon_0 - \epsilon_2) = 0, \quad 
(\gamma_{\, 1}, \gamma_{\, 2})  :=(\epsilon_0 - \epsilon_1, \epsilon_0 - \epsilon_2) = 1\;.
$$
The dual basis of $\h_{B}$ will be given as $\{\mathfrak{l}_1 := h_0 -h_1, \mathfrak{l}_2 := h_0 - h_2\}$, where we used the images of the generators $h_0, h_1, h_2$ of $\h_2$ under the linear map $\h_2 \to \h_B$ linearizing the above group homomorphism $T^A_c \to T^B$. Then one can check from Equations \eqref{tIIB} that there exists a unique element $-K_B \in \h_B$ which acts as the degree operator \eqref{degree-operator} on the Quillen model $Q(IIB)$, namely, the element
\vspace{-2mm} 
\[
-K_B = 2 \mathfrak{l}_1 + 2 \mathfrak{l}_2. 
\]
Indeed, $2(h_0-h_1) + 2 (h_0-h_2)$ acts on $s_1 g_4$ in $M(\mc{L}_c IIA)$ by scaling by $2(1-1) + 2(1-0) = 2$ and the same way on $s_2 g_4$. This implies that $-K_B$ acts on $h_3$ and $\omega_3$ in $M(IIB)$ by a factor of 2, which is the degree of the dual elements in the Quillen model $Q(IIB)$. The element $2(h_0-h_1) + 2 (h_0-h_2)$ acts on $s_2 w_1$ in $M(\mc{L}_c IIA)$ by zero, and this implies that $-K_B$ acts on $\omega_1 \in M(IIB)$ also trivially, just as it is supposed to act on an degree-zero element of $Q(IIB)$. Equations \eqref{htypeIIB} and the fact that $\h_B$ acts on $M(IIB)$ by derivations then imply that $-K_B \in \h_B$ acts on the remaining generators $\omega_5$, $h_7$, and $\omega_7$ by 4, 6, and 6, respectively, again compatible with acting as the degree operator on the Quillen model $Q(IIB)$. 

\medskip 
As concerns uniqueness of an element realizing the degree operator, an arbitrary element $a_1 \mathfrak{l}_1 + a_2 \mathfrak{l}_2 \in \h_B$ acts on $\omega_1$ by a factor of $a_1 - a_2$ and on $h_3$ by $a_2$. Thus, if it acts by the Quillen-model degree, we must have $a_1 - a_2 = 0$ and $a_2 = 2$, whence $a_1 = a_2 = 2$.

%\newpage 
\medskip 
This way, as in \Cref{rootdata}, we create a root system $R_B$, which we might denote by $E_B$, corresponding  to the real homotopy type $IIB$. In contrast with the type $IIA$ root system, which is empty
(in the notation of \Cref{table1}):
\[
E_1 = R_1 = \big\{ \alpha \in \PPP(\h_1) \; | \; (\alpha, 3\epsilon_0 -\epsilon_1) = 0 , \; (\alpha, \alpha) = -2 \big\} = \varnothing = A_0\,,
\]
we have:

\begin{prop}[Exceptional root data from the rational model for type IIB]
\label{Root-IIB}
The data, as in Theorem \ref{rootdata}, associated to the Sullivan minimal model $M(IIB)$
of type IIB  replicates the data  determined by the del Pezzo surface 
$\BB'_1 = \CC \PP^1 \times \CC \PP^1$
and produces the root system
\[
E_B = R_B := \big\{ \alpha \in \PPP(\h_B) \; | \; (\alpha, 2\gamma_1 + 2 \gamma_2) = 0 , \;  
(x, x) = -2 \big\} = \big\{\gamma_1 - \gamma_2, \; \gamma_2 - \gamma_1\big\}
\]
which is nonempty and may be identified as a root system of type $A_1$. 
\end{prop}

This justifies the type IIB row of \Cref{table1}.

%%%%%%%%%%%%%%%%%%%%%%%%%%%%
\subsection{27 ``Lines'' in the cyclic loop space $\mc{L}_c^6 S^4$}
\label{Sec-27}
%%%%%%%%%%%%%%%%%%%%%%%%%%%%%%

In this section we show that our discovery of at least the toroidal part of $E_k$ symmetry of the rational homotopy type of $\mc{L}_c^k S^4$ may lead to surprising consequences, such as the existence of 27 ``lines'' in $\mc{L}_c^6 S^4$: shortly, $\dim \pi_2^\RR (\mc{L}_c^6 S^4) = 27$. Since 27 is also remarkable as the dimension of a fundamental representation of the Lie algebra $\mathfrak{e}_6$, this is suggestive of the possibility of extending the rational homotopy symmetries of $\mc{L}_c^k S^4$ outside of the toroidal part of $E_k$, if not to the whole Lie algebra $\mathfrak{e}_k$.

Given that the data $\left(\h_6^*, \{\epsilon_0, \dots, \epsilon_6\}, (-,-), K_6^*\right)$ arising in \Cref{rootdata} is exactly the same as that for the del Pezzo surface $\BB_6$, we can identify the 27 ``lines'' in $\mc{L}_c^6 S^4$. These lines are generated by the $\RR$-homotopy classes of 27 maps $\CC\PP^1 \to \mc{L}_c^6 S^4$ (to be more precise, linear combinations of such, i.e., 27 elements of $\pi_2^\RR (\mc{L}_c^6 S^4) = \pi_2 (\mc{L}_c^6 S^4) \otimes \RR$) supplied by the following result.

\begin{theorem}
[27 lines via rational homotopy of 6-fold cyclic loop space]
\label{27lines}
The 27 exceptional vectors $\alpha \in \PPP(\h_6)$, $(\alpha, \alpha) = (\alpha, K_6^*) = -1$, give rise to 27 canonically defined lines in the $\RR$-vector space $\pi_2^\RR (\mc{L}_c^6 S^4)$. Moreover, these lines freely generate $\pi_2^\RR (\mc{L}_c^6 S^4)$ and thus $\dim \pi_2^\RR (\mc{L}_c^6 S^4) = 27$.
\end{theorem}

\begin{proof}
The 27 exceptional vectors for the data 
$\left(\h_6^*,  \{\epsilon_0, \dots, \epsilon_6\}, (-,-), K_6^*\right)$ associated with the space $\mc{L}_c^6 S^4$ by \Cref{rootdata} are the elements $\alpha$ of the weight lattice $\PPP(\h_6) = \ZZ \epsilon_0 \oplus \dots \oplus \ZZ \epsilon_6$ which pair to 1 with $-K_6^*$ or, equivalently, evaluate to 1 at $-K_6 \in \h_6$: $\alpha(-K_6) = 1$. By the definition of the ``anticanonical class'' $-K_6$, for any weight $\alpha \in \PPP(\h_6)$, an element $x$ of the weight space $Q_\alpha$ of the Quillen minimal model $Q = \pi_\bullet^\RR (\mc{L}_c^6 S^4)[1]$ has degree $\abs{x} = \alpha(-K_6)$:
\vspace{-3mm} 
\begin{align*}
h x & = \alpha (h) \cdot x  && \text{for any $h \in \h_6$},\\
-K_6 x & = \abs{x} \cdot x  && \text{for  $h = -K_6$}.
\end{align*}
Thus, the weights $\alpha \in \PPP(\h_6)$ which evaluate to 1 at $-K_6$
are precisely the weights of the degree-one component of $\pi_\bullet^\RR (\mc{L}_c^6 S^4)[1]$,
\[
\big(\pi_\bullet^\RR (\mc{L}_c^6 S^4)[1]\big)_1 = \pi_2^\RR (\mc{L}_c^6 S^4).
\]
The vector space $\pi_2^\RR (\mc{L}_c^6 S^4)$ is linear dual to the degree-two component $V^2$ of the generator vector space $V = Q^*[1] = \pi_\bullet^\RR (\mc{L}_c^6 S^4)^*$ of the Sullivan minimal model $M (\mc{L}_c^6 S^4) = (S(V), d)$. Note that in the weight decomposition 
\[
V^2 = \bigoplus_{\alpha \in \PPP(\h_6)} V^2_\alpha \; ,
\]

\vspace{-2mm} 
\noindent the weights that actually occur, i.e., $V^2_\alpha \ne 0$, are exactly on the following list:
\vspace{-2mm} 
\begin{align}
\nonumber
    &\epsilon_1, \dots, \epsilon_6,\\
\label{27weights}
&\epsilon_0 - \epsilon_i - \epsilon_j, \hspace{2.92cm}  1 \le i < j \le 6,\\
\nonumber
   & 2 \epsilon_0 - \epsilon_{1} - \dots - \widehat{\epsilon_i} - \dots - \epsilon_6, \quad 1 \le i \le 6.
\end{align}
Moreover, the corresponding weight spaces are all one-dimensional and generated by the elements
\vspace{-2mm} 
\begin{align}
\nonumber
   & w_1, \dots, w_6,\\
 \label{27fields}
 &    s_j s_i g_4, \hspace{1.7cm} 1 \le i < j \le 6,\\
\nonumber
  &  s_6 \dots \widehat{s_i} \dots s_1 g_7, \quad 1 \le i \le 6,
\end{align}
respectively, just as in the table which presented the 27 spacetime fields earlier in this section. The corresponding one-dimensional linear-dual subspaces $(V^2_\alpha)^* = (Q_1)_\alpha \subset \pi_2^\RR (\mc{L}_c^6 S^4)$ are canonically defined as weight spaces, and we have
\vspace{-1mm} 
\begin{align*}
V^2 & = \bigoplus_{\scalebox{0.55}{${\alpha  \text{ on the list  \eqref{27weights}}}$}} V^2_\alpha \; ,\\
\pi_2^\RR (\mc{L}_c^6 S^4) & = \bigoplus_{\scalebox{0.55}{${\alpha  \text{ on the list  \eqref{27weights}}}$}} \pi_2^\RR (\mc{L}_c^6 S^4)_\alpha \; .
\end{align*}
These are the 27 lines in $\pi_2^\RR (\mc{L}_c^6 S^4)$. For each line, a generator, defined up to a nonzero real factor, is represented by a real homotopy class in $\pi_2^\RR (\mc{L}_c^6 S^4)$, a ``line $\CC \PP^1 = S^2 \to \mc{L}_c^6 S^4$''.
\end{proof}

\begin{remark}[Other cases]
A similar count of ``lines in $\mc{L}_c^k S^4$'' works for all $k$, $0 \le k \le 6$, with the same numbers as those for the exceptional vectors $I_k \subset N_k$ and exceptional curves on del Pezzo surfaces $\BB_k$, cf.\ \cite[Theorem 4.3]{Man}. This count starts to break for $k > 6$, because some of the exceptional vectors will start having trivial weight spaces. In particular, for $k = 7$, instead of 28 pairs of lines in $\BB_7$, we will have 28 ``lines,'' 21 of which are paired to other 21 ``lines'' in $\mc{L}_c^7 S^4$, with 7 ``lines'' missing a pair. We will address this in an upcoming paper \cite{SV:M-theory}.
% This is reflected in the irregularity of the form of spacetime fields in the last two tables of \Cref{Sec-excvec}.
\end{remark}

%%%%%%%%%%%%%%%%%%
\subsection{%Blowing up $\ge 9$ points, c
Cyclifying $\ge 9$ times and Kac-Moody algebras}
\label{Sec-KM}
%%%%%%%%%%%%%%%%%%%%

Nothing prevents us from cyclifying the 4-sphere 9 times and more, just like blowing up $\CC \PP^2$ at $k \ge 9$ points makes perfect sense --- the resulting surfaces are just no longer del Pezzo \cite{Dol1}\cite{Do-Coble}. In our case, it is interesting to see what ``phase transition'' is happening between $k=8$ and 9. It is also reasonable to expect the emergence of Kac-Moody algebras of type $E_k$ for $k \ge 9$ in relation to higher cyclifications $\mc{L}_c^k S^4$. These Lie algebras play a role in further blowups of $\CC \PP^2$, as well as in reductions of 11d supergravity to dimensions 2, 1, and 0 \cite{Julia86}\cite{N92}. We plan to address the relation to algebraic geometry and physics in \cite{SV:M-theory}.

\paragraph{The Lie algebras $\mathfrak{e}_k$ of type $E_k$ for $k \ge 9$.} 
For $k \ge 9$, the Dynkin diagram
%\vspace{-1mm} 
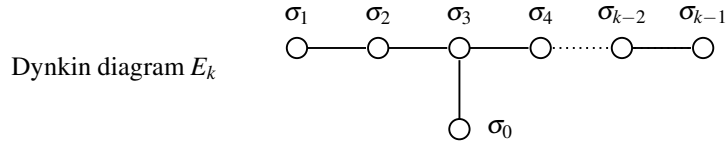
\[
% \label{Dynkin}
\hspace{-1cm} 
\mbox{\small Dynkin diagram $E_k$}
\qquad
  \scalebox{.9}{$
  \raisebox{-30pt}{\begin{tikzpicture}[scale=.6]
 %   \draw (-1,1) node[anchor=east] 
 %   {
    % $T_{2,3,k-3}$
%    Dynkin diagram $E_k$};
    \foreach \x in {0,...,5}
    \draw[thick,xshift=\x cm] (\x cm,0) circle (2.5 mm);
 %   \foreach \y in {0,...,5}
    \foreach \y in {0, 1,2, 4}
    \draw[thick,xshift=\y cm] (\y cm,0) ++(.25 cm, 0) -- +(14.5 mm,0);
    \foreach \y in {3,4}
    \draw[dotted, thick,xshift=\y cm] (\y cm,0) ++(.3 cm, 0) -- +(14 mm,0);
    \draw[thick] (4 cm, -2 cm) circle (2.5 mm);
    \draw[thick] (4 cm, -3mm) -- +(0, -1.45 cm);
    \node at (0,.8) {$\sigma_1$};
    \node at (2,.8) {$\sigma_2$};
    \node at (4,.8) {$\sigma_3$};
    \node at (6,.8) {$\sigma_4$};
    \node at (8,.8) {$\sigma_{k-2}$};
    \node at (10,.8) {$\sigma_{k-1}$};
    \node at (5,-2) {$\sigma_0$};
  \end{tikzpicture}
  }
  $}
\]
corresponds to Kac-Moody algebras:
the affine Lie
% Weyl 
algebra % $\mathfrak{e}_9$
of type $E_9 = \widehat{E}_8 = E_8^{(1)}$, the hyperbolic Kac-Moody algebra of type $E_{10}$, and the Lorentzian Kac-Moody algebras of type $E_{k}$ for $k \ge 11$.
% The corresponding lattices arise

\paragraph{The $k$-fold cyclifications $\mc{L}_c^k S^4$ and root lattices for $k \ge 9$.} The $E_k$ root system data arises from $\mc{L}_c^k S^4$ according to \Cref{rootdata}: the maximal split real torus of $\Aut \mc{L}_c^k S^4$ is $(k+1)$-dimensional; the dual space $\h_{k}^*$ of its Lie algebra has a natural basis $\epsilon_0, \epsilon_1, \dots, \epsilon_k$, Lorentzian inner product, and distinguished element $K_k^*$. The sublattice $(K_k^*)^\perp \subset \PPP(\h_k)$ is a root lattice of type $E_k$, see \cref{table2}.
% = \widehat{E}_8$.

\vspace{-3mm} 
\paragraph{Exceptional vectors for $\mc{L}_c^k S^4$ for $k \ge 9$.}
In contrast to the $k \le 8$ case, for $k \ge 9$ the set
\[
I_k=\left\{ \alpha \in \PPP(\h_k) \; \vert \; (\alpha, K_k^*) =
(\alpha, \alpha) = -1 \right\}
\]
of exceptional vectors is infinite for $\mc{L}_c^k S^4$. Indeed, the Weyl group $W_k$ acts on the dual of the Lie algebra $\h_k$ of the maximal $\RR$-split torus $T^{k+1}$ of $\Aut M(\mc{L}_c^k S^4)$ by isometries preserving $K_k^* = -3 \epsilon_0 + \epsilon_1 + \dots + \epsilon_k$ and the lattice $\PPP(\h_k)$. In particular, the Weyl group acts on the exceptional vectors. Examples, such as in \eqref{27weights}, show there are enough exceptional vectors to generate the whole vector space $\h_k^*$. This implies that an element of $W_k$ is determined by its action on the exceptional vectors.
Then the fact that the group $W_k$ is infinite for $k \ge 9$ implies that there are infinitely many exceptional vectors.

\medskip
Recall from \Cref{27lines} that in the $k=6$ case, the 27 exceptional weight spaces of the action of $\h_k$ on the Quillen model $Q(\mc{L}_c^6 S^4) = \pi_\bullet^\RR (\mc{L}_c^6 S^4)[1]$, i.e., on $\pi_2^\RR (\mc{L}_c^6 S^4)$ are all one-dimensional and present 27 lines in $\mc{L}_c^6 S^4$.

\medskip 
On the contrary, for $k \ge 9$, only finitely many exceptional weights will be populated, both in $Q(\mc{L}_c^k S^4)$ and $M(\mc{L}_c^k S^4)$. This is because the underlying real vector space $\pi_\bullet^\RR (\mc{L}_c^k S^4)$ of the Quillen model and the subspaces of bounded degree of the Sullivan model of $M(\mc{L}_c^k S^4)$ are finite-dimensional for all $k \ge 0$, as follows from the identification of $M(\mc{L}_c^k S^4)$ in \Cref{Sullmm}.

The ``phase transition'' from $k \le 8$ to $k \ge 9$ is summarized in the following table.
%\vspace{-1mm} 
\begin{center} 
\setlength{\tabcolsep}{7pt} % Default value: 6pt
\renewcommand{\arraystretch}{1.2} % Default value: 1
\begin{tabular}{cccccc}
\hline 
\rowcolor{lightgray} \textbf{Case} & {\bf $\deg \mc{L}_c^k S^4$} & {\bf Inner product on $K_k^\perp$} & {\bf $K_k$ and $K_k^\perp$} & {\bf Weyl group $W_k$} & {\bf Exceptnl vectors for $\mc{L}_c^k S^4$}
\\
\hline 
\hline 
 $k \le 8$ & $>0$ & negative-definite & $K_k \not \in K_k^\perp$ & finite & finitely many
\\
\rowcolor{lightgray} $k=9$ & 0  & negative semi-definite & $K_k \in K_k^\perp$ & infinite & infinitely many\\
 $k \ge 10$ & $<0$ & Lorentzian & $K_k \not \in K_k^\perp$ & infinite & infinitely many\\
\hline 
\end{tabular} 
 \end{center}

%%%%%%%%%%%%%%%%%

\noindent  Hisham Sati, {\it Mathematics, Division of Science, and 
\\
\indent Center for Quantum and Topological Systems (CQTS),  NYUAD Research Institute, 
\\
\indent New York University Abu Dhabi, UAE.}
\\
{\tt hsati@nyu.edu}
\\
\\
\noindent  Alexander A. Voronov, {\it School of Mathematics, University of Minnesota, Minneapolis, MN 55455, USA, and
\\
\indent Kavli IPMU (WPI), UTIAS, University of Tokyo, Kashiwa, Chiba 277-8583, Japan.}
\\
{\tt voronov@umn.edu}

\end{document}